\documentclass[pdfa,a4paper,UKenglish,cleveref, autoref, thm-restate]{lipics-v2021}

\pdfoutput=1 
\hideLIPIcs  

\title{Monads and Distributive Laws in Substructural Contexts (Extended Version)} 

\author{Soichiro Fujii}{National Institute of Informatics, Tokyo, Japan}{s.fujii.math@gmail.com}{https://orcid.org/0000-0002-4109-3472}{}

\author{{Yun Chen} Tsai}{National Institute of Informatics, Tokyo, Japan \and SOKENDAI (The Graduate University for Advanced Studies), Kanagawa, Japan}{yctsai@nii.ac.jp}{https://orcid.org/0009-0003-7705-9609}{}

\author{Yo\`av Montacute}{National Institute of Informatics, Tokyo, Japan}{montacute@nii.ac.jp}{https://orcid.org/0000-0001-9814-7323}{Supported by ACT-X, Grant No.\ JPMJAX24CR, JST, Japan.}

\author{Ichiro Hasuo}{National Institute of Informatics, Tokyo, Japan \and SOKENDAI (The Graduate University for Advanced Studies), Kanagawa, Japan \and Imiron Co., Ltd., Tokyo, Japan}{i.hasuo@acm.org}{https://orcid.org/0000-0002-8300-4650}{}

\authorrunning{S. Fujii, Y.\,C. Tsai, Y. Montacute, and I. Hasuo} 

\Copyright{Soichiro Fujii, Yun Chen Tsai, Yo\`av Montacute, and Ichiro Hasuo} 

\ccsdesc{Theory of computation~Categorical semantics}
\ccsdesc{Theory of computation~Probabilistic computation}
\ccsdesc{Theory of computation~Program semantics}

\keywords{Monad, distributive law,  operad, category theory, effect} 

\category{} 

\relatedversion{This paper is the extended version of \cite{FTMH}.} 

\funding{All authors are supported by ASPIRE Grant No.\ JPMJAP2301, JST, Japan.}

\acknowledgements{We thank the anonymous referees of \cite{FTMH} for their helpful comments.}

\nolinenumbers 

\EventEditors{Claudia Faggian and Joost-Pieter Katoen}
\EventNoEds{2}
\EventLongTitle{41st Annual Symposium on Logic in Computer Science (LICS 2026)}
\EventShortTitle{LICS 2026}
\EventAcronym{LICS}
\EventYear{2026}
\EventDate{July 20--23, 2026}
\EventLocation{Lisbon, Portugal}
\EventLogo{}
\SeriesVolume{380}
\ArticleNo{76}

\usepackage{tikz}
\usetikzlibrary{arrows,arrows.meta,decorations.pathreplacing,decorations.markings,shapes,calc,cd,nfold}
\tikzset{labelsize/.style={font=\scriptsize}}
\usepackage{mathtools}
\usepackage{proof}
\usepackage[export]{adjustbox}
\newtheorem{notation}[theorem]{Notation}

\crefname{theorem}{Thm.}{Thms}
\crefname{section}{\S}{\S\S}
\crefname{definition}{Def.}{Defs}
\crefname{notation}{Notation}{Notations}
\crefname{example}{Ex.}{Exs}
\crefname{remark}{Rem.}{Rems}
\crefname{problem}{Problem}{Problems}
\crefname{equation}{}{}
\crefname{construction}{Construction}{Constructions}
\crefname{corollary}{Cor.}{Cors}
\crefname{proposition}{Prop.}{Props}
\crefname{lemma}{Lem.}{Lems}
\crefname{appendix}{Appendix}{Appendices}
\crefname{claim}{Claim}{Claims}
\crefname{figure}{Fig.}{Figs}

\newcommand{\cat}[1]{{\mathcal #1}} 
\newcommand{\mnd}[1]{{\mathbb #1}} 
\newcommand{\mon}[1]{{\mathsf #1}} 
\newcommand{\sr}[1]{{\mathsf #1}} 
\newcommand{\Set}{\mathbf{Set}}
\newcommand{\Monzero}{\mathbf{Mon}_0}
\newcommand{\Mon}{\mathbf{Mon}}
\newcommand{\SRing}{\mathbf{SRing}}

\DeclareMathOperator{\supp}{\mathrm{supp}}
\newcommand{\NN}{\mathbb{N}}

\newcommand{\RRnonneg}{\mathbb{R}_{\geq 0}}
\newcommand{\PPfin}{\mnd P_{\mathrm{fin}}}
\newcommand{\Pfin}{P_{\mathrm{fin}}}
\newcommand{\PPfinne}{\mnd P_{\mathrm{fin}}^{+}}
\newcommand{\Pfinne}{P_{\mathrm{fin}}^{+}}
\newcommand{\mult}{\mathrm{mult}}
\newcommand{\multzero}{\mathrm{mult,0}}
\newcommand{\Finj}{\mathbf{F}_{\mathrm{inj}}}
\newcommand{\Fid}{\mathbf{F}_{\mathrm{id}}}
\newcommand{\F}{\mathbf{F}}
\newcommand{\Fsurj}{\mathbf{F}_{\mathrm{surj}}}
\newcommand{\Fbij}{\mathbf{F}_{\mathrm{bij}}}
\newcommand{\Fmonoinj}{\mathbf{F}_{\mathrm{m,inj}}}

\newcommand{\EEnd}{\mathbb{E}\mathbf{nd}}
\newcommand{\End}{\mathrm{End}}

\newcommand{\Lan}{\mathrm{Lan}}
\newcommand{\nameof}[1]{#1}
\newcommand{\vN}[1]{#1} 
\newcommand{\aff}{\mathrm{aff}}
\newcommand{\Ctxt}{\mathbf{W}}
\newcommand{\W}{\mathbf{W}}

\renewcommand{\phi}{\varphi}
\renewcommand{\epsilon}{\varepsilon}

\newcommand{\id}{\mathrm{id}}
\newcommand{\ido}{\mathsf{id}} 
\newcommand{\subst}{\mathsf{subst}}
\newcommand{\str}{\tau}
\newcommand{\comm}{\psi}
\newcommand{\iso}{\mathrm{iso}}
\mathchardef\mhyphen="2D

\newcommand{\Mnd}[1]{\mathrm{Mnd}(#1)}
\newcommand{\Opd}[2]{{#1}\mhyphen\mathrm{Opd}(#2)}
\newcommand{\Rf}[2]{\mathrm{Rf}_{#1}(#2)}
\newcommand{\Mndcat}[1]{\mathbf{Mnd}(#1)}
\newcommand{\Opdcat}[1]{{#1}\mhyphen\mathbf{Opd}}

\newcommand{\op}{\mathrm{op}}

\begin{document}

\maketitle

\begin{abstract}
We present a categorical theory of monads and distributive laws \emph{in substructural contexts}. In the study of distributive laws, the roles of (the absence of) structural rules for variable contexts have been recognized; our theory formalizes these substructural situations using Tronin's \emph{verbal categories} $\mathbf W$, in a uniform and presentation-independent manner. We introduce the classes of \emph{$\mathbf W$-operadic monads} (those \emph{defined} via the structural rules in $\mathbf W$) and of \emph{$\mathbf W$-commutative monads} (those \emph{invariant} under the structural rules in $\mathbf W$). We give a canonical construction of a distributive law $ST\to TS$ of monads on $\mathbf{Set}$; it is applicable when $S$ is $\mathbf W$-operadic and $T$ is $\mathbf W$-commutative (under mild conditions). This accounts for many known and new distributive laws. Even when $S$ fails to be $\mathbf W$-operadic, we can \emph{refine} $S$ and force $\mathbf W$-operadicity; this captures Varacca and Winskel's construction of indexed valuations.
\end{abstract}

\section{Introduction}
\label{sec:intro}

\subparagraph*{Monads}
In program semantics and, more generally, in system behavior modeling,  the categorical notion of \emph{monad} has been successfully used as a uniform interface to various computational \emph{effects}. Examples of effects include non-functional \emph{side effects} in functional programming~\cite{Moggi91a}, additional behavioral layers for reactive systems such as (possibilistic or probabilistic) nondeterminism~\cite{Bartels04,Jacobs04c}, and so on. The mathematical theory of monads---together with a closely related theory of \emph{universal algebra}---not only offers unified foundations for effects, but also enables novel programming principles such as \emph{effect handlers}~\cite{PlotkinP13}. 

\newpage

\subparagraph*{Distributive laws}
In this context, \emph{distributive laws} between monads have attracted attention. Originally a mathematical construct~\cite{Beck-distributive-laws}, distributive laws enable us to \emph{combine effects}. Specifically, given monads $S$ and $T$ on the category $\Set$ of sets (in this paper, we focus on monads on $\Set$), a distributive law $\delta\colon ST\to TS$ of $S$ over $T$ equips the composite $TS$ with a monad structure, with the consequence that  unified monadic frameworks for programming, static analysis, verification, etc.\ accommodate the coexistence of two different layers of effects modeled by $S$ and $T$.

The question of obtaining  a distributive law has intrigued many theoreticians, too. This was initiated by Plotkin's observation that there is no distributive law $DP\to PD$ of the (finitely supported) probability distribution monad $D$ over the powerset monad $P$ (see \cite{Varacca-Winskel-dist-law}). 
This observation sparked much work: both on the positive side (obtaining distributive laws~\cite{Dahlqvist-Parlant-Silva-layer-by-layer,Parlant-thesis,GoyP20}) and the negative side (proving ``no-go'' theorems~\cite{KlinS18,ZwartM22}).

\subparagraph*{Two approaches: weak distributive laws and refinement}
On the positive side, towards obtaining a distributive law, there have mainly been two technical approaches. The first ``weak distributive law'' approach,
when a distributive law $ST\to TS$ does not exist,  finds it sufficient to have a \emph{weak distributive law} $\delta\colon ST\to TS$. Here ``weak'' means $\delta$'s compatibility with the unit $\eta^{S}$ of $S$ is dropped. A weak distributive law $\delta$ induces a combination of $S$ and $T$, but this time it is not given simply by the composite $TS$. Roughly speaking, this approach modifies $T$: concretely, for  $S=D$ and $T=P$, the combination essentially replaces $T=P$ with the convex powerset monad over the Eilenberg--Moore category $\Set^D$ \cite{Garner-Vietoris-monad,GoyP20,BonchiS22,Aristote25,TixKP05}.

The second ``refinement'' approach  modifies $S$ instead. It finds a reason for the failure of a distributive law $\delta\colon ST\to TS$ in \emph{equations} in (an equational presentation of) $S$, and thus proposes to use a \emph{refinement} $\widehat{S}$ of $S$---obtained by dropping the problematic equations---in place of $S$. This approach is initiated 
 in~\cite{Varacca-Winskel-dist-law} for specific monads (variations of $D$ and $P$---see \cref{subsec:IV-revisited}).
Generalization of this approach is sought in~\cite{Dahlqvist-Parlant-Silva-layer-by-layer,Parlant-thesis}. 
We follow this line of work.

\subparagraph*{Categorical theory of substructural contexts for monads and distributive laws}
In this paper, we introduce a general  framework for the aforementioned ``refinement'' approach to  distributive laws, centered around the idea of \emph{substructural contexts}. Our categorical framework uses \emph{verbal categories}~\cite{Tronin-abstract-clones-and-operads} for modeling substructural contexts.

The relevance of structural rules (exchange, weakening, and contraction) to algebraic theories (and thus monads) is widely acknowledged, explicitly or implicitly~\cite{Curien-operads-clones,Jacobs-weakening-contraction,Leinster-higher-operads}. In particular, \cite{Dahlqvist-Parlant-Silva-layer-by-layer,Parlant-thesis,Manes-Mulry-monad-compositions-1} study the roles of  structural rules in distributive laws. Our framework here offers a unifying categorical picture parametrized by a verbal category $\W$. 

We sketch our theory in the coming paragraphs.

\subparagraph*{Substructural contexts and verbal categories (\cref{sec:notions-of-operad})}
Restricting structural rules---\emph{exchange} (E), \emph{weakening} (W), and \emph{contraction} (C)---in variable contexts is a common tool in logic. 
In the categorical study of algebraic theories, operads, and binding syntax, it is common to model the application of structural rules with 1) morphisms of $\F$, the category of finite ordinals and functions between them \cite{FiorePlotkinTuri-AbstractSyntax}, and 2) their string diagram presentation \cite{Curien-operads-clones,Leinster-higher-operads}. For example, the context transformation
from $x_0,x_1,x_2\vdash p(x_0,x_1,x_2)$ to 
$x_0,x_1,x_2,x_3\vdash p(x_1,x_3,x_1)$---obtained with E, W, and C for a term $p$---is represented by 
the function $\alpha\colon \vN{3}\to\vN{4}$
with $\alpha(0)=\alpha(2)=1$ and $\alpha(1)=3$;
the corresponding string diagram (in which we have $\alpha(i)=j$ iff $i$ on the right is connected to $j$ on the left) is depicted on the left of \eqref{eqn:string-diag} below.
\begin{equation}\label{eqn:string-diag}
\begin{tikzpicture}[baseline=-\the\dimexpr\fontdimen22\textfont2\relax ]
\node[labelsize] at (-1.5,0.3) {$0$};
\node[labelsize] at (-1.5,0.1) {$1$};
\node[labelsize] at (-1.5,-0.1) {$2$};
\node[labelsize] at (-1.5,-0.3) {$3$};
\node[labelsize] at (-0.3,0.25) {$0$};
\node[labelsize] at (-0.3,0) {$1$};
\node[labelsize] at (-0.3,-0.25) {$2$};
\draw[thick] (-0.5,0) .. controls (-0.7,0) and (-0.8,-0.3) .. (-1.3,-0.3);
\draw[thick] (-0.5,-0.25) .. controls (-0.7,-0.25) and (-0.8,0.1) .. (-1,0.1);
\draw[thick] (-0.5,0.25) .. controls (-0.7,0.25) and (-0.8,0.1) .. (-1,0.1);
\draw[thick] (-1,0.1)--(-1.3,0.1);
\draw[thick] (-0.95,0.3) -- (-1.3,0.3);
\draw[thick] (-0.95,-0.1) -- (-1.3,-0.1);
\filldraw [fill=black] (-0.95,-0.1) circle [radius=1.5pt];
\filldraw [fill=black] (-0.95,0.1) circle [radius=1.5pt];
\filldraw [fill=black] (-0.95,0.3) circle [radius=1.5pt];
\end{tikzpicture}
\qquad\qquad\qquad\qquad
    \begin{tikzpicture}[baseline=-\the\dimexpr\fontdimen22\textfont2\relax ]
\draw[thick] (0,0.1) .. controls (0.3,0.1) and (0.5,-0.1) .. (0.7,-0.1);
\draw[thick] (0,-0.1) .. controls (0.3,-0.1) and (0.5,0.1) .. (0.7,0.1);
\end{tikzpicture}\qquad\qquad
\begin{tikzpicture}[baseline=-\the\dimexpr\fontdimen22\textfont2\relax ]
\draw[thick] (-0.35,0)-- (-0.7,0);
\filldraw [fill=black] (-0.35,0) circle [radius=1.5pt];
\end{tikzpicture}
\qquad\qquad
    \begin{tikzpicture}[baseline=-\the\dimexpr\fontdimen22\textfont2\relax ]
\draw[thick] (0,0.1) .. controls (-0.2,0.1) and (-0.3,0) .. (-0.4,0);
\draw[thick] (0,-0.1) .. controls (-0.2,-0.1) and (-0.3,0) .. (-0.4,0);
\draw[thick] (-0.4,0)-- (-0.7,0);
\filldraw [fill=black] (-0.35,0) circle [radius=1.5pt];
\end{tikzpicture}
\end{equation}
The structural rules E, W, and C correspond to the three ``parts'' of the string diagram as on the right of \eqref{eqn:string-diag}, respectively (see also \cite[\S~2]{Curien-operads-clones} for a related discussion). 

Conveniently, considering a \emph{substructural context}---i.e.\ prohibiting the use of certain structural rules---corresponds to singling out a wide subcategory of $\F$. Concretely, allowing the rules designated below on the left corresponds to the subcategories of $\F$ on the right.
\begin{equation}
\label{eqn:Ctxt-logically}
\begin{tikzpicture}[x=5.5ex,y=5ex,baseline=-\the\dimexpr\fontdimen22\textfont2\relax ]
    \node(01) at (0,0.6) {$\emptyset$};
    \node(02) at (2,0.6) {$\{\mathrm E\}$};
    \node(03) at (4.5,0.6) {$\{\mathrm E,\mathrm C\}$};
    \node(11) at (0,-0.6) {$\{\mathrm W\}$};
    \node(12) at (2,-0.6) {$\{\mathrm E,\mathrm W\}$};
    \node(13) at (4.5,-0.6) {$\{\mathrm E,\mathrm W,\mathrm C\}$};
    \draw [->] (01) to (02);
    \draw [->] (02) to (03);
    \draw [->] (11) to (12);
    \draw [->] (12) to (13);
    \draw [->] (01) to (11);
    \draw [->] (02) to (12);
    \draw [->] (03) to (13);
\end{tikzpicture}
\qquad\qquad
\begin{tikzpicture}[x=5ex,y=5ex,baseline=-\the\dimexpr\fontdimen22\textfont2\relax ]
    \node(01) at (0,0.6) {$\Fid$};
    \node(02) at (2,0.6) {$\Fbij$};
    \node(03) at (4,0.6) {$\Fsurj$};
    \node(11) at (0,-0.6) {$\Fmonoinj$};
    \node(12) at (2,-0.6) {$\Finj$};
    \node(13) at (4,-0.6) {$\F$};
    \draw [->] (01) to (02);
    \draw [->] (02) to (03);
    \draw [->] (11) to (12);
    \draw [->] (12) to (13);
    \draw [->] (01) to (11);
    \draw [->] (02) to (12);
    \draw [->] (03) to (13);
\end{tikzpicture}
\end{equation}
Here, 
     $\Fid$ has only identity functions as morphisms; 
     $\Fmonoinj$ has monotone injections (with respect to the total order $0<1<\dots<n-1$); 
     $\Fbij$ has bijections;
     $\Finj$ has injections;
     $\Fsurj$ has surjections. 
(The settings $\{\mathrm C\}$ and $\{\mathrm W,\mathrm C\}$ are ill-behaved and thus omitted; see \cref{ex:verbal-cats-nonex}.)

The notion of \emph{verbal category}~\cite{Tronin-abstract-clones-and-operads} axiomatizes (wide) subcategories of $\F$ suited for modeling substructural contexts. Examples include the six categories in~\eqref{eqn:Ctxt-logically} (see also \cref{ex:verbal-cats}).

\subparagraph*{$\W$-operads and their monads (\cref{sec:notions-of-operad})}
The notion of \emph{$\W$-operad} in~\cite{Tronin-abstract-clones-and-operads} embodies the idea of ``algebraic theories in the $\W$-substructural context''. Its definition follows the same format---with $\ido$ and $\subst$---as for (symmetric or non-symmetric) operads. 

For our theory, we introduce the notion of \emph{$\W$-operadic monad}, as those which are induced by a $\W$-operad via a left Kan extension. This formalizes the idea of ``a monad in the $\W$-substructural context''. (For example, $\F$-operadic monads coincide with finitary monads.)

\subparagraph*{$\W$-commutative monads (\cref{sec:notions-of-commutative-monad})}
Our  construction of a distributive law $ST\to TS$ (\cref{thm:main-detailed}) requires, very roughly speaking, that
\begin{itemize}
 \item $S$ is a monad defined using the structural rules in $\W$, and
 \item $T$'s effect is invariant under the structural rules in $\W$.
\end{itemize}
We just formalized the former as $\W$-operadic monads; for the latter, we introduce the notion of \emph{$\W$-commutative monad}.
The definition of $\W$-commutativity is by a suitable diagram, one similar to those used in~\cite{Parlant-thesis,Dahlqvist-Parlant-Silva-layer-by-layer,Parlant-Rot-Silva-Bas}. It unifies some known notions such as
(ordinary) \emph{commutativity}~\cite{Kock-monads-on-SMCC}, \emph{affinity}~\cite{Kock-bilinearity}, and \emph{relevance}~\cite{Jacobs-weakening-contraction}.

\subparagraph*{Canonical distributive laws and application (\cref{sec:dist-law,sec:monad-to-operad,sec:applications})} 
One of our main results in this paper is \cref{thm:main-detailed} where, as announced above, we present a general construction of a distributive law $\delta\colon ST\to TS$, 
assuming that
\begin{itemize}
 \item $S$ is a $\W$-operadic monad, and
 \item 
 $T$ is a $\W$-commutative monad.
\end{itemize}
The construction is categorical and canonical, along the coend description of (the left Kan extension used in) a $\W$-operadic monad. 

The construction of this canonical $\delta$ is uniform for any verbal category $\W$. An additional condition is required for $\delta$ to be compatible with the multiplication $\mu^{T}$ of $T$: it is  that $\W$ accommodates the exchange rule (E)---unless $S$ is so lean that it is generated by operations of arity $\le 1$. See \cref{thm:main-detailed} for a precise statement.

Given the last condition (of accommodating (E)), it is natural for us to focus on the following four verbal categories.
\begin{equation}\label{eq:fourVerbalCatforIntersection}
\begin{tikzpicture}[x=5ex,y=5ex,baseline=-\the\dimexpr\fontdimen22\textfont2\relax ]
    \node(02) at (1.5,-0.75) {$\Fbij$};
    \node(03) at (3,0) {$\Fsurj$};
    \node(12) at (0,0) {$\Finj$};
    \node(13) at (1.5,0.75) {$\F$};
    \draw [->] (02) to (03);
    \draw [->] (12) to (13);
    \draw [->] (02) to (12);
    \draw [->] (03) to (13);
\end{tikzpicture}
\end{equation}
We also show the following \emph{monotonicity} results:
\begin{itemize}
 \item $\W$-operadicity is up-closed (i.e.\ if $S$ is $\W$-operadic and $\W\subseteq \W'$, then $S$ is $\W'$-operadic, \cref{extension-of-operad}(2)); and
 \item $\W$-commutativity is down-closed (\cref{rmk:commutativity-monotonicity}).
\end{itemize}
Therefore, for the application of the general construction in \cref{thm:main-detailed}, we can look for $\W$ in the intersection of the upper set of $S$'s $\W$-operadicity and the lower set of $T$'s $\W$-commutativity. 

This ``find an intersection'' strategy has two usages. The first usage is negative: given a pair $(S,T)$, the empty intersection of the above two sets shows that the canonical distributive law in \cref{thm:main-detailed} does not work. While this does not prove a no-go result (other distributive laws may exist), it is often easy to check and thus a useful red flag. Of course, the intersection is empty for the combinations with known no-go results.

The second use is a positive one: when the intersection is empty, we can modify $S$ or $T$ for a nonempty intersection, so that the canonical distributive law works. While modifying $T$ is possible---techniques for forcing $\W$-commutativity on a monad are studied for some $\W$, see e.g.~\cite{Jacobs-weakening-contraction,CaretteLZ23,Lindner-affine-part}---we are  interested in modifying $S$ in this paper. This modification of $S$ must extend the upper set of $\W$-operadicity downwards, so that a new monad relies on fewer structural rules. 

To do so, we introduce a categorical construction, namely  the \emph{$\W$-operadic refinement} $\Rf{\W}{S}$ of a monad $S$ (\cref{sec:monad-to-operad}). It is the monad induced, in the Kan extension-way described in \cref{sec:notions-of-operad}, by the $\W$-operad given by \emph{restricting} $S$. The refinement $\Rf{\W}{S}$ is $\W$-operadic by construction. We show that the \emph{indexed valuation monad} $IV$ in~\cite{Varacca-Winskel-dist-law}---the first example of the ``refinement'' approach for distributive laws---is indeed an instance of $\W$-operadic refinement.

\subparagraph*{Contributions and organization} 
We introduce a  theory of substructural contexts for monads and distributive laws. It is  categorical and uniform, unifying previous observations for specific substructural contexts~\cite{Parlant-thesis,Dahlqvist-Parlant-Silva-layer-by-layer,Parlant-Rot-Silva-Bas,Kock-monads-on-SMCC,Kock-bilinearity,Jacobs-weakening-contraction} and the axiomatization by verbal categories $\W$~\cite{Tronin-abstract-clones-and-operads}. The theory has the following main components.
\begin{itemize}
 \item After reviewing verbal categories~\cite{Tronin-abstract-clones-and-operads}, we introduce the classes of \emph{$\W$-operadic}  and \emph{$\W$-commutative} monads (\cref{sec:notions-of-operad,sec:notions-of-commutative-monad}). 
 \item We present a  construction of a \emph{canonical distributive law} $\delta\colon ST\to TS$ for  $\W$-operadic $S$ and $\W$-commutative $T$ (\cref{thm:main-detailed}).
 \item We present the notion of \emph{$\W$-operadic refinement} of a monad (\cref{sec:monad-to-operad}). It can be used to force $\W$-operadicity required by \cref{thm:main-detailed}.
 \item We discuss many  examples of  general constructs. These include various ``multiset monads'' (\cref{ex:MS,ex:AS,ex:Fbij-operadic,ex:Finj-operadic,ex:Fsurj-operadic}) and new distributive laws
 (e.g.\ one that resolves a problem in~\cite{Tsai-et-al-concur},
 see \cref{ex:IVAndTsaiEtAlCONCUR25})
 obtained by forcing $\W$-operadicity (\cref{sec:applications}). 
\end{itemize}
Most proofs are deferred to \cref{sec:apx-proof}.

\subparagraph*{Related work} 
Many related works have been discussed in the context of our motivation, or will be discussed later in a suitable technical context. In particular, the relationship to closely related works~\cite{Manes-Mulry-monad-compositions-1,Parlant-thesis,Dahlqvist-Parlant-Silva-layer-by-layer} is discussed further in \cref{subsec:EMKl}. 

Here are some other related works. A big motivation for distributive laws comes from   the combination of (possibilistic) nondeterminism and probabilistic branching. This combination is observed e.g.\ in Markov decision processes (MDPs), a fundamental model for probabilistic model checking and reinforcement learning. Categorical studies of this specific combination include~\cite{Jacobs21,KozenS24,OngMK25,MioV20,MioSV21}.

Once program/system semantics is defined by monads, reasoning about it with \emph{program logic} is the next step. Recent works in this direction include~\cite{ZilbersteinKST25,LiellCockS25,ChenMM24,ChenMM25,MatacheS19,SimpsonV20}.

\section{Preliminaries}\label{sec:prelim}
The set of all natural numbers (or finite ordinals) is denoted by $\NN=\{0,1,2,\dots\}$.
Each $n\in\NN$ is identified with the $n$-element set $\{0,1,\dots,n-1\}$. 
The category of sets and functions between them is denoted by $\Set$.
It has a symmetric monoidal structure $(1,\times)$ via products.
For any set $X$, we denote by $!_X\colon X\to \vN 1$ the unique function from $X$ to $1$.

Since our focus in this paper is on monads on $\Set$, by a \emph{monad} we shall always mean a \emph{monad on $\Set$}.
Specifically, a \emph{monad} $\mnd T$ (on $\Set$) is a triple $\mnd T= (T,\eta,\mu)$ where $T\colon \Set \to \Set$ is a functor (called the \emph{functor part} of $\mnd T$) and $\eta\colon \id_{\Set}\to T$ and $\mu\colon TT\to T$ are natural transformations (called the \emph{unit} and \emph{multiplication} of $\mnd T$, respectively), such that $\mu\circ T\eta=\id_T=\mu\circ \eta T$ and $\mu\circ T\mu=\mu\circ \mu T$ hold.
Here are some examples of monads.

\begin{example}[$C$-exception monad $\mnd E^C$]\label{ex:exception-monad}
    Any set $C$ has a unique monoid structure $(C,e_C\colon \vN 0\to C,\nabla_C\colon C+C\to C)$ with respect to coproducts. 
    This induces a monad $\mnd E^{C}$ called the \emph{$C$-exception monad}. Its functor part is $C+ (-)\colon \Set\to \Set$, and its unit and multiplication are $\eta_X=e_C+X\colon X\to C+X$ and 
    $\mu_X=\nabla_C+X\colon C+C+X\to C+X$ for each set $X$. \lipicsEnd
\end{example}

\begin{example}[$C$-reader monad $\mnd R^C$]\label{ex:reader-monad}
    Any set $C$ also admits a unique comonoid structure $(C,!_C\colon C\to \vN 1,\Delta_C\colon C\to C\times C)$ with respect to products. 
    It induces the \emph{$C$-reader monad} $\mnd R^C$ in the same way as in \cref{ex:exception-monad}: its functor part
    is $(-)^C\colon \Set\to \Set$, and its unit and multiplication
    are $\eta_X=X^{!_C}\colon X\to X^C$ and $\mu_X=X^{\Delta_C}\colon (X^C)^C\cong X^{C\times C}\to X^C$.
    \lipicsEnd
 \end{example}

\begin{example}[$\mon M$-writer monad $\mnd W^{\mon M}$]\label{ex:writer-monad}
    Let $\mon M=(M,1\in M, \cdot\colon M\times M\to M)$ be a monoid (with respect to products). 
    It induces the monad $\mnd W^{\mon M}$ called the \emph{$\mon M$-writer monad}, whose functor part is $M\times (-)\colon \Set\to \Set$, and whose unit and multiplication are defined by $\eta_X(x)=(1,x)$ and $\mu_X(m,n,x)=(mn,x)$ for each set $X$, $x\in X$, and $m,n\in M$.
    \lipicsEnd
\end{example}

Recall that a \emph{semiring} is a tuple $\sr S=(S, 0\in S, +\colon S\times S\to S, 1\in S, \cdot\colon S\times S\to S)$ where $S$ is a set, $(S,0,+)$ is a commutative monoid, $(S,1,\cdot)$ is a monoid, and 
such that $\cdot\colon (S,0,+)\times (S,0,+)\to (S,0,+)$ is a bihomomorphism of monoids, i.e., for all $x,y,z\in S$,
 $x\cdot 0=0$, $x\cdot (y+z)= x\cdot y + x\cdot z$,
 $0\cdot z=0$, and $(x+y) \cdot z= x\cdot z + y\cdot z$ hold.
The semiring $\sr S$ is \emph{commutative} if $(S,1,\cdot)$ is a commutative monoid.

\begin{example}[${\sr S}$-multiset monad $\mnd M^{\sr S}$]\label{ex:MS}
    Let $\sr S=(S, 0, +, 1, \cdot)$ be a semiring. We define the monad $\mnd M^{\sr S}$ called the \emph{$\sr S$-multiset monad}. The functor part $M^{\sr S}$ of $\mnd M^{\sr S}$ maps each set $X$ to the set $M^{\sr S}X$ of all functions $p\colon X\to S$ such that the set $\{\,x\in X\mid p(x)\neq 0\,\}$
    (called the \emph{support} $\supp p$ of $p$)
    is finite, and maps each function $f\colon X\to Y$ to the function $M^{\sr S}f\colon M^{\sr S}X\to M^{\sr S}Y$ defined for each $p\in M^{\sr S}X$ by
    \[
    \bigl((M^{\sr S}f)(p)\bigr)(y)=\textstyle\sum_{x\in f^{-1}(y)\,\cap\,\supp p}p(x).
    \]
    We can also regard $p\in M^{\sr S}X$ as a formal linear combination 
    \begin{equation}\label{eqn:formal-linear-comb}
        \textstyle\sum_{x\in \supp p}p(x)\cdot x
    \end{equation}
    of elements of $X$ with coefficients from $\sr S$; from this viewpoint, the function $M^{\sr S}f$ maps \eqref{eqn:formal-linear-comb} to
    $\textstyle\sum_{x\in \supp p}p(x)\cdot f(x)$.

    The unit $\eta_X$ and multiplication $\mu_X$ of $\mnd M^{\sr S}$ at $X\in\Set$ are given as follows.
 For each $x\in X$, $\eta_X(x)\in M^{\sr S}X$ is defined by
\begin{math}
    \bigl(\eta_X(x)\bigr)(x')= 1
\end{math}
if $x'=x$, and
$=0$ otherwise.
    That is, $\eta_X(x)$ is $x$ regarded as a trivial linear combination.
 For each $\phi\in  M^{\sr S}M^{\sr S}X$, $\mu_X(\phi)\in M^{\sr S}X$ is defined by 
\begin{math}
     \bigl(\mu_X(\phi)\bigr)(x)=
    \textstyle\sum_{p\in\supp\phi}\phi(p)\cdot  p(x)
\end{math}.
    This flattens a ``linear combination of linear combinations'' $\phi$ into a linear combination. 

    Here are some instances of $\mnd M^{\sr S}$.
    \begin{bracketenumerate}
        \item When $\sr S$ is the semiring $\sr N=(\NN,0,+,1,\cdot)$ of natural numbers, $\mnd M^{\sr N}$ is called the \emph{(finite) multiset monad} $\mnd M=(M,\eta,\mu)$. For a set $X$, $MX$ is the set of all finite multisets over $X$. 
        \item When $\sr S$ is the semiring $\sr 2=(\{0,1\},0,\vee,1,\wedge)$, $\mnd M^{\sr 2}$ is called the \emph{finite powerset monad} $\PPfin=(\Pfin,\eta,\mu)$. 
        For each set $X$, $\Pfin X$ is the set of all finite subsets of $X$. 
        \item When $\sr S$ is the semiring $\sr R_{\geq 0}=(\RRnonneg,0,+,1,\cdot)$ of nonnegative real numbers, $\mnd M^{\sr R_{\geq 0}}$ is called the \emph{(finite) valuation monad} (see e.g.\ \cite{Varacca-Winskel-dist-law}) and is written as $\mnd V=(V,\eta,\mu)$.
    \end{bracketenumerate}
    Note that the semirings $\sr N$, $\sr 2$, and $\sr R_{\geq 0}$ are all commutative.
    \lipicsEnd
\end{example}

\begin{example}[affine ${\sr S}$-multiset monad $\mnd A^{\sr S}$]\label{ex:AS}
    Let $\sr S=(S, 0, +, 1, \cdot)$ be a semiring. We define the monad $\mnd A^{\sr S}=(A^{\sr S},\eta,\mu)$ called the \emph{affine $\sr S$-multiset monad}.
    This is the submonad of the $\sr S$-multiset monad $\mnd M^{\sr S}$ such that, for each set $X$, $A^{\sr S} X$ consists of all formal \emph{affine} combinations of elements of $X$, i.e., $p\in M^{\sr S} X$ with $\sum_{x\in \supp p}p(x)=1$ (cf.\ \cref{eqn:formal-linear-comb}). ($\mnd A^{\sr S}$ is the \emph{affine part} of $\mnd M^{\sr S}$ in the sense of \cite{Lindner-affine-part}.)
    Here are the instances with $\sr S=\sr 2$ (\cref{ex:MS}(2)) and $\sr R_{\geq 0}$ (\cref{ex:MS}(3)).
    \begin{bracketenumerate}
        \item $\mnd A^{\sr 2}$ is the \emph{nonempty finite powerset monad} $\PPfinne$. 
        \item $\mnd A^{\sr R_{\geq 0}}$ is the \emph{(finitely supported) probability distribution monad} $\mnd D$ (cf.\ \cite[\S~16]{Kock-comm-mnd-distribution}).
        \lipicsEnd
    \end{bracketenumerate}
\end{example}

\section{\texorpdfstring{$\W$}{W}-Operadic Monads}\label{sec:notions-of-operad}
In this section, we will recall \emph{$\W$-operads} \cite[Def.~2]{Tronin-abstract-clones-and-operads}, where $\W$ is any \emph{verbal category} in the sense of \cite[Def.~1]{Tronin-abstract-clones-and-operads}.
We then define monads induced by
$\W$-operads, and the ``change of base'' constructions for $\W$-operads.

\subsection{Verbal Categories}\label{subsec:verbal-cat}
Towards a formal definition of verbal categories (\cref{def:verbal-cat}), we start with an informal introduction. Recall from \cref{sec:intro} that $\F$ is the category of finite ordinals and all functions between them, and that
string diagrams represent morphisms in $\F$ (see \eqref{eqn:string-diag}).

A verbal category $\W$ is defined to be a wide subcategory of $\F$---meaning that $\mathop{\mathrm{ob}}\W=\mathop{\rm ob} \F$---that is closed under what we call the \emph{verbal substitution} operation $\star$ on morphisms of $\F$. 
We sketch the verbal substitution operation $\star$. Its type is as follows.
\begin{displaymath}
 \vcenter{\infer{\beta\star\vec{\alpha}\colon \sum_{i'\in n'} m'_{\beta(i')}\to \sum_{i\in n}m_i}{
  \beta\colon n'\to n \qquad
  &
  \quad
  \alpha_{0}\colon m'_{0}\rightarrow m_0 
  \qquad\cdots\qquad
  \alpha_{n-1}\colon m'_{n-1}\rightarrow m_{n-1} 
}} 
\end{displaymath}
Here, $\vec \alpha= (\alpha_i)_{i\in n}$.
Its concrete definition (\cref{def:operad-of-functions}) is illustrated as follows. 
Let $\beta\colon 4\to 3$, $\alpha_{0}\colon 2\to 2$, 
$\alpha_{1}\colon 3\to 2$, and
$\alpha_{2}\colon 2\to 1$, as depicted below.
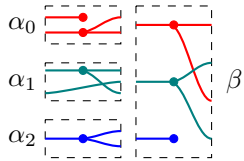
\begin{equation*}
\begin{tikzpicture}[baseline=-\the\dimexpr\fontdimen22\textfont2\relax ]
\draw[dashed] (0,1) -- (1,1) -- (1,0.5) -- (0,0.5) -- cycle;
\draw[dashed] (0,0.25) -- (1,0.25) -- (1,-0.25) -- (0,-0.25) -- cycle;
\draw[dashed] (0,-0.5) -- (1,-0.5) -- (1,-1) -- (0,-1) -- cycle;
\draw[thick,red] (0,0.85) -- (0.5,0.85);
\filldraw [red] (0.5,0.85) circle [radius=1.5pt];
\filldraw [red] (0.5,0.65) circle [radius=1.5pt];
\draw[thick,red] (0,0.65) -- (1,0.65);
\draw[thick,red] (0.45,0.65) .. controls (0.65,0.65) and (0.8,0.85) .. (1,0.85);
\node at (-0.3,0.75) {$\alpha_0$};
\draw[thick,teal] (0,0.15) -- (1,0.15);
\draw[thick,teal] (0.45,0.15) .. controls (0.65,0.15) and (0.8,-0.15) .. (1,-0.15);
\draw[thick,teal] (0,-0.15) .. controls (0.2,-0.15) and (0.8,0) .. (1,0);
\filldraw [teal] (0.5,0.15) circle [radius=1.5pt];
\node at (-0.3,0) {$\alpha_1$};
\draw[thick,blue] (0,-0.75)--(0.5,-0.75);
\draw[thick,blue] (0.45,-0.75) .. controls (0.65,-0.75) and (0.8,-0.65) .. (1,-0.65);
\draw[thick,blue] (0.45,-0.75) .. controls (0.65,-0.75) and (0.8,-0.85) .. (1,-0.85);
\filldraw [blue] (0.5,-0.75) circle [radius=1.5pt];
\node at (-0.3,-0.75) {$\alpha_2$};
\draw[dashed] (1.2,1) -- (2.2,1) -- (2.2,-1) -- (1.2,-1) -- cycle;
\draw[thick,red] (1.2,0.75) -- (2.2,0.75);
\draw[thick,red] (1.65,0.75) .. controls (1.85,0.75) and (2,-0.25) .. (2.2,-0.25);
\draw[thick,teal] (1.2,0)--(1.7,0);
\draw[thick,teal] (1.65,0) .. controls (1.85,0) and (2,0.25) .. (2.2,0.25);
\draw[thick,teal] (1.65,0) .. controls (1.85,0) and (2,-0.75) .. (2.2,-0.75);
\draw[thick,blue] (1.2,-0.75)--(1.7,-0.75);
\filldraw [red] (1.7,0.75) circle [radius=1.5pt];
\filldraw [blue] (1.7,-0.75) circle [radius=1.5pt];
\filldraw [teal] (1.7,0) circle [radius=1.5pt];
\node at (2.5,0) {$\beta$};
\end{tikzpicture}
\end{equation*}
The morphism $\beta\star\vec{\alpha}$ is then obtained  by 1) finding the correspondence between $\alpha_{i}$ and the $i$-th end
(on the left) 
of $\beta$, for each $i\in 3$ (we used colors to highlight this), 2) ``duplicating'' wires in $\beta$ so that the numbers of wires match between $\beta$ and $\vec{\alpha}$, and 3) connecting them. 
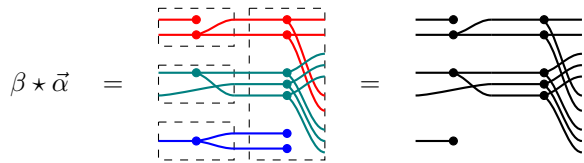
\begin{equation*}
\beta\star\vec{\alpha}\quad=\quad
\begin{tikzpicture}[baseline=-\the\dimexpr\fontdimen22\textfont2\relax ]
\draw[dashed] (0,1) -- (1,1) -- (1,0.5) -- (0,0.5) -- cycle;
\draw[dashed] (0,0.25) -- (1,0.25) -- (1,-0.25) -- (0,-0.25) -- cycle;
\draw[dashed] (0,-0.5) -- (1,-0.5) -- (1,-1) -- (0,-1) -- cycle;
\draw[thick,red] (0,0.85) -- (0.5,0.85);
\filldraw [red] (0.5,0.85) circle [radius=1.5pt];
\filldraw [red] (0.5,0.65) circle [radius=1.5pt];
\draw[thick,red] (0,0.65) -- (1,0.65);
\draw[thick,red] (0.45,0.65) .. controls (0.65,0.65) and (0.8,0.85) .. (1,0.85);
\draw[thick,teal] (0,0.15) -- (1,0.15);
\draw[thick,teal] (0.45,0.15) .. controls (0.65,0.15) and (0.8,-0.15) .. (1,-0.15);
\draw[thick,teal] (0,-0.15) .. controls (0.2,-0.15) and (0.8,0) .. (1,0);
\filldraw [teal] (0.5,0.15) circle [radius=1.5pt];
\draw[thick,blue] (0,-0.75)--(0.5,-0.75);
\draw[thick,blue] (0.45,-0.75) .. controls (0.65,-0.75) and (0.8,-0.65) .. (1,-0.65);
\draw[thick,blue] (0.45,-0.75) .. controls (0.65,-0.75) and (0.8,-0.85) .. (1,-0.85);
\filldraw [blue] (0.5,-0.75) circle [radius=1.5pt];
\draw[dashed] (1.2,1) -- (2.2,1) -- (2.2,-1) -- (1.2,-1) -- cycle;
\draw[thick,red] (1,0.85) -- (2.2,0.85);
\draw[thick,red] (1,0.65) -- (2.2,0.65);
\draw[thick,red] (1.65,0.65) .. controls (1.85,0.65) and (2,-0.35) .. (2.2,-0.35);
\draw[thick,red] (1.65,0.85) .. controls (1.85,0.85) and (2,-0.15) .. (2.2,-0.15);
\filldraw [red] (1.7,0.85) circle [radius=1.5pt];
\filldraw [red] (1.7,0.65) circle [radius=1.5pt];
\draw[thick,teal] (1,0.15)--(1.7,0.15);
\draw[thick,teal] (1,0)--(1.7,0);
\draw[thick,teal] (1,-0.15)--(1.7,-0.15);
\draw[thick,teal] (1.65,0) .. controls (1.85,0) and (2,0.25) .. (2.2,0.25);
\draw[thick,teal] (1.65,0) .. controls (1.85,0) and (2,-0.75) .. (2.2,-0.75);
\draw[thick,teal] (1.65,0.15) .. controls (1.85,0.15) and (2,0.4) .. (2.2,0.4);
\draw[thick,teal] (1.65,0.15) .. controls (1.85,0.15) and (2,-0.6) .. (2.2,-0.6);
\draw[thick,teal] (1.65,-0.15) .. controls (1.85,-0.15) and (2,0.1) .. (2.2,0.1);
\draw[thick,teal] (1.65,-0.15) .. controls (1.85,-0.15) and (2,-0.9) .. (2.2,-0.9);
\draw[thick,blue] (1,-0.65)--(1.7,-0.65);
\draw[thick,blue] (1,-0.85)--(1.7,-0.85);
\filldraw [blue] (1.7,-0.65) circle [radius=1.5pt];
\filldraw [blue] (1.7,-0.85) circle [radius=1.5pt];
\filldraw [teal] (1.7,0) circle [radius=1.5pt];
\filldraw [teal] (1.7,0.15) circle [radius=1.5pt];
\filldraw [teal] (1.7,-0.15) circle [radius=1.5pt];
\end{tikzpicture}
\quad=\quad
\begin{tikzpicture}[baseline=-\the\dimexpr\fontdimen22\textfont2\relax ]
\draw[thick] (0,0.85) -- (0.5,0.85);
\filldraw [fill=black] (0.5,0.85) circle [radius=1.5pt];
\filldraw [fill=black] (0.5,0.65) circle [radius=1.5pt];
\draw[thick] (0,0.65) -- (1,0.65);
\draw[thick] (0.45,0.65) .. controls (0.65,0.65) and (0.8,0.85) .. (1,0.85);
\draw[thick] (0,0.15) -- (1,0.15);
\draw[thick] (0.45,0.15) .. controls (0.65,0.15) and (0.8,-0.15) .. (1,-0.15);
\draw[thick] (0,-0.15) .. controls (0.2,-0.15) and (0.8,0) .. (1,0);
\filldraw [fill=black] (0.5,0.15) circle [radius=1.5pt];
\draw[thick] (0,-0.75)--(0.5,-0.75); 
\filldraw [fill=black] (0.5,-0.75) circle [radius=1.5pt];
\draw[thick] (1,0.85) -- (2.2,0.85);
\draw[thick] (1,0.65) -- (2.2,0.65);
\draw[thick] (1.65,0.65) .. controls (1.85,0.65) and (2,-0.35) .. (2.2,-0.35);
\draw[thick] (1.65,0.85) .. controls (1.85,0.85) and (2,-0.15) .. (2.2,-0.15);
\filldraw [fill=black] (1.7,0.85) circle [radius=1.5pt];
\filldraw [fill=black] (1.7,0.65) circle [radius=1.5pt];
\draw[thick] (1,0.15)--(1.7,0.15);
\draw[thick] (1,0)--(1.7,0);
\draw[thick] (1,-0.15)--(1.7,-0.15);
\draw[thick] (1.65,0) .. controls (1.85,0) and (2,0.25) .. (2.2,0.25);
\draw[thick] (1.65,0) .. controls (1.85,0) and (2,-0.75) .. (2.2,-0.75);
\draw[thick] (1.65,0.15) .. controls (1.85,0.15) and (2,0.4) .. (2.2,0.4);
\draw[thick] (1.65,0.15) .. controls (1.85,0.15) and (2,-0.6) .. (2.2,-0.6);
\draw[thick] (1.65,-0.15) .. controls (1.85,-0.15) and (2,0.1) .. (2.2,0.1);
\draw[thick] (1.65,-0.15) .. controls (1.85,-0.15) and (2,-0.9) .. (2.2,-0.9);
\filldraw [fill=black] (1.7,0) circle [radius=1.5pt];
\filldraw [fill=black] (1.7,0.15) circle [radius=1.5pt];
\filldraw [fill=black] (1.7,-0.15) circle [radius=1.5pt];
\end{tikzpicture}
\end{equation*}

In order to define $\star$ formally, we fix specific coproduct diagrams in the category $\F$.

\begin{definition}[chosen coproducts in $\F$]\label{notation-kappa}
 Let $n\in\NN$, $\vec{m}=(m_{i}\in\NN)_{i\in n}$, and  $m=\sum_{i\in n}m_{i}$. 
    For each $i\in \vN n$, we define the function
 \begin{displaymath}
  \iota_{\vec m,i}\colon m_{i} \to m
 \end{displaymath}
 by 
 $\iota_{\vec m, i}(j)=\bigl(\sum_{k\in \vN i}m_k\bigr)+j$ for each $j\in \vN{m_i}$. 
The family $(\iota_{\vec{m},i}\colon m_{i}\to m )_{i\in n}$ exhibits $m$ as a coproduct of $\vec{m}$ in $\F$. \lipicsEnd
\end{definition}

\begin{definition}[$\beta\star\vec{\alpha}$; cf.\ {\cite[\S~2.2]{Szawiel-Zawadowski-monads-of-regular-theories} and \cref{W-operad-of-monoids}}]\label{def:operad-of-functions}
    Given morphisms 
    $\bigl(\,\beta\colon \vN{n'}\to \vN n,\, \vec \alpha=(\alpha_i\colon \vN{m'_i}\to \vN{m_i}\,)_{i\in \vN n}\bigr)$
    in $\F$, 
    define $m'=\sum_{i'\in\vN{n'}}m'_{\beta(i')}$ and $m=\sum_{i\in\vN n}m_i$.
    We define
    \[
    \beta\star\vec \alpha\colon \vN{m'}\to \vN{m}
    \]
    as the unique morphism in 
    $\F$ making the following square commute for each $i'\in\vN{n'}$.
    \begin{equation*}
\begin{tikzpicture}[baseline=-\the\dimexpr\fontdimen22\textfont2\relax ]
    \node(01) at (0,0.5) {$\vN{m'_{\beta(i')}}$};
    \node(02) at (2,0.5) {$\vN{m_{\beta(i')}}$};
    \node(11) at (0,-0.5) {$\vN{m'}$};
    \node(12) at (2,-0.5) {$\vN m$};
    \draw [->] (01) to node[auto,labelsize] {$\alpha_{\beta(i')}$} (02);
    \draw [->] (11) to node[auto,labelsize] {$\beta\star\vec\alpha$} (12);
    \draw [->] (01) to node[auto,swap,labelsize] {$\iota_{\vec m'\beta,i'}$} (11);
    \draw [->] (02) to node[auto,labelsize] {$\iota_{\vec m,\beta(i')}$} (12);
\end{tikzpicture}
\end{equation*}
Here, $\vec m'\beta=(m'_{\beta(i')})_{i'\in\vN{n'}}$ and $\vec m=(m_i)_{i\in\vN n}$.\lipicsEnd
\end{definition}

\begin{definition}[verbal category]\label{def:verbal-cat}
    A wide subcategory $\W$ of $\F$ is called a \emph{verbal category} \cite[Def.~1]{Tronin-abstract-clones-and-operads} if it is closed under the verbal substitution operation $\beta\star\vec{\alpha}$ of \cref{def:operad-of-functions}.
    We say that a verbal category is \emph{symmetric} if it contains $\Fbij$ (see \cref{eqn:Ctxt-logically}).\lipicsEnd
\end{definition}

To see the equivalence of the above definition of verbal categories to the original \cite[Def.~1]{Tronin-abstract-clones-and-operads}, consider the cases of $\beta\star\vec\alpha$ where $\beta$ or each of $\alpha_i$ is an identity map.
We note that verbal categories are called \emph{faithful cartesian clubs} in \cite[Def.\ 2.6.3]{Shulman-catlog}, and that symmetric verbal categories are called \emph{contextual subcategories of $\F$} in \cite{Lenke-Milius-Urbat-subst}.

\begin{example}\label{ex:verbal-cats}
As already mentioned, the six categories in \cref{eqn:Ctxt-logically} are verbal categories
\cite[p.~748]{Tronin-abstract-clones-and-operads}.
For every positive integer $n$, the verbal category $\Fsurj^{(n)}$ generated by $\mathrm{swap}\colon 2\to 2$ and $!_n\colon n\to 1$ (i.e.,
\begin{tikzpicture}[baseline=-\the\dimexpr\fontdimen22\textfont2\relax ]
\draw[thick] (0,0.1) .. controls (0.3,0.1) and (0.5,-0.1) .. (0.7,-0.1);
\draw[thick] (0,-0.1) .. controls (0.3,-0.1) and (0.5,0.1) .. (0.7,0.1);
\end{tikzpicture}
and
\begin{tikzpicture}[baseline=-\the\dimexpr\fontdimen22\textfont2\relax ]
\draw[thick] (0,0.2) .. controls (-0.2,0.2) and (-0.3,0) .. (-0.4,0);
\draw[thick] (0,0.1) .. controls (-0.2,0.1) and (-0.3,0) .. (-0.4,0);
\draw[thick] (0,-0.2) .. controls (-0.2,-0.2) and (-0.3,0) .. (-0.4,0);
\draw[thick] (-0.4,0)-- (-0.7,0);
\filldraw [fill=black] (-0.35,0) circle [radius=1.5pt];
\node[rotate=90] at (0,-0.04) {$\scriptscriptstyle\cdots$};
\draw [decorate,decoration={brace,amplitude=4pt}]
  (0.1,0.25) -- (0.1,-0.25) node[midway,xshift=9pt]{$n$};
\end{tikzpicture})
consists of all surjections $\alpha\colon\vN \ell\to\vN m$ such that each fiber $\alpha^{-1}(i)$ has cardinality $k(n-1)+1$ for some $k\in \NN$;
this yields an infinite family of verbal categories, with $\Fsurj^{(1)}=\Fbij$ and $\Fsurj^{(2)}=\Fsurj$. Intuitively, $\Fsurj^{(n)}$ models contexts with  exchange and ``contraction with arity $n$''.\lipicsEnd
\end{example}

\begin{example}\label{ex:verbal-cats-nonex}
As non-examples, we note that the wide subcategory of $\F$ consisting of all monotone (resp.\ surjective monotone) maps (corresponding to the set $\{\mathrm W,\mathrm C\}$ (resp.\ $\{\mathrm C\}$) of structural rules; cf.\ \cref{eqn:Ctxt-logically}) is not a verbal category.\footnote{This shows that the third example in \cite[p.\ 748]{Tronin-abstract-clones-and-operads} is incorrect.} 
This is because the function $!_2\colon \vN 2\to \vN 1$ is monotone but 
$!_2\star(\id_{\vN 2})\colon \vN 4\to \vN 2$ is not,
as the picture 
\begin{tikzpicture}[baseline=-\the\dimexpr\fontdimen22\textfont2\relax ]
\draw[dashed] (0,0.25) -- (1,0.25) -- (1,-0.25) -- (0,-0.25) -- cycle;
\draw[thick] (0,0.07) -- (1.7,0.07);
\draw[thick] (0,-0.07) -- (1.7,-0.07);
\filldraw [fill=black] (1.7,0.07) circle [radius=1.5pt];
\filldraw [fill=black] (1.7,-0.07) circle [radius=1.5pt];
\draw[dashed] (1.2,0.25) -- (2.2,0.25) -- (2.2,-0.25) -- (1.2,-0.25) -- cycle;
\draw[thick] (1.65,0.07) .. controls (1.85,0.07) and (2,0.21) .. (2.2,0.21);
\draw[thick] (1.65,0.07) .. controls (1.85,0.07) and (2,-0.07) .. (2.2,-0.07);
\draw[thick] (1.65,-0.07) .. controls (1.85,-0.07) and (2,-0.21) .. (2.2,-0.21);
\draw[thick] (1.65,-0.07) .. controls (1.85,-0.07) and (2,0.07) .. (2.2,0.07);
\end{tikzpicture} 
shows
(cf.~\cite[\S~2]{Curien-operads-clones}).\lipicsEnd
\end{example}

\subsection{\texorpdfstring{$\W$}{W}-Operads}
A formal definition (\cref{def:Ctxt-operad}) of $\W$-operads may look complicated. Illustration by string diagram, which will follow, should be more intuitive.
\begin{definition}[{$\W$-operad \cite[Def.~2]{Tronin-abstract-clones-and-operads}}]\label{def:Ctxt-operad}
    Let $\W$ be a verbal category.
    A \emph{$\W$-operad} $\mnd O=(O,\ido,\subst)$ consists of the following.
    \begin{itemize}
        \item A functor $O\colon \W\to \Set$. 
        For each object $\vN n\in \W$, the set $O\vN n$ is denoted by $O_n$.
        For each morphism $\alpha\colon \vN m\to\vN n$ in $\W$ and $p\in O_m$, the element $(O\alpha)p\in O_n$ is also denoted by $\alpha p$.
        \item An element $\ido\in O_1$.
        \item For each $n\in \NN$ and $\vec m=(m_i\in\NN)_{i\in\vN n}$, a function 
        \begin{equation}\label{eqn:subst}
            \subst_{n,\vec m}\colon O_n\times \textstyle\prod_{i\in\vN n}O_{m_i}\to O_m,
        \end{equation}
        where $m=\sum_{i\in\vN n}m_i$.
        Its value at $(q,\vec p)\in O_n\times \prod_{i\in\vN n}O_{m_i}$ is denoted by $\subst_{n,\vec m}(q,\vec p)$.
    \end{itemize}
    These data are subject to the following axioms.
    \begin{itemize}
        \item The compatibility of the $\subst$ functions \eqref{eqn:subst} with 
        the action by morphisms in $\W$: for any $\beta\colon \vN{n'}\to\vN{n}$ and $\vec \alpha=(\alpha_i\colon \vN{m_i'}\to\vN{m_i})_{i\in\vN n}$ in $\W$, the following diagram commutes. 
        \begin{equation}\label{eqn:subst-compatibility}
\begin{tikzpicture}[baseline=-\the\dimexpr\fontdimen22\textfont2\relax ]
    \node(01) at (0,0.5) {$O_{n'}\times \prod_{i\in\vN{n}}O_{m'_{i}}$};
    \node(02) at (5.75,0.5) {$O_{n'}\times \prod_{i'\in\vN{n'}}O_{m'_{\beta(i')}}$};
    \node(03) at (9,0.5) {$O_{m'}$};
    \node(11) at (0,-0.5) {$O_{n}\times \prod_{i\in\vN{n}}O_{m_{i}}$};
    \node(13) at (9,-0.5) {$O_m$};
    \draw [->] (01) to node[auto,labelsize] {$1\times\prod_{i'\in\vN{n'}}\pi_{\beta(i')}$} (02);
    \draw [->] (02) to node[auto,labelsize] {$\subst_{n',\vec m\beta}$} (03);
    \draw [->] (11) to node[auto,labelsize] {$\subst_{n,\vec{m}}$} (13);
    \draw [->] (01) to node[auto,swap,labelsize] {$O\beta\times \prod_{i\in\vN n}O\alpha_i$} (11);
    \draw [->] (03) to node[auto,labelsize] {$O(\beta\star\vec\alpha)$} (13);
\end{tikzpicture}
\end{equation}
    Here, $\pi_{i}\colon \prod_{i\in\vN n}O_{m'_i}\to O_{m'_i}$ is the $i$-th projection (and hence $\pi_{\beta(i')}$ is the $\beta(i')$-th projection), $\vec m\beta =(m_{\beta(i')})_{i'\in n'}$, and $\beta\star\vec{\alpha}$ is from \cref{def:operad-of-functions}.
        \item The unit laws: for each natural number $m$ and each $p\in O_m$, we have 
        \[\subst_{1,(m)}(\ido,(p))=p=\subst_{m,(1)_{i\in\vN m}} (p,(\ido)_{i\in\vN m}).\]
        \item The associativity law:
        \begin{equation*}
            \subst_{n,\vec m}\bigl(\,r,\,\bigl(\,\subst_{m_j,\vec\ell^{(j)}}(q_j,\vec p^{(j)})\,\bigr)_{j\in\vN n}\bigr)
           =\ \subst_{m,\vec \ell}\bigl(\,\subst_{n,\vec m} (r,\vec q),\, \vec p\,\bigr),
        \end{equation*}
        for all 
        $n\in \NN$, 
        $\vec m=(m_j\in\NN)_{j\in n}$, 
        $\vec \ell= (\vec \ell^{(j)})_{j\in n}=\bigl((\ell^{(j)}_i\in\NN)_{i\in m_j}\bigr)_{j\in n}$,
        $r\in O_n$, $\vec q=(q_j\in O_{m_j})_{j\in n}$, and $\vec p=(\vec p^{(j)})_{j\in n}=\bigl((p^{(j)}_i\in O_{\ell_{i}^{(j)}})_{i\in m_j}\bigr)_{j\in n}$, where $m=\sum_{j\in n}m_j$.\lipicsEnd
    \end{itemize}
\end{definition}

Note that the unit and associativity laws in \cref{def:Ctxt-operad} are 
identical to those for (symmetric or non-symmetric) operads and multicategories (see e.g.~\cite[Def.~2.1.1]{Leinster-higher-operads}).

The idea is that a $\W$-operad $\mnd O$ encodes 
an algebraic theory in the ``substructural context'' corresponding to $\W$.
More specifically,
an element $p$ of $O_n$ represents an $n$-ary term (modulo provable equality) in the algebraic theory, and can be depicted as the  ``gadget'' 
\begin{tikzpicture}[baseline=-\the\dimexpr\fontdimen22\textfont2\relax, x=6ex,y=6ex]
 \draw[thick] (0,0) -- (-0.5,-0.5) -- (-0.5,0.5) -- cycle;
 \draw[thick] (-0.5,0.3) -- (-0.8,0.3);
 \draw[thick] (-0.5,-0.3) -- (-0.8,-0.3);
 \draw[thick] (0.3,0) -- (0,0);
 \node at (-0.35,0) {$p$};
 \node[rotate=90] at (-0.7,0) {$\cdots$};
 \draw [decorate,decoration={brace,mirror,amplitude=5pt}]
  (-0.9,0.4) -- (-0.9,-0.4) node[midway,xshift=-9pt]{$n$};
\end{tikzpicture} 
with $n$ input wires and  one output wire.
The action of morphisms 
of $\W$ on the functor $O$ amounts to rearranging the input wires. For example, if $\W=\F$ and $\alpha\colon \vN 3\to \vN 4$ is the function depicted on the left of \eqref{eqn:string-diag},
then 
$\alpha p\in O_4$ can be obtained from 
$p\in O_3$ as follows.
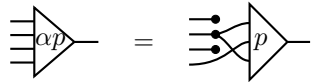
\begin{equation*}
\begin{tikzpicture}[baseline=-\the\dimexpr\fontdimen22\textfont2\relax,x=7ex,y=6ex]
\draw[thick] (0,0) -- (-0.5,-0.5) -- (-0.5,0.5) -- cycle;
\draw[thick] (-0.5,0.3) -- (-0.8,0.3);
\draw[thick] (-0.5,-0.3) -- (-0.8,-0.3);
\draw[thick] (-0.5,0.1) -- (-0.8,0.1);
\draw[thick] (-0.5,-0.1) -- (-0.8,-0.1);
\draw[thick] (0.3,0) -- (0,0);
\node at (-0.3,0) {$\alpha p$};
\end{tikzpicture}
\quad =\quad 
\begin{tikzpicture}[baseline=-\the\dimexpr\fontdimen22\textfont2\relax ]
\draw[thick] (0,0) -- (-0.5,-0.5) -- (-0.5,0.5) -- cycle;
\draw[thick] (-0.5,0) .. controls (-0.7,0) and (-0.8,-0.3) .. (-1.3,-0.3);
\draw[thick] (-0.5,-0.25) .. controls (-0.7,-0.25) and (-0.8,0.1) .. (-1,0.1);
\draw[thick] (-0.5,0.25) .. controls (-0.7,0.25) and (-0.8,0.1) .. (-1,0.1);
\draw[thick] (-1,0.1)--(-1.3,0.1);
\draw[thick] (-0.95,0.3) -- (-1.3,0.3);
\draw[thick] (-0.95,-0.1) -- (-1.3,-0.1);
\draw[thick] (0.3,0) -- (0,0);
\node at (-0.35,0) {$p$};
\filldraw [fill=black] (-0.95,-0.1) circle [radius=1.5pt];
\filldraw [fill=black] (-0.95,0.1) circle [radius=1.5pt];
\filldraw [fill=black] (-0.95,0.3) circle [radius=1.5pt];
\end{tikzpicture}
\end{equation*}
Syntactically, the $4$-ary term $\alpha p$ is defined from the $3$-ary term $p$ as 
\begin{math}\label{eqn:alpha-p-syntactically}
(\alpha p)(x_0,x_1,x_2,x_3)= p(x_1,x_3,x_1).
\end{math}
We depict $\ido\in O_1$ as a wire
(i.e. \begin{tikzpicture}[baseline=-\the\dimexpr\fontdimen22\textfont2\relax,x=6ex,y=6ex]
\draw[thick] (0,0) -- (-0.7,-0.3) -- (-0.7,0.3) -- cycle;
\draw[thick] (-0.7,0) -- (-1,0);
\draw[thick] (0.3,0) -- (0,0);
\node at (-0.5,0) {$\ido$};
\end{tikzpicture}
 $=$
\begin{tikzpicture}[baseline=-\the\dimexpr\fontdimen22\textfont2\relax ]
\draw[thick] (-0.1,0) -- (0.8,0);
\end{tikzpicture})
and $\subst_{n,\vec m}(q,\vec p)$
as follows.
\begin{equation}\label{eqn:substStringDiagram}
\scalebox{.8}{\begin{tikzpicture}[baseline=-\the\dimexpr\fontdimen22\textfont2\relax]
\draw[thick] (2,0) -- (1,1) -- (1,-1) -- cycle;
\node at (1.4,0) {$q$};
\node[rotate=90] at (0.8,0) {$\cdots$};
\draw[thick] (2.5,0) -- (2,0);
\draw[thick] (0.5,0.7) -- (-0.5,0.3) -- (-0.5,1.1) -- cycle;
\draw[thick] (-0.5,1) -- (-0.8,1);
\draw[thick] (-0.5,0.4) -- (-0.8,0.4);
\draw[thick] (1,0.7) -- (0.5,0.7);
\node at (-0.2,0.7) {$p_0$};
\node[rotate=90] at (-0.7,0.7) {$\cdots$};
\draw [decorate,decoration={brace,mirror,amplitude=5pt}]
  (-0.9,1.1) -- (-0.9,0.3) node[midway,xshift=-12pt]{$m_0$};
\draw[thick] (0.5,-0.7) -- (-0.5,-1.1) -- (-0.5,-0.3) -- cycle;
\draw[thick] (-0.5,-0.4) -- (-0.8,-0.4);
\draw[thick] (-0.5,-1) -- (-0.8,-1);
\draw[thick] (1,-0.7) -- (0.5,-0.7);
\node at (-0.12,-0.7) {$p_{n-1}$};
\node[rotate=90] at (-0.7,-0.7) {$\cdots$};
\draw [decorate,decoration={brace,mirror,amplitude=5pt}]
  (-0.9,-0.3) -- (-0.9,-1.1) node[midway,xshift=-17pt]{$m_{n-1}$};
  \draw [decorate,decoration={brace,mirror,amplitude=5pt}]
  (-1.9,1.1) -- (-1.9,-1.1) node[midway,xshift=-11pt]{$m$};
\end{tikzpicture}}
\end{equation}
Syntactically, $\ido$ corresponds to the variable $x_0$ regarded as a unary term, and the term $\subst_{n,\vec m}(q,\vec p)$ is obtained by suitably \emph{substituting} $\vec p$ for the variables in $q$.
For example,
\[
\subst_{2,(2,1)}\bigl(\,q(x_0,x_1),\,\bigl(p_0(x_0,x_1),\,p_1(x_0)\bigr)\,\bigr)= q\bigl(\,p_0(x_0,x_1),\,p_1(x_2)\,\bigr);
\]
note that the variable $x_0$ in the term $p_1$ is renamed to $x_2$.
The unit and associativity laws for $\W$-operads will be trivial from these pictorial or syntactic viewpoints;
the compatibility axiom \cref{eqn:subst-compatibility} can be understood in an analogous manner, too. 
See \cite[Chap.~2]{Leinster-higher-operads} for an extended discussion with many pictures in the case where $\W=\Fid$ or $\Fbij$.
We will discuss examples of $\W$-operads in \cref{subsec:examples-of-W-operads}.

The general notion of $\W$-operad specializes to the following:
\begin{itemize}
    \item $\Fid$-operads coincide with \emph{non-symmetric operads} \cite{Leinster-higher-operads},
    \item $\Fbij$-operads coincide with \emph{symmetric operads} \cite{Kelly-operads-of-May}, 
    \item $\Fsurj$-operads coincide with \emph{regular operads} \cite{Szawiel-Zawadowski-monads-of-regular-theories}, and 
    \item $\F$-operads coincide with \emph{(abstract) clones} \cite[Thm.\ on p.~751]{Tronin-abstract-clones-and-operads}.
\end{itemize}
The
classes of $\W$-operads for various $\W$ can also be regarded as a generalization of the \emph{Boom hierarchy}, whose relevance to distributive laws has been noted; see \cite[\S~7]{King-Wadler-combining-monads} and \cite[Chap.~6]{Zwart-thesis}.

\begin{remark}\label{rmk:operads-as-monoid-objects} 
    For any verbal category $\W$, the functor category $[\W,\Set]$ is known to have a monoidal structure $(J,\bullet)$ called the \emph{substitution monoidal structure}, with the inclusion functor $J\colon \W\to \Set$ as the unit object; see \cite{Lenke-Milius-Urbat-subst,Tanaka-Power-pseudo-distributive,Tanaka-Power-substructural} or \cite[\S~4]{Curien-operads-clones}.
    Given $P,Q\in[\W,\Set]$ and $m\in\W$, an element of $(Q\bullet P)_m$ is a tuple 
    \begin{equation}\label{eqn:Q-bullet-P}
        \bigl(\,q\in Q_n,\,(p_i\in P_{m_i})_{i\in n},\, \alpha\in \W(\textstyle\sum_{i\in n}m_i,m) \,\bigr)
    \end{equation}
    modulo a suitable equivalence relation (see \cite{Curien-operads-clones} for a coend formula).
 The data \cref{eqn:Q-bullet-P} can be depicted as follows.
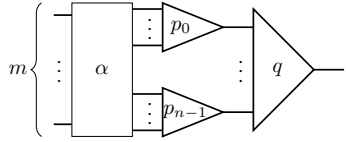
\[
\scalebox{.8}{\begin{tikzpicture}[baseline=-\the\dimexpr\fontdimen22\textfont2\relax]
\draw[thick] (2,0) -- (1,1) -- (1,-1) -- cycle;
\node at (1.4,0) {$q$};
\node[rotate=90] at (0.8,0) {$\cdots$};
\draw[thick] (2.5,0) -- (2,0);
\draw[thick] (0.5,0.7) -- (-0.5,0.3) -- (-0.5,1.1) -- cycle;
\draw[thick] (-0.5,1) -- (-1,1);
\draw[thick] (-0.5,0.4) -- (-1,0.4);
\draw[thick] (1,0.7) -- (0.5,0.7);
\node at (-0.2,0.7) {$p_0$};
\node[rotate=90] at (-0.7,0.7) {$\cdots$};
\draw[thick] (0.5,-0.7) -- (-0.5,-1.1) -- (-0.5,-0.3) -- cycle;
\draw[thick] (-0.5,-0.4) -- (-1,-0.4);
\draw[thick] (-0.5,-1) -- (-1,-1);
\draw[thick] (1,-0.7) -- (0.5,-0.7);
\node at (-0.12,-0.7) {$p_{n-1}$};
\node[rotate=90] at (-0.7,-0.7) {$\cdots$};
  \draw[] (-1,1.1) -- (-2,1.1) -- (-2,-1.1) -- (-1,-1.1) -- cycle;
  \node at (-1.5,0) {$\alpha$};
  \draw[thick] (-2,0.9) -- (-2.3,0.9);
  \node[rotate=90] at (-2.2,0) {$\cdots$};
  \draw[thick] (-2,-0.9) -- (-2.3,-0.9);
  \draw [decorate,decoration={brace,mirror,amplitude=5pt}]
  (-2.5,1.1) -- (-2.5,-1.1) node[midway,xshift=-11pt]{$m$};
\end{tikzpicture}}
\]
    Using the substitution monoidal structure $(J,\bullet)$, $\W$-operads can be identified with monoid objects in $([\W,\Set],J,\bullet)$.

    While the substitution monoidal structure provides a nice conceptual viewpoint (cf.\ \cref{rmk:Lan-restriction-lax-monoidal}), its definition requires some work and is not necessary for the purposes of this paper.
    Hence, we give an exposition that does not rely on the substitution monoidal structure.
    \lipicsEnd
\end{remark}

\subsection{The Monad Induced by a \texorpdfstring{$\W$}{W}-Operad}\label{subsec:from-operad-to-monad}
In this section, we show that every $\Ctxt$-operad 
$\mnd O$ induces a monad $\Mnd{\mnd O}$ on $\Set$.
Intuitively, $\Mnd{\mnd O}$ is the (finitary) monad corresponding to the algebraic theory expressed by $\mnd O$.
One way of making this precise is to proceed as follows.
Define the notion of \emph{$\mnd O$-algebra} in $\Set$ (as in \cite[Def.~1.1]{Tronin-operads-varieties-polylinear}), show that the forgetful functor $\mathbf{Alg}(\mnd O)\to\Set$ from the category $\mathbf{Alg}(\mnd O)$ of $\mnd O$-algebras is monadic, and then let $\Mnd{\mnd O}$ be the induced monad on $\Set$.

Instead, here we use left Kan extensions 
to define $\Mnd{\mnd O}$. This is useful for describing the canonical distributive laws later.

\begin{definition}[$\Mnd{\mnd O}$]
 Let $\W$ be a verbal category and $\mnd O=(O,\ido,\subst)$ be a $\W$-operad. The monad $\Mnd{\mnd O}$ \emph{induced by} $\mnd O$ is 
\begin{equation}\label{eqn:MndO-data}
 \Mnd{\mnd O}=(\Lan_JO,\eta,\mu),
\end{equation}
using the unit and multiplication described below.\lipicsEnd
\end{definition}

Here, $\Lan_JO$ in \eqref{eqn:MndO-data} is 
the left Kan extension of
 $O$ along the inclusion functor $J$:
 \begin{equation}
 \label{eqn:LanJO}
\begin{tikzpicture}[baseline=-\the\dimexpr\fontdimen22\textfont2\relax ]
      \node(01) at (0,0.5) {$\Ctxt$};
      \node(02) at (1.5,0.5) {$\Set$};
      \node(11) at (0,-0.5) {$\Set$};
      \draw [->] (01) to node[auto,labelsize] {$O$} (02);
      \draw [->] (01) to node[auto,swap,labelsize] {$J$} (11);
      \draw [->] (11) to node[auto,swap,labelsize] {$\Lan_JO$} (02);
      \node[rotate=-90] at (0.5,0.1) {$\Rightarrow$};
\end{tikzpicture}
\end{equation}
Concretely, for each $X\in\Set$, the set $(\Lan_JO)X$ is the coend
\begin{equation}\label{eqn:LanJO-coend}
    \textstyle\int^{\vN n\in \Ctxt} O_n\times X^n
\end{equation}
(cf.\ e.g.\ \cite[Prop.\ 4.1]{Hyland-Power}), which can be constructed as follows. First form the set 
\begin{equation}\label{eqn:LanJO-coproduct}
    \textstyle\sum_{n\in\NN} O_n\times  X^n,
\end{equation}
whose element we write as $\bigl(p,(x_i)_{i\in\vN n}\bigr)$ where $p\in O_n$ and $x_i\in X$ (with $n\in\NN$), or sometimes as $(p,\vec x)$.
The set $(\Lan_JO)X$ is the quotient of the set \cref{eqn:LanJO-coproduct}
under the smallest equivalence relation $\sim$ on \cref{eqn:LanJO-coproduct} such that 
\begin{equation}\label{eqn:sim-generating}
\bigl(p,(x_{\alpha(i)})_{i\in\vN m}\bigr)\sim
\bigl(\alpha p,(x_j)_{j\in\vN n}\bigr)
\end{equation}
for each morphism $\alpha\colon\vN m\to \vN n$ in $\Ctxt$, $p\in O_m$, and $(x_j)_{j\in\vN n}\in X^n$.
(Writing $(x_{\alpha(i)})_{i\in\vN m}$ as $\vec x \alpha$,
\cref{eqn:sim-generating} can be expressed concisely as 
$(p,\vec x\alpha)\sim
(\alpha p,\vec x)$.)
The $\sim$-equivalence class containing $(p,\vec x)$ is denoted by $[p,\vec x]$.
Given a function $f\colon X\to Y$ between sets, the function 
$(\Lan_JO)f\colon (\Lan_JO)X\to (\Lan_JO)Y$
maps each 
$\bigl[p,(x_i)_{i\in \vN n}\bigr]\in (\Lan_JO)X$ to $\bigl[p,(fx_i)_{i\in \vN n}\bigr]$. 

The unit $\eta$ of the monad $\Mnd{\mnd O}$ at $X\in \Set$ is the function 
$\eta_X\colon X\to (\Lan_JO)X$
mapping each $x\in X$ to $[\ido\in O_1,x\in X^1]\in (\Lan_JO)X$.
The naturality of $\eta$ is clear. 

The multiplication $\mu$ of $\Mnd{\mnd O}$ at $X$ is 
$\mu_X\colon (\Lan_JO)(\Lan_JO)X\to (\Lan_JO)X$
given by
\begin{align*}
    \bigl[\,q\in O_n,\,\bigl([p_j\in O_{m_j},\,(x_{i}^{(j)})_{i\in m_j}&\in X^{m_j}]\bigr)_{j\in\vN n}\in \bigl((\Lan_JO)X\bigr)^n\,\bigr]\\
    &\longmapsto \quad \bigl[\,\subst_{n,\vec m}(q,\vec p)\in O_m,\, \bigl((x_i^{(j)})_{i\in m_j}\bigr)_{j\in n}\in X^m\,\bigr],
\end{align*}
where $\vec m=(m_j)_{j\in \vN n}$, $m=\sum_{j\in \vN n}m_j$, and $\vec p=(p_j)_{j\in \vN n}$.
 Well-definedness of $\mu_X$ follows from \cref{eqn:subst-compatibility}, while the naturality of $\mu$ is clear. This definition of $\mu$ crucially uses the $\subst$ operation of $\mnd O$---its depiction in~\cref{eqn:substStringDiagram} with a string diagram should illustrate how $\mu$ works.

The monad axioms for $\Mnd{\mnd O}=(\Lan_JO,\eta,\mu)$ follow from the unit and associativity axioms for the $\Ctxt$-operad $\mnd O$.

It is not hard to see that the Eilenberg--Moore category of the monad $\Mnd{\mnd O}$ coincides with the category $\mathbf{Alg}(\mnd O)$ of $\mnd O$-algebras defined in \cite[Def.~1.1]{Tronin-operads-varieties-polylinear}. 
This shows that $\Mnd{\mnd O}$ can indeed be obtained by
the construction sketched in the beginning of \cref{subsec:from-operad-to-monad}.
One can also see the coincidence by endowing the set $(\Lan_JO)X$ with 
the free $\mnd O$-algebra structure over $X$, as in \cite[Thm.~1.4]{Tronin-operads-varieties-polylinear}.

\begin{definition}[$\W$-operadicity]
 A monad $\mnd S$ (on $\Set$) is \emph{$\Ctxt$-operadic} if there exists a $\Ctxt$-operad $\mnd O$ with $\Mnd{\mnd O}\cong\mnd S$.
 When moreover such a $\Ctxt$-operad $\mnd O$ is unique up to isomorphism, 
 we say that
 $\mnd S$ is \emph{uniquely $\Ctxt$-operadic}.\lipicsEnd
\end{definition}

\begin{remark}\label{W-operadic-monad-intrinsically}
    Intrinsic characterization of $\W$-operadicity is known for some $\W$.
    \begin{bracketenumerate}
        \item See \cite[Def.~3.1 and Thms~10.1 and 2]{Weber-generic} for a characterization of $\Fid$-operadic monads (also known as \emph{strongly analytic monads}). 
        \item See \cite[Def.~3.2 and Thms~10.10 and 11]{Weber-generic} for a characterization of 
        $\Fbij$-operadic monads (also known as \emph{analytic monads}).
        \item See \cite[\S~2.4]{Szawiel-Zawadowski-monads-of-regular-theories} for a characterization of $\Fsurj$-operadic monads (also known as \emph{semi-analytic monads} or \emph{finitary taut monads}).
        \item A monad is $\F$-operadic iff it is finitary (folklore).\lipicsEnd
    \end{bracketenumerate}
\end{remark}

We can simplify the coend formula \cref{eqn:LanJO-coend} in some cases.
\begin{proposition}\label{LanJO-simplification}
\begin{bracketenumerate}
    \item When $\W=\Fid$, \cref{eqn:LanJO-coend} coincides with \cref{eqn:LanJO-coproduct}.
    \item When $\W=\Fbij$, \cref{eqn:LanJO-coend} is $\sum_{n\in \NN}O_n\times_{\Fbij(n,n)}X^n$, where $O_n\times_{\Fbij(n,n)}X^n$ is the quotient of $O_n\times X^n$ with respect to the following equivalence relation \cite{Joyal-analytic-functor}:
    \[
    \bigl\{\, \bigl((p,\vec x\alpha),(\alpha p,\vec x)\bigr)\,\big\vert\, p\in O_n,\, \vec x\in X^n,\, \alpha\in \Fbij(n,n)\,\bigr\}.
    \]
    \item When $\W=\Fsurj$, \cref{eqn:LanJO-coend} is $\sum_{S\in \Pfin X}O_{|S|}$, where $|S|$ is the cardinality of $S$. \qed
\end{bracketenumerate}
\end{proposition}

Here is an observation with potential consequences on distributive laws (see \cref{rmk:choice-of-operad}).

\begin{proposition}\label{uniqueness-of-operads-inducing-monads}
    When 
$\Ctxt$ is $\Fbij$, $\Fsurj$, or $\F$, 
any $\Ctxt$-operadic monad is uniquely $\Ctxt$-operadic. 
    When
$\Ctxt$ is $\Fid$, $\Fmonoinj$, or $\Finj$, 
there exists a $\Ctxt$-operadic monad which is not uniquely $\Ctxt$-operadic.\qed
\end{proposition}

\begin{remark}\label{rmk:Lan-restriction-lax-monoidal}
    Recall \cref{rmk:operads-as-monoid-objects}.
    The left Kan extension and restriction along the inclusion functor $J\colon\W\to\Set$ induces a \emph{lax monoidal} adjunction
\begin{equation}\label{eqn:LanJ-JSet-adjunction}
\begin{tikzpicture}[baseline=-\the\dimexpr\fontdimen22\textfont2\relax ]
      \node(0) at (0,0) {$([\Set,\Set],1,\circ)$};
      \node(1) at (4,0) {$([\Ctxt,\Set],J,\bullet)$,};
      \draw [<-,transform canvas={yshift=.4em}] (0) to node[yshift=-2pt,auto,labelsize] {$\Lan_J$} (1);
      \draw [->,transform canvas={yshift=-.4em}] (0) to node[auto,swap,labelsize] {$[J,\Set]$} (1);
      \node[rotate=90] at (2,0) {$\vdash$};
\end{tikzpicture}
\end{equation}
where the monoidal structure $(1,\circ)$ on $[\Set,\Set]$ is given by composition of endofunctors (cf.\ \cite[\S~5]{Curien-operads-clones}).
It then follows that the adjunction \cref{eqn:LanJ-JSet-adjunction} lifts to the categories of monoid objects:
\begin{equation*}
\begin{tikzpicture}[baseline=-\the\dimexpr\fontdimen22\textfont2\relax ]
      \node(0) at (0,0) {$\Mon([\Set,\Set],1,\circ)$};
      \node(1) at (5.5,0) {$\Mon([\Ctxt,\Set],J,\bullet)$.};
      \draw [<-,transform canvas={yshift=.4em}] (0) to node[yshift=-2pt,auto,labelsize] {$\Mon(\Lan_J)$} (1);
      \draw [->,transform canvas={yshift=-.4em}] (0) to node[auto,swap,labelsize] {$\Mon([J,\Set])$} (1);
      \node[rotate=90] at (2.75,0) {$\vdash$};
\end{tikzpicture}
\end{equation*}
Note that $\Mon([\Set,\Set],1,\circ)$ is the category $\Mndcat{\Set}$ of monads on $\Set$, whereas $\Mon([\Ctxt,\Set],J,\bullet)$ is the category $\Opdcat{\W}$ of $\Ctxt$-operads.
Our construction $\Mnd{-}$ is $\Mon(\Lan_J)$. 
We discuss the functor $\Mon([J,\Set])$ in \cref{sec:monad-to-operad};
see in particular \cref{Mnd-W-Opd-adjunction}.
\lipicsEnd
\end{remark}

\subsection{Change of Base}\label{subsec:from-operad-to-operad}

When one verbal category $\Ctxt$ is a subcategory of another $\Ctxt'$, we obtain two-way constructions between $\Ctxt$-operads and $\Ctxt'$-operads---this is much like the one in \cref{rmk:Lan-restriction-lax-monoidal} between $\Ctxt$-operads and monads. We describe these constructions, together with their consequence on the relationship with monads.

Let $\Ctxt$ and $\Ctxt'$ be
 verbal categories such that $\Ctxt$  is a subcategory of $\Ctxt'$; this fixes the inclusion functor $K\colon\Ctxt\to\Ctxt'$. Firstly, given a $\Ctxt'$-operad ${\mnd O}'$, it is easy to see that the \emph{restriction} of ${\mnd O}'$ along $K$ is a $\Ctxt$-operad. Its underlying functor is given by $\Ctxt\xrightarrow{K}\Ctxt'\xrightarrow{O'}\Set$. 

Conversely, given a $\Ctxt$-operad $\mnd O$, we define its \emph{extension} ${\mnd O}'=(O',\ido^{\mnd O'},\subst^{\mnd O'})$ along $K$ which is a $\Ctxt'$-operad. Its underlying functor is given, much like in \eqref{eqn:LanJO},  by the left Kan extension $O'=\Lan_KO$. Concretely, for any $n'\in \W'$, we have 
\[
O'_{n'}=
(\Lan_KO)_{n'}\cong \textstyle\int^{\vN{n}\in \Ctxt}O_{n}\times \Ctxt'(\vN{n},\vN{n'}),
\]
whose element is denoted by $[\,p\in O_{n},\,\alpha\in \Ctxt'(\vN{n},\vN{n'})\,]$.

The unit $\ido^{\mnd O'}$ for the extension $\mnd O'$ is $[\,\ido^{\mnd O}\in O_1,\,\id_{\vN 1}\in\Ctxt'(\vN 1,\vN 1)\,]\in O'_1$.
For $n'\in \NN$ and $\vec m'=(m'_{i'}\in \NN)_{i'\in n'}$, 
the substitution operation 
\[\subst^{\mnd O'}_{n',\vec m'}\colon O'_{n'}\times \textstyle\prod_{i'\in n'}O'_{m'_{i'}}\to O'_{m'}\] 
for $\mnd O'$ (with $m'=\sum_{i'\in n'} m'_{i'}$) is given by
\begin{align*}
    \bigl(\,[q\in O_{n},\,\beta\in \W'(n,n')],\,
    \bigl([p_{i'}&\in O_{m_{i'}},\,\alpha_{i'}\in \W'(m_{i'},m'_{i'})]\bigr)_{i'\in\vN n'}\,\bigr)\\
    &\longmapsto \quad \bigl[\,\subst^{\mnd O}_{n,\vec m\beta}(q,\vec p\beta)\in O_{m},\, \beta\star\vec\alpha\in \W'(m,m')\,\bigr],
\end{align*}
where $\vec m\beta=(m_{\beta(i)})_{i\in n}$, $\vec p\beta=(p_{\beta(i)})_{i\in n}$, $m=\sum_{i\in n}m_{\beta(i)}$, and $\vec \alpha=(\alpha_{i'})_{i'\in n'}$. The use of $\star$, from \cref{def:operad-of-functions},  is crucial here.

Regarding the relationship to \cref{subsec:from-operad-to-monad}, we have the following.
\begin{proposition}\label{extension-of-operad}
\begin{bracketenumerate}
 \item 
 Let $K\colon\Ctxt\to\Ctxt'$ be an inclusion  of verbal categories and ${\mnd O}'$ be the extension of a $\Ctxt$-operad ${\mnd O}$ along $K$. Then ${\mnd O'}$ is a $\Ctxt'$-operad and 
 $\Mnd{{\mnd O}'}\cong \Mnd{{\mnd O}}$. 
 \item Consequently, any $\Ctxt$-operadic monad is $\Ctxt'$-operadic, too. \qed
\end{bracketenumerate}
\end{proposition}

    From the viewpoint of algebraic theories, a $\W$-operad $\mnd O$ and its \emph{extension} $\mnd O'$ express the same algebraic theory.
    In contrast, the \emph{restriction} construction may modify the theory.

\subsection{Examples}\label{subsec:examples-of-W-operads}
Here we look at examples of $\W$-operads and $\W$-operadic monads. 
In general, for any verbal category $\W$, the functor $\Delta_{\vN 1}^{\W}\colon \W\to\Set$ constant at the terminal object $\vN 1\in\Set$ has a unique structure of a $\W$-operad; we call it the \emph{terminal $\W$-operad} and write $\Delta_{1}^{\W}$ for it.

\begin{example}[$\W=\Fid$]\label{ex:Fid-operadic}
Using \cref{LanJO-simplification}(1), we obtain the following.
\begin{itemize}
 \item The \emph{list monad} $(-)^{*}$  (corresponding to the theory of monoids) is induced by the terminal $\Fid$-operad $\Delta_{1}^{\Fid}$.
 \item The $C$-exception monad $\mnd E^C$ (\cref{ex:exception-monad}) is induced by an $\Fid$-operad; specifically ${\mnd O}=(O,\ido,\subst)$ with $O_{0}=C$, $O_{1}=\{\ido\}$, and $O_{2}=O_{3}=\cdots=\emptyset$. 
 \item The $\mon M$-writer monad $\mnd W^{\mon M}$ (\cref{ex:writer-monad}) is induced by 
an $\Fid$-operad  ${\mnd O}=(O,\ido,\subst)$ with $O_{0}=\emptyset$, $O_{1}=M$ (where $\ido$ is the unit $1\in M$ of the monoid $\mon M$), and $O_{2}=O_{3}=\cdots=\emptyset$. 
\end{itemize}
Hence these monads are $\Fid$-operadic.
See \cite[Chap.~2]{Leinster-higher-operads} for more examples.
It follows from \cref{W-operadic-monad-intrinsically}(1) that
every $\Fid$-operadic monad is \emph{cartesian} in the sense of \cite[Def.~4.1.1]{Leinster-higher-operads}.\lipicsEnd
\end{example}

We revisit the verbal substitution operation $\star$ (\cref{def:operad-of-functions}). 
\begin{proposition}\label{W-operad-of-monoids}
 Let $\W$ be a verbal category and $\mnd O^{\mathrm{mon}}$ be the $\W$-operad obtained from the terminal $\Fid$-operad $\Delta_{1}^{\Fid}$ via extension (\cref{subsec:from-operad-to-operad}). 
\begin{bracketenumerate}
 \item  $\mnd O^{\mathrm{mon}}$ is the \emph{$\W$-operad for monoids}, with the underlying functor $O^{\mathrm{mon}}_{n}=\sum_{n'\in\NN}\W(\vN{n'},\vN n)$.
 \item The substitution operation of $\mnd O^{\mathrm{mon}}$ coincides with the verbal substitution operation $\star$ in $\W$, i.e., $\subst^{\mnd O^{\mathrm{mon}}}(\beta,\vec{\alpha})=\beta\star\vec{\alpha}$. \qed
\end{bracketenumerate}
\end{proposition}

The intuition for \cref{W-operad-of-monoids}(1) is that $\alpha\colon n'\to n$ in $\W$ represents a word of length $n'$ over the alphabet $n=\{0,1,\dotsc,n-1\}$ in the $\W$-substructural context. 

We return to examples of $\W$-operads.

\begin{example}[$\W=\Fmonoinj$]\label{ex:Fmonoinj-operadic}
Whereas not much seems to be known about $\Fmonoinj$-operads, we note the following.
\begin{itemize}
    \item The terminal monad constant at $1\in \Set$ (corresponding to the algebraic theory with a constant symbol $c$ and the axiom $x=c$) is induced by the terminal $\Fmonoinj$-operad $\Delta_1^{\Fmonoinj}$.
    \item The monad $1+(-)^\ast$ for the algebraic theory of \emph{monoids with zero} ($\mon P=(P,1,\cdot,0)$ such that $(P,1,\cdot)$ is a monoid and $0\in P$ satisfies $0\cdot x=0=x\cdot 0$ for all $x\in P$) is $\Fmonoinj$-operadic.
\end{itemize} 
We can also see that these $\Fmonoinj$-operadic monads are \emph{not} $\Fid$-operadic, via the fact that every $\Fid$-operadic monad is cartesian (see \cref{ex:Fid-operadic}):
the terminal monad is not cartesian since its unit is not; 
the monad $1+(-)^\ast$ is not cartesian since its multiplication is not.\lipicsEnd
\end{example}

\begin{example}[$\W=\Fbij$]\label{ex:Fbij-operadic}
The following monads are $\Fbij$-operadic.
\begin{itemize}
    \item The multiset monad $\mnd M$ (\cref{ex:MS}(1)). It is induced by the terminal $\Fbij$-operad $\Delta_1^{\Fbij}$.
    \item More generally, the $\sr N[\mon M]$-multiset monad $\mnd M^{\sr N[\mon M]}$ (\cref{ex:MS}), where $\sr N[\mon M]$ is the \emph{monoid semiring} of a monoid $\mon M$; see below.
\end{itemize}
The monoid semiring construction $\sr N[-]$ (with natural numbers as coefficients) is 
an analogue of the well-known \emph{group ring} construction, and is 
the left adjoint in the following adjunction between the categories $\SRing$ of semirings and $\Mon$ of monoids (folklore): 
\begin{equation}\label{eqn:monoid-semiring}
\begin{tikzpicture}[baseline=-\the\dimexpr\fontdimen22\textfont2\relax ]
      \node(0) at (0,0) {$\SRing$};
      \node(1) at (3,0) {$\Mon$.};
      \draw [<-,transform canvas={yshift=.4em}] (0) to node[auto,labelsize] {$\sr N[-]$} (1);
      \draw [->,transform canvas={yshift=-.4em}] (0) to node[auto,swap,labelsize] {$(-)_\mult$} (1);
      \node[rotate=90] at (1.5,0) {$\vdash$};
\end{tikzpicture}
    \end{equation}
    Here, the right adjoint $(-)_\mult$ sends each semiring $\sr S=(S,0,+,1,\cdot)$ to its multiplicative monoid $\sr S_\mult=(S,1,\cdot)$. 
    
    Given a monoid $\mon M$, the $\Fbij$-operad $\mnd O^{\mon M}$ inducing the $\sr N[\mon M]$-multiset monad $\mnd M^{\sr N[\mon M]}$ can be constructed as follows \cite[Ex.~3]{Tronin-Kopp}.
    Define the functor $O^{\mon M}\colon\Fbij\to \Set$ with $O^{\mon M}_n=M^n$; this assignment is \emph{covariantly} functorial because every morphism in $\Fbij$ is invertible. 
    We set $\ido=1\in M$, and $\subst$ by componentwise multiplication: e.g., 
    \[\subst_{2,(1,2)}\bigl((u,v),((x),(y,z))\bigr)=(ux,vy,vz).\]

In contrast,
the $\sr Z$-multiset monad $\mnd M^{\sr Z}$ (corresponding to the theory of abelian groups) is not $\Fbij$-operadic; 
note that the ring ${\sr Z}$ of integers is not of the form $\sr N[\mon M]$.
For a rigorous proof of this, one can 
use \cref{W-operadic-monad-intrinsically}(2)
and observe that the functor part of the monad $\mnd M^{\sr Z}$ does not
 preserve the following (weak) pullback in $\Set$:
\begin{math}
 \begin{tikzpicture}[baseline=-\the\dimexpr\fontdimen22\textfont2\relax ]
    \node[labelsize](21) at (0,0.3) {$\vN{0}$};
    \node[labelsize](22) at (0.6,0.3) {$\vN{0}$};
    \node[labelsize](11) at (0,-0.3) {$\vN{2}$};
    \node[labelsize](12) at (0.6,-0.3) {$\vN 1$};
    \draw [->] (11) to (12);
    \draw [->] (21) to (11);
    \draw [->] (21) to (22);
    \draw [->] (22) to  (12);
 \end{tikzpicture}
\end{math}.
(In fact, this shows that $\mnd M^{\sr Z}$ is not even $\Fsurj$-operadic, by \cref{W-operadic-monad-intrinsically}(3).)\lipicsEnd
\end{example}

\begin{example}[$\W=\Finj$]\label{ex:Finj-operadic}
Thanks to \cref{extension-of-operad}(2), any monad that is either $\Fmonoinj$-operadic (\cref{ex:Fmonoinj-operadic}) or $\Fbij$-operadic (\cref{ex:Fbij-operadic}) is $\Finj$-operadic. 
Here is an additional class of $\Finj$-operadic monads, combining certain features of the $\Fmonoinj$-operadic monad $1+(-)^\ast$ for monoids with zero and the $\Fbij$-operadic monad $\mnd M^{\sr N[\mon M]}$:
\begin{itemize}
    \item The $\sr N_0[\mon P]$-multiset monad $\mnd M^{\sr N_0[\mon P]}$ is $\Finj$-operadic for any monoid with zero $\mon P$. 
\end{itemize}
Here, $\sr N_0[\mon P]$ is 
the \emph{contracted monoid semiring} 
(cf.\ \cite[\S~5.2]{Clifford-Preston-semigroup-1})
of $\mon P$, which we now explain.
The underlying set of $\sr N_0[\mon P]$ 
is the set $M(P\setminus\{0\})$ of all finite multisets (see \cref{ex:MS}(1)) over the set $P\setminus \{0\}$ of \emph{nonzero} elements of $\mon P$. 
The addition in $\sr N_0[\mon P]$ is given by the usual addition of multisets (so the additive commutative monoid of $\sr N_0[\mon P]$ is the free commutative monoid over $P\setminus\{0\}$), whereas the multiplication in $\sr N_0[\mon P]$ is defined using the multiplication of $\mon P$.
The construction $\sr N_0[-]$ admits a characterization similar to that for $\sr N[-]$ in \eqref{eqn:monoid-semiring}:
it is the left adjoint of the adjunction $\sr N_0[-]\dashv (-)_{\multzero}\colon\SRing\to \Monzero$,
    where $\Monzero$ is the category of monoids with zero and the right adjoint $(-)_\multzero$ maps each semiring $\sr S=(S,0,+,1,\cdot)\in \SRing$ to $\sr S_\multzero=(S,1,\cdot,0)\in \Monzero$.

An $\Finj$-operad $\mnd O^{\mon P}$ inducing the $\sr N_0[\mon P]$-multiset monad $\mnd M^{\sr N_0[\mon P]}$ can be constructed as follows. 
The underlying functor $O^{\mon P}\colon \Finj\to \Set$ is given by $O^{\mon P}_n=P^n$.
Its action on morphisms of $\Finj$ is defined using $0\in P$; for example, if $\alpha\colon \vN 2\to\vN 3$ is given by $\alpha(0)=2$ and $\alpha(1)=0$, then $O^{\mon P}\alpha$ maps $(p_0,p_1)\in O^{\mon P}_2$ to $\alpha(p_0,p_1)= (p_1,0,p_0)\in O^{\mon P}_3$.\footnote{Whereas letting $\alpha(p_0,p_1)=(p_1,1,p_0)$ also defines a functor $\Finj\to\Set$, the axiom \cref{eqn:subst-compatibility} fails. A similar argument shows that \cite[Ex.~2.3]{Tronin-operads-varieties-polylinear} with $\W=\F$ is incorrect.} 
The rest of the operad structure of $\mnd O^{\mon P}$ is similar to that of $\mnd O^{\mon M}$ in \cref{ex:Fbij-operadic}.

The \emph{indexed valuation monad} $\mathbb{IV}$ of \cite{Varacca-Winskel-dist-law}
turns out to be an instance of $\mnd M^{\sr N_0[\mon P]}$, with $\mon P=(\RRnonneg,1,\cdot,0)$.
Although it happens to be $\Fbij$-operadic as well---note that $\sr N_0[(\RRnonneg,1,\cdot,0)]\cong \sr N[(\mathbb{R}_{>0},1,\cdot)]$ as semirings---it can be naturally obtained from the valuation monad $\mnd V=\mnd M^{\sr R_{\geq 0}}$ (\cref{ex:MS}(3)) as its \emph{$\Finj$-operadic refinement}; see \cref{subsec:IV-revisited} for details. \lipicsEnd
\end{example}

\begin{example}[$\W=\Fsurj$]\label{ex:Fsurj-operadic}
Using \cref{LanJO-simplification}(3), we see that 
the following monads are $\Fsurj$-operadic.
\begin{itemize}
    \item The finite powerset monad $\mnd \Pfin$ (\cref{ex:MS}(2)), induced by the terminal $\Fsurj$-operad $\Delta^{\Fsurj}_1$. 
    \item The nonempty finite powerset monad $\mnd \Pfinne$ (\cref{ex:AS}(1)).
 \item 
 The probability distribution monad $\mnd D$ (\cref{ex:AS}(2)).

\end{itemize}
For example, the probability distribution monad $\mnd D$ is induced by the $\Fsurj$-operad $\mnd O$ with $O_n=\{\,p\in Dn\mid \supp p=n\,\}$ (see \cref{ex:MS} for $\supp$).
More generally, for any semiring $\sr S=(S,0,+,1,\cdot)$ such that $0\neq 1$ and that $S\setminus\{0\}$ is closed under $+$ and $\cdot$, the $\sr S$-multiset monad $\mnd M^{\sr S}$ (\cref{ex:MS}) and the affine $\sr S$-multiset monad $\mnd A^{\sr S}$ (\cref{ex:AS}) are $\Fsurj$-operadic.\lipicsEnd
\end{example}

\begin{example}[$\W=\F$]\label{ex:F-operadic-monads}
As mentioned in \cref{W-operadic-monad-intrinsically}(4), every finitary monad is $\F$-operadic. 
Thus all monads corresponding to finitary algebraic theories (such as those for groups, rings, etc.) are $\F$-operadic.
In particular, all monads mentioned in \cref{sec:prelim} are $\F$-operadic,
with the proviso that for the $C$-reader monad $\mnd R^C$ (\cref{ex:reader-monad}) to be $\F$-operadic, $C$ must be a finite set. \lipicsEnd
\end{example}

\section{\texorpdfstring{$\W$}{W}-Commutative Monads}\label{sec:notions-of-commutative-monad}

In this section, we introduce the notion of \emph{$\W$-commutative monad} (on $\Set$) for any verbal category $\W$.
The definition (\cref{def:W-comm-monad}) uses a lax monoidal structure $(\eta_{\vN 1},\comm)$ on 
the functor part
 $T$ of a monad $\mnd T=(T,\eta,\mu)$---the structure exists even when $\mnd T$ is not commutative. We recall it in \cref{subsec:strength-monoidal-str}.
For verbal categories $\W$ in \cref{eqn:Ctxt-logically}, the $\W$-commutativity of monads can be characterized by a suitable combination of known conditions on monads, namely (ordinary) commutativity, affinity, and relevance; see \cref{W-comm-mnd-characterizations}.

\subsection{Strengths and Lax Monoidal Structures}\label{subsec:strength-monoidal-str}
In order to define $\W$-commutative monads, we first recall standard facts about 
\emph{strength} and \emph{lax monoidal structure} on a monad on $\Set$.
See \cref{apx:details-of-subsec-strength-monoidal-str} for details.

\begin{definition}[strength $\str$ \cite{Kock-monads-on-SMCC,Kock-strong-functors-and-monoidal-monads,McDermott-Uustalu}]\label{def:strength}
    Let $T\colon \Set\to \Set$ be a functor. 
    A \emph{left (tensorial) strength} on $T$ is a natural transformation 
    \[
    \str= \bigl(\str_{X,Y}\colon X\times TY\to T(X\times Y)\bigr)_{X,Y\in\Set}
    \]
    making the following diagrams commute for all sets $X$, $Y$, and $Z$.
    \begin{equation*}
        \begin{tikzpicture}[baseline=-\the\dimexpr\fontdimen22\textfont2\relax ]
      \node(01) at (0,0.5) {$\vN 1\times TX$};
      \node(02) at (2,0.5) {$T(\vN 1\times X)$};
      \node(11) at (1,-0.5) {$TX$};
      \draw [->] (01) to node[auto,labelsize] {$\str_{\vN 1,X}$} (02);
      \draw [->] (01) to node[auto,swap,labelsize] {$\cong$} (11);
      \draw [->] (02) to node[auto,labelsize] {$\cong$} (11);
\end{tikzpicture}
\quad
        \begin{tikzpicture}[baseline=-\the\dimexpr\fontdimen22\textfont2\relax ]
      \node(01) at (0,0.5) {$(X\times Y)\times TZ$};
      \node(02) at (7,0.5) {$T\bigl((X\times Y)\times Z\bigr)$};
      \node(11) at (0,-0.5) {$X\times (Y\times TZ)$};
      \node(12) at (3.5,-0.5) {$X\times T(Y\times Z)$};
      \node(13) at (7,-0.5) {$T\bigl(X\times (Y\times Z)\bigr)$};
      \draw [->] (01) to node[auto,labelsize] {$\str_{X\times Y,Z}$} (02);
      \draw [->] (01) to node[auto,swap,labelsize] {$\cong$} (11);
      \draw [->] (11) to node[auto,labelsize] {$1\times \str_{Y,Z}$} (12);
      \draw [->] (12) to node[auto,labelsize] {$\str_{X,Y\times Z}$} (13);
      \draw [->] (02) to node[auto,labelsize] {$\cong$} (13);
\end{tikzpicture}
    \end{equation*}
     
    Dually, a \emph{right (tensorial) strength} on $T$ is a natural transformation  
\begin{math}
     \str'= \bigl(\str'_{X,Y}\colon TX\times Y\to T(X\times Y)\bigr)_{X,Y\in\Set}
\end{math}
    satisfying the dual axioms.\lipicsEnd
\end{definition}

\begin{remark}[see e.g.\ {\cite{McDermott-Uustalu}}]\label{unique-strength}
    Any functor $T\colon\Set\to \Set$ has a unique left (resp.\ right) strength $\str$ (resp.\ $\str'$). 
    For sets $X$ and $Y$, $\tau_{X,Y}\colon X\times TY\to T(X\times Y)$ is defined as follows.
    For each $x\in X$, $\tau_{X,Y}(x,-)\colon TY\to T(X\times Y)$ is $Ti_x$, where $i_x\colon Y\to X\times Y$ is the function mapping each $y\in Y$ to $(x,y)\in X\times Y$.
    (See \cref{unique-strength-apx} for details.)
    \lipicsEnd
\end{remark}

\begin{definition}[lax monoidal structure $\comm$ \cite{Kock-monads-on-SMCC}]\label{def:comm}
    Let $\mnd T=(T,\eta,\mu)$ be a monad, and $\str$ and $\str'$ be the unique left and right strengths on $T$ as in \cref{unique-strength}. 
    Define a family 
    \begin{align*}
        \comm&=\bigl(\comm_{X,Y}\colon TX\times TY\to T(X\times Y)\bigr)_{X,Y\in\Set}
    \end{align*}
    of morphisms as follows.
    \begin{equation*}
\comm_{X,Y}\colon TX\times TY\xrightarrow{\str'_{X,TY}}T(X\times TY)\xrightarrow{T\str_{X,Y}}TT(X\times Y)\xrightarrow{\mu_{X\times Y}}T(X\times Y)\lipicsEnd
    \end{equation*}
\end{definition}

\begin{proposition}\label{strong-monad-monoidal}
    Let $\mnd T=(T,\eta,\mu)$ be a monad.
    Then the family
    $\comm$ in \cref{def:comm} together with $\eta_1\colon 1\to T1$ makes the functor part $T$ of $\mnd T$ into a lax monoidal functor $(T,\eta_1,\comm)\colon (\Set,1,\times)\to (\Set,1,\times)$ \cite[Thm.~2.1]{Kock-monads-on-SMCC}.
    Moreover, the unit $\eta$ and the left strength $\tau$ are monoidal natural transformations.
    \qed
\end{proposition}
\begin{remark}
    Note that in \cref{strong-monad-monoidal} we only claim that the \emph{functor part} $T$ of a monad $\mnd T$ is lax monoidal. 
    In order for the \emph{monad} $\mnd T$ to be lax monoidal (in the sense that it is a monad in the $2$-category of monoidal categories, lax monoidal functors, and monoidal natural transformations), 
    the multiplication $\mu$ has to be monoidal, which is equivalent to $\mnd T$ being commutative in the ordinary sense.\lipicsEnd
\end{remark}

\subsection{\texorpdfstring{$\W$}{W}-Commutativity}
Let $\mnd T$ be a monad (on $\Set$). 
Thanks to \cref{strong-monad-monoidal}, for any $n\in\NN$ and $n$-tuple $\vec X=(X_i)_{i\in\vN n}$ of sets,
we obtain a well-defined morphism 
\begin{equation}\label{eqn:comm-n-ary}
\comm_{\vec X}\colon \textstyle\prod_{i\in\vN n}TX_i\to T\bigl(\textstyle\prod_{i\in\vN n}X_i\bigr)
\end{equation}
by composing various components of $\comm$ appropriately. 
(When $n=0$, \cref{eqn:comm-n-ary} is $\eta_{\vN 1}\colon \vN 1\to T\vN 1$; when $n=1$, \cref{eqn:comm-n-ary} is the identity on $TX_0$.) 
In particular, this means that
we construct the morphisms \cref{eqn:comm-n-ary} so that the following holds.
For any $n\in\NN$, $n$-tuple $\vec X=(X_i)_{i\in\vN n}$ of sets, and $0\leq k\leq n$, let $\vec X'=(X_i)_{i\in \vN k}$ and $\vec X''=(X_{k+i})_{i\in\vN{n-k}}$. Then
\begin{equation}\label{eqn:phi-decomposition}
\begin{tikzpicture}[baseline=-\the\dimexpr\fontdimen22\textfont2\relax ]
      \node(01) at (0,0.9) {$\prod_{i\in\vN k}TX_i\times \prod_{i\in\vN{n-k}}TX_{k+i}$};
      \node(02) at (4,0.9) {$\prod_{i\in\vN n}TX_i$};
      \node(11) at (0,0) {$T\bigl(\prod_{i\in\vN k}X_i\bigr)\times T\bigl(\prod_{i\in\vN{n-k}}X_{k+i}\bigr)$};
      \node(21) at (0,-0.9) {$T\bigl(\prod_{i\in\vN k}X_i\times 
      \prod_{i\in\vN{n-k}}X_{k+i}\bigr)$};
      \node(22) at (4,-0.9) {$T\bigl(\prod_{i\in \vN n}X_i\bigr)$};
      \draw [->] (01) to node[auto,labelsize] {$\cong$} (02);
      \draw [->] (01) to node[auto,swap,labelsize] {$\comm_{\vec X'}\times \comm_{\vec X''}$} (11);
      \draw [->] (11) to node[auto,swap,labelsize] {$\comm_{X',X''}$} (21);
      \draw [->] (02) to node[auto,labelsize] {$\comm_{\vec X}$} (22);
      \draw [->] (21) to node[auto,labelsize] {$\cong$} (22);
\end{tikzpicture}
    \end{equation}
commutes, where $X'=\prod_{i\in\vN k}X_i$ and $X''=\prod_{i\in \vN{n-k}}X_{k+i}$.
As in \cite[\S~3.2]{Dahlqvist-Parlant-Silva-layer-by-layer}, we set 
\begin{equation}\label{eqn:comm-n-power}
    \comm^{(n)}_X\colon (TX)^n\to T(X^n)
\end{equation}
as the special case of \eqref{eqn:comm-n-ary} with $X_0=X_1=\dots=X_{n-1}=X$.
Note in particular that $\comm^{(0)}_X=\eta_{\vN 1}\colon \vN 1\to T\vN 1$ and $\comm^{(1)}_X$ is the identity on $TX$.

\begin{definition}[$\W$-commutative monad]\label{def:W-comm-monad}
    Let $\Ctxt$ be a verbal category.
    A monad $\mnd T=(T,\eta,\mu)$ is \emph{$\Ctxt$-commutative} if, for each morphism $\alpha\colon \vN m\to \vN n$ in $\Ctxt$ and set $X$, the diagram
    \begin{equation}\label{eqn:Ctxt-dist-monad-Ctxt-morphism-power}
\begin{tikzpicture}[baseline=-\the\dimexpr\fontdimen22\textfont2\relax ]
      \node(01) at (0,0.5) {$(TX)^n$};
      \node(02) at (2,0.5) {$T(X^n)$};
      \node(11) at (0,-0.5) {$(TX)^m$};
      \node(12) at (2,-0.5) {$T(X^m)$};
      \draw [->] (01) to node[auto,labelsize] {$\comm_{X}^{(n)}$} (02);
      \draw [->] (01) to node[auto,swap,labelsize] {$(TX)^\alpha$} (11);
      \draw [->] (11) to node[auto,labelsize] {$\comm_{X}^{(m)}$} (12);
      \draw [->] (02) to node[auto,labelsize] {$T(X^\alpha)$} (12);
\end{tikzpicture}
		\end{equation}
        commutes.\lipicsEnd
\end{definition}

\begin{remark}\label{rmk:commutativity-monotonicity}
It follows directly from \cref{def:W-comm-monad}
that $\W$-commutativity of a monad $\mnd T$ is down-closed with respect to $\W$: if ${\mnd T}$ is $\W'$-commutative, then it is $\W$-commutative for any $\W\subseteq \W'$. 
We discussed this monotonicity in \cref{sec:intro} (cf.\ \cref{eq:fourVerbalCatforIntersection}). 
\lipicsEnd
\end{remark}

\begin{remark}\label{rmk:comm-comm'}
    The functor part $T$ of a monad $\mnd T$ admits another canonical lax monoidal structure $(\eta_1,\comm')$, where $\comm'$ is obtained by swapping the order of $\str$ and $\str'$ in the definition of $\comm$ (\cref{def:comm}). We have $\comm=\comm'$ iff $\mnd T$ is commutative in the ordinary sense.
    Using $\comm'$ instead of $\comm$, we obtain the evident analogues of \eqref{eqn:comm-n-ary} and \eqref{eqn:comm-n-power}, and hence of \cref{def:W-comm-monad}. 
    The resulting class of $\W$-commutative monads coincides with the one defined here using $\comm$.\lipicsEnd
\end{remark}

The following result establishes certain ``robustness'' of the notion of $\W$-commutative monad, by giving an alternative characterization. 
\begin{proposition}\label{W-comm-general-characterization}
    Let $\mnd T=(T,\eta,\mu)$ be a monad and $\alpha\colon m\to n$ be a morphism in $\F$. Then the following are equivalent.
    \begin{bracketenumerate}
        \item For each $n$-tuple $\vec X=(X_j)_{j\in n}$ of sets, the following diagram commutes.
        \begin{equation}\label{eqn:Ctxt-dist-monad-Ctxt-morphism-prod}
\begin{tikzpicture}[baseline=-\the\dimexpr\fontdimen22\textfont2\relax ]
      \node(01) at (0,0.5) {$\prod_{j\in\vN n}TX_j$};
      \node(02) at (3,0.5) {$T\bigl(\prod_{j\in\vN n}X_j\bigr)$};
      \node(11) at (0,-0.5) {$\prod_{i\in\vN m}TX_{\alpha(i)}$};
      \node(12) at (3,-0.5) {$T\bigl(\prod_{i\in\vN m}X_{\alpha(i)}\bigr)$};
      \draw [->] (01) to node[auto,labelsize] {$\comm_{\vec X}$} (02);
      \draw [->] (01) to node[auto,swap,labelsize] {$\langle \pi_{\alpha(i)}\rangle_{i\in\vN m}$} (11);
      \draw [->] (11) to node[auto,labelsize] {$\comm_{\vec X\alpha}$} (12);
      \draw [->] (02) to node[auto,labelsize] {$T\langle \pi_{\alpha(i)}\rangle_{i\in \vN m}$} (12);
\end{tikzpicture}
\end{equation}
Here, $\vec X\alpha=(X_{\alpha(i)})_{i\in\vN m}$ and $\pi_{\alpha(i)}$ is the $\alpha(i)$-th projection from the product $\prod_{j\in\vN n}TX_j$ or $\prod_{j\in\vN n}X_j$.
        \item For each set $X$, the diagram \cref{eqn:Ctxt-dist-monad-Ctxt-morphism-power} commutes. \qed
    \end{bracketenumerate}
\end{proposition}

We note that \cref{eqn:Ctxt-dist-monad-Ctxt-morphism-power} is closely related to \emph{residual diagrams} of \cite{Dahlqvist-Parlant-Silva-layer-by-layer,Parlant-thesis}; see \cref{subsec:EMKl}. The diagram 
\cref{eqn:Ctxt-dist-monad-Ctxt-morphism-prod}
also appears in \cite[Def.~2]{Tronin-algebras-over-multicats}.

As already mentioned, 
for some verbal categories $\W$, the $\W$-commutativity of a monad can be characterized by means of known conditions on monads,
such as (ordinary) \emph{commutativity}~\cite{Kock-monads-on-SMCC}, \emph{affinity}~\cite{Kock-bilinearity}, \emph{relevance}~\cite{Jacobs-weakening-contraction}, \emph{$n$-relevance}~\cite{Parlant-Rot-Silva-Bas,Parlant-thesis}, and \emph{hyperaffinity}~\cite{Johnstone-collapsed-toposes,Garner-cartesian-closed-varieties-1}
(see \cref{apx:comm-aff-rel} for details about these conditions).

\begin{proposition}[{cf.\ \cite{Parlant-thesis}}]\label{W-comm-mnd-characterizations}
    \begin{bracketenumerate}
        \item Any monad is $\Fid$-commutative.
        \item A monad is $\Fmonoinj$-commutative iff it is affine.
        \item A monad is $\Fbij$-commutative iff it is commutative.
        \item A monad is $\Finj$-commutative iff it is commutative and affine.
        \item A monad is $\Fsurj$-commutative iff it is commutative and relevant.
        \item A monad is $\F$-commutative iff it is hyperaffine.
        \item For $n>0$, a monad is $\Fsurj^{(n)}$-commutative (\cref{ex:verbal-cats}) iff it is commutative and $n$-relevant. \qed
    \end{bracketenumerate}
\end{proposition}

\subsection{Examples}\label{subsec:examples-of-W-commutative-monads}
Here are examples of $\W$-commutative monads. We describe them via the characterizations of \cref{W-comm-mnd-characterizations}, since examples for those characterizations have been studied well.

\begin{example}[commutative monads]\label{ex:comm-mnd}
\begin{itemize}
    \item
 The $\mon M$-writer monad $\mnd W^{\mon M}$ (\cref{ex:writer-monad}) is commutative
iff the monoid $\mon M$ is commutative.
    \item The $\sr S$-multiset monad $\mnd M^{\sr S}$ (\cref{ex:MS}) is commutative 
iff the semiring $\sr S$ is commutative. In particular, the finite powerset monad $\mnd \Pfin$ is commutative, but is not affine or relevant.
    \item The $C$-exception monad $\mnd E^C$ (\cref{ex:exception-monad}) is commutative iff $|C|\leq 1$. \lipicsEnd
\end{itemize}
\end{example}
\begin{example}[affine monads]\label{ex:aff-mnd}
    The affine $\sr S$-multiset monad $\mnd A^{\sr S}$ (\cref{ex:AS}) is affine for any semiring $\sr S$. In particular, 
    the probability distribution monad $\mnd D$ and the nonempty finite powerset monad $\mnd \Pfinne$ are affine (and commutative). $\mnd D$ and $\mnd \Pfinne$ are not relevant. \lipicsEnd
\end{example}
\begin{example}[relevant monads]\label{ex:rel-mnd}
 The $\mon M$-writer monad $\mnd W^{\mon M}$ (\cref{ex:writer-monad}) is relevant iff the monoid $\mon M$ is idempotent.
 The $C$-exception monad $\mnd E^C$ (\cref{ex:exception-monad}) is relevant for any 
 set $C$.\lipicsEnd
\end{example}
\begin{example}[hyperaffine monads]\label{ex:hyperaff-mnd}
    The $C$-reader monad $\mnd R^C$ (\cref{ex:reader-monad}) is hyperaffine. \lipicsEnd
\end{example}

\section{Canonical Distributive Laws}\label{sec:dist-law}
Having introduced the notions of $\W$-operad (\cref{sec:notions-of-operad}) and $\W$-commutative monad (\cref{sec:notions-of-commutative-monad}) for each verbal category $\W$, we are now ready to present the  construction of the canonical distributive laws.
We first recall the general definition of distributive laws.

\begin{definition}[distributive law \cite{Beck-distributive-laws}]\label{def:dist-law}
    Let $\mnd S = (S,\eta^{\mnd S},\mu^{\mnd S})$ and $\mnd T = (T,\eta^{\mnd T},\mu^{\mnd T})$ be monads. A \emph{distributive law of $\mnd S$ over $\mnd T$} is a natural transformation $\delta\colon ST\to TS$ making the following four diagrams (stating the compatibility of $\delta$ with $\eta^{\mnd S}$, $\mu^{\mnd S}$, $\eta^{\mnd T}$, and $\mu^{\mnd T}$, respectively) commute.
    \[
\begin{tikzpicture}[baseline=-\the\dimexpr\fontdimen22\textfont2\relax ]
      \node(01) at (0.75,0.5) {$T$};
      \node(11) at (0,-0.5) {$ST$};
      \node(12) at (1.5,-0.5) {$TS$};
      \draw [->] (01) to node[auto,swap,labelsize] {$\eta^{\mnd S}T$} (11);
      \draw [->] (11) to node[auto,labelsize] {$\delta$} (12);
      \draw [->] (01) to node[auto,labelsize] {$T\eta^{\mnd S}$} (12);
\end{tikzpicture}\ 
\begin{tikzpicture}[baseline=-\the\dimexpr\fontdimen22\textfont2\relax ]
      \node(01) at (0,0.5) {$SST$};
      \node(02) at (1.5,0.5) {$STS$};
      \node(03) at (3,0.5) {$TSS$};
      \node(11) at (0,-0.5) {$ST$};
      \node(12) at (3,-0.5) {$TS$};
      \draw [->] (01) to node[auto,labelsize] {$S\delta$} (02);
      \draw [->] (02) to node[auto,labelsize] {$\delta S$} (03);
      \draw [->] (01) to node[auto,swap,labelsize] {$\mu^{\mnd S}T$} (11);
      \draw [->] (11) to node[auto,labelsize] {$\delta$} (12);
      \draw [->] (03) to node[auto,labelsize] {$T\mu^{\mnd S}$} (12);
\end{tikzpicture}\ 
\begin{tikzpicture}[baseline=-\the\dimexpr\fontdimen22\textfont2\relax ]
      \node(01) at (0.75,0.5) {$S$};
      \node(11) at (0,-0.5) {$ST$};
      \node(12) at (1.5,-0.5) {$TS$};
      \draw [->] (01) to node[auto,swap,labelsize] {$S\eta^{\mnd T}$} (11);
      \draw [->] (11) to node[auto,labelsize] {$\delta$} (12);
      \draw [->] (01) to node[auto,labelsize] {$\eta^{\mnd T}S$} (12);
\end{tikzpicture}\ 
\begin{tikzpicture}[baseline=-\the\dimexpr\fontdimen22\textfont2\relax ]
      \node(01) at (0,0.5) {$STT$};
      \node(02) at (1.5,0.5) {$TST$};
      \node(03) at (3,0.5) {$TTS$};
      \node(11) at (0,-0.5) {$ST$};
      \node(12) at (3,-0.5) {$TS$};
      \draw [->] (01) to node[auto,labelsize] {$\delta T$} (02);
      \draw [->] (02) to node[auto,labelsize] {$T\delta$} (03);
      \draw [->] (01) to node[auto,swap,labelsize] {$S\mu^{\mnd T}$} (11);
      \draw [->] (11) to node[auto,labelsize] {$\delta$} (12);
      \draw [->] (03) to node[auto,labelsize] {$\mu^{\mnd T}S$} (12);
\end{tikzpicture}
\]
When $\delta$ satisfies the first two axioms above (compatibility of $\delta$ with $\eta^{\mnd S}$ and $\mu^{\mnd S}$) but not necessarily the last two, we say that
$(T,\delta)$ is a \emph{monad functor from $\mnd S$ to $\mnd S$}  \cite{Street-FTM}, or that
$\delta$ is a \emph{distributive law of the monad $\mnd S$ over the endofunctor $T$} (as opposed to the \emph{monad $\mnd T$}). \lipicsEnd
\end{definition}

The following theorem presents our construction of distributive laws.
Note the structure of the result: without additional assumptions, it yields only a distributive law $\delta$ of a monad over an \emph{endofunctor $T$}. To obtain a genuine distributive law of a monad over a \emph{monad $\mnd T$}, one must furthermore impose either condition 1) or 2) in the second paragraph of \cref{thm:main-detailed}.

\begin{theorem}\label{thm:main-detailed}
    Let $\W$ be a verbal category, $\mnd O=(O,\ido,\subst)$ be a $\W$-operad, and $\mnd T=(T,\eta^{\mnd T},\mu^{\mnd T})$ be a $\W$-commutative monad. 
Let
\begin{math}
     \Mnd{\mnd O}=(\Lan_JO,\eta^{\Mnd{\mnd O}},\mu^{\Mnd{\mnd O}})
\end{math}    
denote the $\W$-operadic monad induced by $\mnd O$ (\cref{subsec:from-operad-to-monad}).
    Then the family 
    \begin{equation}\label{eqn:deltaX}
    \delta=\bigl(\delta_X\colon (\Lan_JO)T X\to T(\Lan_JO)X\bigr)_{X\in\Set},
\end{equation}
    of functions defined below (in \cref{subsec:construction-of-dist-law}) gives rise to a distributive law of the monad $\Mnd{\mnd O}$ over the \emph{endofunctor} $T$.

    Moreover, $\delta$ is a distributive law of 
 $\Mnd{\mnd O}$ over the \emph{monad} $\mnd T$---$\delta$ is compatible with $ \mu^{\mnd T}$, in particular---whenever either of the following conditions is satisfied: 1) $\mnd T$ is commutative, or 2) $O_n=0$ for all $n> 1$.
\end{theorem}

\begin{remark}
    When $\W$ is symmetric (see \cref{def:verbal-cat}), the $\W$-commutativity of a monad $\mnd T$ implies the ordinary commutativity of $\mnd T$ (\cref{W-comm-mnd-characterizations}(3)), and hence 
    Cond.~1) of \cref{thm:main-detailed}.
    In contrast, Cond.~2) of \cref{thm:main-detailed} can never be satisfied when the verbal category $\W$ contains $\Fmonoinj$. 
    We use Cond.~2) of \cref{thm:main-detailed} only when $\W=\Fid$ (see \cref{ex:exception-dist-law,ex:writer-dist-law}). 

    To see the necessity of the additional conditions, consider the case where $\W=\Fid$, $\mnd O$ is the terminal $\Fid$-operad (thus $\Mnd{\mnd O}$ is the list monad; see \cref{ex:Fid-operadic}), and $\mnd T$ is the list monad.
    By \cref{W-comm-mnd-characterizations}(1), $\mnd T$ is $\Fid$-commutative. However, it is known that there is no distributive law of the list monad over itself \cite[Ex.\ 4.18]{ZwartM22}. 
    \lipicsEnd
\end{remark}

As announced in \cref{sec:intro}, the additional condition 1) for $\delta$'s compatibility with $\mu^{\mnd T}$ directs our main focus to verbal categories that admit exchange (E), in particular $\W=\Fbij,\Finj,\Fsurj$, and $\F$ (see~\cref{eq:fourVerbalCatforIntersection}). 
In fact, compatibility with $\eta^{\mnd T}$ holds without such a condition; see below.

\subsection{The Construction}\label{subsec:construction-of-dist-law}

Here we define $\delta$ in \cref{eqn:deltaX}, under the situation of the first paragraph of \cref{thm:main-detailed}.
Fix a set $X$.
Since the domain $(\Lan_JO)TX$ of $\delta_X$ is the coend
\begin{math}
\textstyle\int^{\vN n\in \Ctxt} O_n\times (TX)^n
\end{math}
(see \eqref{eqn:LanJO-coend}),
to give the function $\delta_X$ is equivalent to giving a family 
\begin{equation}\label{eqn:deltabar}
\bigl(\bar \delta_{X,n}\colon O_n\times (TX)^n\to T(\Lan_JO)X\bigr)_{\vN n\in\W}
\end{equation}
of functions extranatural in $\vN n\in \W$ (see e.g.\ \cite[\S~IX.6]{MacLane-CWM}).
We define $\bar \delta_{X,n}$ as the composite of
\[
\begin{tikzpicture}[baseline=-\the\dimexpr\fontdimen22\textfont2\relax ]
      \node(1) at (0,0) {$O_n\times (TX)^{n}$};
      \node(2) at (3.5,0) {$O_n\times T(X^{n})$};
      \node(3) at (7,0) {$T(O_n\times X^{n})$};
      \node(4) at (10.5,0) {$T(\Lan_JO)X$,};
      \draw [->] (1) to node[auto,labelsize] {$1\times \comm^{(n)}_{X}$} (2);
      \draw [->] (2) to node[auto,labelsize] {$\str_{O_n,X^n}$} (3);
      \draw [->] (3) to node[auto,labelsize] {$T\kappa_{n}$} (4);
\end{tikzpicture}
\]
where 
1) 	 $\comm^{(n)}_X\colon (TX)^n\to T(X^n)$ was defined in \cref{eqn:comm-n-power}, 
2) $\tau_{O_n,X^n}$ is the left strength (\cref{def:strength}) of $T$, and 
3) $\kappa_n\colon O_n\times X^n\to (\Lan_J O)X$ is the $n$-th coprojection of the coend \eqref{eqn:LanJO-coend}.

We have to show the extranaturality of \eqref{eqn:deltabar},
    that is, that the following diagram commutes for any morphism $\alpha\colon \vN m\to \vN n$ in $\Ctxt$:
    \begin{equation}\label{eqn:dinaturality}
\begin{tikzpicture}[baseline=-\the\dimexpr\fontdimen22\textfont2\relax ]
      \node(01) at (0,0.5) {$O_m\times (TX)^n$};
      \node(02) at (3.5,0.5) {$O_n\times (TX)^n$};
      \node(11) at (0,-0.5) {$O_m\times 
      (TX)^m$};
      \node(12) at (3.5,-0.5) {$T(\Lan_JO)X$.};
      \draw [->] (01) to node[auto,labelsize] {$O\alpha\times 1$} (02);
      \draw [->] (01) to node[auto,swap,labelsize] {$1\times (TX)^\alpha$} (11);
      \draw [->] (11) to node[auto,labelsize] {$\bar\delta_{X,m}$} (12);
      \draw [->] (02) to node[auto,labelsize] {$\bar\delta_{X,n}$} (12);
\end{tikzpicture}
		\end{equation}
    This can be seen by considering the following diagram.
    \begin{equation*}
\begin{tikzpicture}[baseline=-\the\dimexpr\fontdimen22\textfont2\relax ]
      \node(01) at (0,1.5) {$O_m\times (TX)^n$};
      \node(04) at (11.7,1.5) {$O_n\times (TX)^n$};
      \node(12) at (3.9,0.5) {$O_m\times 
      T(X^n)$};
      \node(14) at (11.7,0.5) {$O_n\times T(X^n)$};
      \node(23) at (7.8,-0.5) {$T(O_m\times X^n)$};
      \node(24) at (11.7,-0.5) {$T(O_n\times X^n)$};
      \node(31) at (0,-1.5) {$O_m\times (TX)^m$};
      \node(32) at (3.9,-1.5) {$O_m\times T(X^m)$};
      \node(33) at (7.8,-1.5) {$T(O_m\times X^m)$};
      \node(34) at (11.7,-1.5) {$T(\Lan_JO)X$};
      \draw [->] (01) to node[auto,labelsize] {$O\alpha\times 1$} (04);
      \draw [->] (01) to node[xshift=-5pt,yshift=-3pt,auto,labelsize] {$1\times \comm^{(n)}_X$} (12);
      \draw [->] (01) to node[auto,near start,labelsize] {$1\times (TX)^\alpha$} (31);
      \draw [->] (04) to node[auto,swap,labelsize] {$1\times \comm^{(n)}_X$} (14);
      \draw [->] (12) to node[auto,labelsize] {$O\alpha\times 1$} (14);
      \draw [->] (12) to node[auto,swap,labelsize] {$1\times T(X^\alpha)$} (32);
      \draw [->] (12) to node[xshift=-5pt,yshift=-3pt,auto,labelsize] {$\str_{O_m,X^n}$} (23);
      \draw [->] (31) to node[auto,labelsize] {$1\times \comm^{(m)}_X$} (32);
      \draw [->] (32) to node[auto,labelsize] {$\str_{O_m,X^m}$} (33);
      \draw [->] (33) to node[auto,labelsize] {$T\kappa_m$} (34);
      \draw [->] (23) to node[auto,labelsize] {$T(O\alpha\times 1)$} (24);
      \draw [->] (23) to node[auto,swap,labelsize] {$T(1\times X^\alpha)$} (33);
      \draw [->] (14) to node[auto,swap,labelsize] {$\str_{O_n,X^n}$} (24);
      \draw [->] (24) to node[auto,swap,labelsize] {$T\kappa_n$} (34);
      \node[labelsize] at (5.8,1.1) {(bifunctoriality of $\times$)};
      \node[labelsize] at (8.5,0.1) {(naturality of $\str$)};
      \node[labelsize] at (9.5,-1) {(extranaturality of $\kappa$)};
      \node[labelsize] at (5.4,-0.5) {(naturality of $\str$)};
      \node[labelsize] at (1.95,0) {($\W$-commutativity \cref{eqn:Ctxt-dist-monad-Ctxt-morphism-power})};
\end{tikzpicture}
\end{equation*}
The defining condition \cref{eqn:Ctxt-dist-monad-Ctxt-morphism-power} of $\W$-commutativity is crucial here.

See \cref{proof:thm:main-detailed} for the rest of the proof of \cref{thm:main-detailed}.
We note that the compatibility of $\delta$ with $\eta^{\Mnd{\mnd O}}$, $\mu^{\Mnd{\mnd O}}$, and 
$\eta^{\mnd T}$ holds without Conditions 1) or 2). It is the compatibility of $\delta$ with $\mu^{\mnd T}$ where the additional assumption is needed. 

Some examples are in \cref{subsec:dist-law-ex}.
Other known distributive laws can be seen to be \emph{not} instances of \cref{thm:main-detailed}, see e.g.\ \cite[\S~9]{PirogS17}.

\begin{remark}\label{rmk:choice-of-operad}
    Let $\mnd S$ be a $\W$-operadic monad which is \emph{not} uniquely $\W$-operadic (cf.\ \cref{uniqueness-of-operads-inducing-monads}). 
    Given a $\W$-commutative monad $\mnd T$, the construction of the natural transformation $\delta\colon ST\to TS$
    presented here refers to the choice of a $\W$-operad $\mnd O$ inducing $\mnd S$.
    Hence it is \emph{a priori} possible that $\delta$ varies with this choice.
    At present, we are not aware of such an example.\lipicsEnd
\end{remark}

In contrast, we can show that the choice of $\W$ does not matter, as we now explain.
Let $\W$ and $\W'$ be verbal categories such that $\W\subseteq \W'$, $\mnd O$ be a $\W$-operad, and $\mnd T$ be a $\W'$-commutative monad.
Then $\mnd T$ is $\W$-commutative as well (\cref{rmk:commutativity-monotonicity}). 
On the other hand, we saw in 
\cref{extension-of-operad}(1) that 
$\mnd O$ induces a $\W'$-operad $\mnd O'$ by extension, and that $\mnd O$ and $\mnd O'$ induce isomorphic monads, which we identify and write as $\mnd S$.
Thus we can apply the construction at the beginning of \cref{subsec:construction-of-dist-law} to $(\W,\mnd O,\mnd T)$ as well as to $(\W',\mnd O',\mnd T)$, obtaining natural transformations $\delta$ and $\delta'$ respectively, both of type $ST\to TS$.  

\begin{proposition}\label{invariance-wrt-cob}
    In the above situation, we have $\delta=\delta'$. 
\qed
\end{proposition}

\subsection{Examples}\label{subsec:dist-law-ex}
Here is a concrete description---in algebraic terms, see e.g.~\cite{PirogS17,RossetZHE24}---of the general construction of $\delta\colon ST\to TS$ in \cref{subsec:construction-of-dist-law}. We give it by an example. Consider an  $\mnd S$-term over $\mnd T$-terms over $X$
 \begin{equation}\label{eqn:S-term-of-T-terms}
    p\bigl(\,t_0(v,w),\,t_1(x,y,z)\,\bigr)\in STX,
 \end{equation}
where $p$ is a binary $\mnd S$-term, $t_0$ and $t_1$ are $\mnd T$-terms, and $v,w,x,y,z\in X$. It is carried by $\delta_{X}$ to the following 
 $\mnd T$-term over $\mnd S$-terms over $X$:
 \begin{equation}\label{eqn:T-term-of-S-terms}
        t_0\bigl(\ 
 t_1\bigl(\,p(v,x),\,p(v,y),\,p(v,z)\,\bigr),\ 
 t_1\bigl(\,p(w,x),\,p(w,y),\,p(w,z)\bigr)\ 
\bigr)\in TSX.
 \end{equation}
The term $t_{0}$ appears outside $t_{1}$ because, in the definition of the lax monoidal structure $\comm$ (\cref{def:comm}), the left occurrence of $T$ is pulled out first. When ${\mnd T}$ is commutative—as in most cases of interest—swapping $t_{0}$ and $t_{1}$ in~\cref{eqn:T-term-of-S-terms} (with the corresponding rearrangement of arguments) yields the same element of $TSX$.

 \begin{example}\label{ex:ring-monad}
    Let $\W=\Fbij$, $\mnd S$ be the list monad $(-)^\ast$ (whose operations we write multiplicatively), and $\mnd T$ be the abelian group monad $\mnd M^{\sr Z}$ (written additively). 
    Letting $p(x,y)=xy$, $t_0(v,w)=v+2w$, and $t_1(x,y,z)=x+3y+z$, \cref{eqn:S-term-of-T-terms} becomes $(v+2w)(x+3y+z)$ and \cref{eqn:T-term-of-S-terms} becomes $vx+3vy+vz+2wx+6wy+2wz$.
    Thus \cref{thm:main-detailed} gives rise to the ``archetypal'' distributive law $\delta$
    expressing the ``distribution of multiplication over addition''.
    As is well-known, $\delta$ makes the composite $TS$ into the ring monad \cite[\S~4(1)]{Beck-distributive-laws}.\lipicsEnd
 \end{example}

 \begin{example}\label{ex:MD-Jacobs}
    Let $\W=\Finj$, $\mnd S$ be the multiset monad $\mnd M$ (\cref{ex:Fbij-operadic}), and $\mnd T$ be the probability distribution monad $\mnd D$ (\cref{ex:aff-mnd}).
    \cref{thm:main-detailed} induces the distributive law studied in \cite{Jacobs21,KozenS24}. \lipicsEnd
 \end{example}

 \begin{example}
    Let $\W=\Fsurj$, $\mnd S$ be either of the semigroup, commutative semigroup, or nonunital semilattice monad, and $\mnd T$ be the maybe monad $\vN 1+(-)$ (\cref{ex:comm-mnd,ex:rel-mnd}). 
    \cref{thm:main-detailed} induces the distributive law $S(\vN 1+X)\to \vN 1+SX$ yielding the monoid, commutative monoid, or (unital) semilattice monad as the composite monad.\lipicsEnd
 \end{example}

 \begin{example}
    Let $\W=\F$, $\mnd S$ be any finitary monad on $\Set$, and $\mnd T$ be the $C$-reader monad $\mnd R^C$ (\cref{ex:hyperaff-mnd}). \cref{thm:main-detailed} induces the canonical distributive law $S(X^C)\to (SX)^C$.
    A related distributive law appears in \cite[Prop.~40]{de-Paiva-dialectica}.
    In fact, this construction works for any (not necessarily finitary) monad $\mnd S$ on $\Set$ \cite[Prop.~2.2]{Bunge-composite-tensor}.
    We note that our construction with $\W=\F$
coincides with \cite[Cor.~4.3(c)]{Bunge-composite-tensor}, which is a special case of \cite[Thm.~4.2]{Bunge-composite-tensor}.\lipicsEnd
 \end{example}

 \begin{example}\label{ex:IVAndTsaiEtAlCONCUR25}
    Let $\W=\Finj$, $\mnd S$ be the \emph{indexed valuation monad} $\mathbb{IV}$ of \cite{Varacca-Winskel-dist-law} (which we show is indeed $\Finj$-operadic in \cref{subsec:IV-revisited}), and $\mnd T$ be the nonempty finite powerset monad $\mnd \Pfinne$. 
    \cref{thm:main-detailed} induces the distributive law $IV \Pfinne \to \Pfinne IV$ constructed in \cite{Varacca-Winskel-dist-law} to model combination of probabilistic and nondeterministic choice.

    Letting $\mnd T$ be the probability distribution monad $\mnd D$ instead, we obtain another distributive law $\delta\colon IV D\to D IV$. This $\delta$ turns out to be closely related to a family 
    $\lambda=(\lambda_X\colon DDX\to DDX)_{X\in\Set}$ 
    of morphisms 
    introduced in~\cite{Tsai-et-al-concur} for a new semantics of MDPs.
    The family $\lambda$ does not form a distributive law of $\mnd D$ over itself since it is not natural;
    indeed, no such distributive law exists \cite{ZwartM22}. 

In summary, we observed that the general construction of \cref{thm:main-detailed} yields a distributive law $\delta\colon IV D\to D IV$, which resolves the categorical ill-behavedness of $\lambda\colon DD\to DD$ used in~\cite{Tsai-et-al-concur}. Its consequence in MDP semantics is an exciting direction of future work. The benefit of such ``indexing'' in program semantics is also observed in~\cite{LiellCockS25}. \lipicsEnd
 \end{example}

\cref{ex:exception-dist-law,ex:writer-dist-law} below exploit Cond.~2 of \cref{thm:main-detailed}.

 \begin{example}\label{ex:exception-dist-law}
    Let $\W=\Fid$, $\mnd S$ be the $C$-exception monad $\mnd E^C$ (\cref{ex:Fid-operadic}), and $\mnd T$ be an arbitrary monad on $\Set$. \cref{thm:main-detailed} yields the distributive law 
\begin{math}
     C+TX\to T(C+X)
\end{math}
    induced by the unit of $\mnd T$ \cite[\S~4(2)]{Beck-distributive-laws}.\lipicsEnd
 \end{example}

 \begin{example}\label{ex:writer-dist-law}
 Let $\W=\Fid$, $\mnd S$ be the $\mon M$-writer monad $\mnd W^{\mon M}$ (\cref{ex:Fid-operadic}), and $\mnd T$ be an arbitrary monad on $\Set$.
 \cref{thm:main-detailed} yields the 
 distributive law
\begin{math}
     M\times TX\to T(M\times X)
\end{math}    
given by the left strength $\str$ of $T$ \cite[\S~4(3)]{Beck-distributive-laws}.\lipicsEnd
 \end{example}

 The distributive laws of \cref{ex:exception-dist-law,ex:writer-dist-law} (or more precisely, their generalizations to base categories other than $\Set$) are used in \cite[Cor.~3 and Thm.~12]{Hyland-Plotkin-Power-combining-effects} to model combination of computational effects with exception and output.

\subsection{Eilenberg--Moore Liftings and Kleisli Extensions}\label{subsec:EMKl}
Let $\mnd S$ and $\mnd T$ be monads. 
A distributive law 
$\delta\colon ST\to TS$ is known to correspond to each of the following data (see e.g.\ \cite[Prop.~6]{Garner-Vietoris-monad}).
\begin{description}
    \item[Eilenberg--Moore lifting:] A \emph{lifting} $\overline{\mnd T}$
        of the monad $\mnd T$ to the Eilenberg--Moore category $\Set^{\mnd S}$ of $\mnd S$, as on the left of \cref{eqn:EM-Kl}.
        \item[Kleisli extension:] An \emph{extension} $\underline{\mnd S}$
        of the monad $\mnd S$ to the Kleisli category $\Set_{\mnd T}$ of $\mnd T$, as on the right of \cref{eqn:EM-Kl}.
        \begin{equation}\label{eqn:EM-Kl}
            \begin{tikzpicture}[baseline=-\the\dimexpr\fontdimen22\textfont2\relax ]
      \node(01) at (0,0.4) {$\Set^{\mnd S}$};
      \node(02) at (1.5,0.4) {$\Set^{\mnd S}$};
      \node(11) at (0,-0.4) {$\Set$};
      \node(12) at (1.5,-0.4) {$\Set$};
      \draw [->] (01) to node[auto,labelsize] {$\overline{\mnd T}$} (02);
      \draw [->] (01) to node[auto,swap,labelsize] {$U^{\mnd S}$} (11);
      \draw [->] (11) to node[auto,labelsize] {$\mnd T$} (12);
      \draw [->] (02) to node[auto,labelsize] {$U^{\mnd S}$} (12);
\end{tikzpicture}\qquad\qquad
\begin{tikzpicture}[baseline=-\the\dimexpr\fontdimen22\textfont2\relax ]
      \node(01) at (0,0.4) {$\Set$};
      \node(02) at (1.5,0.4) {$\Set$};
      \node(11) at (0,-0.4) {$\Set_{\mnd T}$};
      \node(12) at (1.5,-0.4) {$\Set_{\mnd T}$};
      \draw [->] (01) to node[auto,labelsize] {$\mnd S$} (02);
      \draw [->] (01) to node[auto,swap,labelsize] {$F_{\mnd T}$} (11);
      \draw [->] (11) to node[auto,labelsize] {$\underline{\mnd S}$} (12);
      \draw [->] (02) to node[auto,labelsize] {$F_{\mnd T}$} (12);
\end{tikzpicture}
        \end{equation}
\end{description}
Here we comment on these aspects. 

There are a few works \cite{Manes-Mulry-monad-compositions-1,Parlant-thesis,Dahlqvist-Parlant-Silva-layer-by-layer} that pursue goals close to ours using the above Eilenberg--Moore characterization. In the latter,
since the functor part $T$ of any monad $\mnd T=(T,\eta,\mu)$ on $\Set$ has the lax monoidal structure $(\eta_{1},\comm)$ (see \cref{strong-monad-monoidal}), any $n$-ary operation $f\colon X^n\to X$ on a set $X$ induces the $n$-ary operation 
\[
(TX)^n\xrightarrow{\comm^{(n)}_X}T(X^n)\xrightarrow{Tf}TX
\]
on $TX$. Thus the key 
question towards the existence of $\overline{\mnd T}$
is the following: does $T$ preserve \emph{equations} used in an equational presentation of the algebraic theory corresponding to $\mnd S$?

In this context, it has been noted that any commutative monad preserves \emph{linear} equations (expressible by $\Fbij$-operads) \cite[Lem.~4.3.3]{Manes-Mulry-monad-compositions-1}, any commutative and affine monad preserves \emph{non-dup} equations (expressible by $\Finj$-operads) \cite[Thm.~9]{Dahlqvist-Parlant-Silva-layer-by-layer}, and any commutative and relevant monad preserves \emph{non-drop} equations (expressible by $\Fsurj$-operads) \cite[Thm.~8]{Dahlqvist-Parlant-Silva-layer-by-layer}. 

These results in~\cite{Dahlqvist-Parlant-Silva-layer-by-layer}  use
what they call \emph{residual diagrams} (\cite[\S~4.1]{Dahlqvist-Parlant-Silva-layer-by-layer}), which look close to our \cref{eqn:Ctxt-dist-monad-Ctxt-morphism-power}. The notion of $\W$-commutative monad (\cref{def:W-comm-monad})  unifies various residual diagrams in our systematic, operadic, and presentation-independent framework parametrized in a verbal category $\W$. 
Residual diagrams are used in~\cite{Parlant-thesis} for concrete applications of  a general and abstract lifting result~\cite[Thm.~3.14]{Parlant-thesis} (which formalizes our  algebraic arguments above).

An interpretation of our results in the second, Kleisli extension characterization, is as follows. 
From now on, let $\W$ be a symmetric verbal category, $\mnd O$ be a $\W$-operad, and $\mnd T$ be a $\W$-commutative monad.
In this situation, the canonical symmetric monoidal structure of the Kleisli category $\Set_{\mnd T}$ has enough structure to accommodate the notion of \emph{$\mnd O$-algebra in $\Set_{\mnd T}$}. 
These $\mnd O$-algebras then correspond to Eilenberg--Moore algebras of the extended monad $\underline{\Mnd{\mnd O}}$ on $\Set_{\mnd T}$, corresponding to our $\delta$ in \cref{thm:main-detailed}.

We sketch a definition of $\mnd O$-algebras in $\Set_{\mnd T}$,
independently from the extended monad $\underline{\Mnd{\mnd O}}$. 
First we need:
\begin{proposition}
    Every set $X$ induces a $\W$-operad $\EEnd_{\mnd T}(X)=(\End_{\mnd T}(X),\ido,\subst)$ with
    \begin{itemize}
        \item $\End_{\mnd T}(X)_n=\Set_{\mnd T}(X^n,X)=\Set(X^n,TX)$;
        \item $\ido=\eta_X\in\Set(X,TX)$; and
        \item given $g\in \Set(X^n,TX)$ and $\vec f=\bigl(f_i\in\Set(X^{m_i},TX)\bigr)_{i\in n}$, $\subst(g,\vec f)\in \Set(X^m,TX)$ (with $m=\sum_{i\in n}m_i$) is the following composite:
        \begin{equation*}
        X^m\cong \textstyle\prod_{i\in n}X^{m_i}\xrightarrow{\prod_{i\in n}f_i}(TX)^n\xrightarrow{\comm_X^{(n)}}T(X^n)\xrightarrow{Tg}TTX\xrightarrow{\mu_X}TX.\hfill\qed
        \end{equation*}
    \end{itemize}
\end{proposition}

We note that the commutativity of $\mnd T$ is used to show the associativity axiom of $\EEnd_{\mnd T}(X)$, whereas the $\W$-commutativity of $\mnd T$ is used to show the compatibility axiom \cref{eqn:subst-compatibility} for $\subst$.
As usual, we define the notion of $\mnd O$-algebra via operad homomorphisms to endomorphism operads (see e.g.\ \cite[\S~6.4]{Leinster-higher-operads}).

\begin{definition}
    An \emph{$\mnd O$-algebra in $\Set_{\mnd T}$} is a pair $(X,\xi)$, where $X$ is a set and $\xi$ is a homomorphism of $\W$-operads from $\mnd O$ to $\EEnd_{\mnd T}(X)$ (meaning a natural transformation $O\to \End_{\mnd T}(X)$ compatible with $\ido$ and $\subst$).
    Similarly, define the notion of \emph{homomorphism of $\mnd O$-algebras} $(X,\xi)\to (Y,\upsilon)$ as a morphism $X\to Y$ in $\Set_{\mnd{T}}$ suitably compatible with the algebra structures $\xi$ and $\upsilon$.\lipicsEnd
 \end{definition}

 \begin{proposition}
     The category of $\mnd O$-algebras in $\Set_{\mnd T}$ and their homomorphisms is isomorphic to the Eilenberg--Moore category of the monad $\underline{\Mnd{\mnd O}}$ extended to the Kleisli category $\Set_{\mnd T}$ via the distributive law $\delta$ of \cref{thm:main-detailed}. \qed
 \end{proposition}

\section{\texorpdfstring{$\W$}{W}-Operadic Refinement of a Monad}\label{sec:monad-to-operad}

Let $\W$ be a verbal category. 
In \cref{subsec:from-operad-to-monad}, we constructed a monad $\Mnd{\mnd O}$ from any $\Ctxt$-operad $\mnd O$. 
In this section, we first construct a $\Ctxt$-operad $\Opd{\W}{\mnd S}$ from any monad $\mnd S$ (\cref{subsec:monad-to-W-operad}; cf.\ \cref{rmk:Lan-restriction-lax-monoidal}).
This provides us with a universal way of turning any monad $\mnd S$ into a $\W$-operadic one, namely $\Mnd{\Opd{\W}{\mnd S}}$, which we call the \emph{$\W$-operadic refinement} of $\mnd S$. 
As explained in \cref{sec:intro}, this construction is useful when an obstacle to the existence of a distributive law $\delta\colon ST\to TS$ of a monad $\mnd S$ over a $\W$-commutative monad $\mnd T$ is the lack of $\W$-operadicity of $\mnd S$.
We illustrate its use in \cref{subsec:IV-revisited,sec:applications}.

\subsection{From Monads to \texorpdfstring{$\W$}{W}-Operads}\label{subsec:monad-to-W-operad}
 
Let $\mnd S=(S,\eta,\mu)$ be a monad (on $\Set$).
The $\W$-operad $\Opd{\W}{\mnd S}$ induced by $\mnd S$ is defined as follows. 
The underlying functor of $\Opd{\W}{\mnd S}$  is the composite $\Ctxt\xrightarrow{J}\Set\xrightarrow{S}\Set$.
The element $\ido\in SJ\vN 1=S\vN 1$ is the one corresponding to $\eta_{\vN 1}\colon \vN 1\to S\vN 1$.
Given $n\in\NN$ and $\vec m=(m_i\in\NN)_{i\in\vN n}$, the function
\[
\subst_{n,\vec m}\colon S\vN n\times \textstyle\prod_{i\in\vN n}S\vN{m_i}\to S\vN m,
\]
where $m=\sum_{i\in\vN n}m_i$, is defined as follows.
For each $q\in S\vN n$ and $\vec p=(p_i\in S\vN{m_i})_{i\in\vN n}$, we set $\subst_{n,\vec m}(q,\vec p)\in S\vN m$ as the element corresponding to the composite
\[
\vN 1\xrightarrow{\nameof q} S\vN n
\xrightarrow{S \bar p} SS\vN m
\xrightarrow{\mu_{\vN m}} S\vN m,
\]
where $\bar p\colon \vN n\to S\vN m$ 
maps each $i\in n$ to $(S\iota_{\vec m,i})(p_i)$
($\iota_{\vec m,i}\colon m_i\to m$ is from  \cref{notation-kappa}).
It is straightforward to check that this defines a $\W$-operad $\Opd{\W}{\mnd S}$.

\begin{example}\label{ex:operads-induced-by-multiset-monad}
    Let $\sr S$ be a semiring. 
    For the $\sr S$-multiset monad $\mnd M^{\sr S}$ (\cref{ex:MS}), we have:
    \begin{itemize}
        \item $\Opd{\Fbij}{\mnd M^{\sr S}}$ is the $\Fbij$-operad $\mnd O^{\sr S_\mult}$ (see \cref{ex:Fbij-operadic}) induced by the multiplicative monoid $\sr S_{\mult}$ of $\sr S$, and
        \item $\Opd{\Finj}{\mnd M^{\sr S}}$ is the $\Finj$-operad $\mnd O^{\sr S_\multzero}$ (see \cref{ex:Finj-operadic}) induced by the multiplicative monoid with zero $\sr S_{\multzero}$ of $\sr S$.\lipicsEnd
    \end{itemize}
\end{example}

\begin{proposition}[{cf.\ \cref{rmk:Lan-restriction-lax-monoidal}}]\label{Mnd-W-Opd-adjunction}
    Let $\W$ be a verbal category. The constructions $\Mnd{-}$ (from \cref{subsec:from-operad-to-monad}) and $\Opd{\W}{-}$ form an adjunction 
\begin{equation}\label{eqn:Mnd-W-Opd-adjunction}
\begin{tikzpicture}[baseline=-\the\dimexpr\fontdimen22\textfont2\relax ]
      \node(0) at (0,0) {$\Mndcat{\Set}$};
      \node(1) at (4,0) {$\Opdcat{\W}$};
      \draw [<-,transform canvas={yshift=.4em}] (0) to node[auto,labelsize] {$\Mnd{-}$} (1);
      \draw [->,transform canvas={yshift=-.4em}] (0) to node[auto,swap,labelsize] {$\Opd{\W}{-}$} (1);
      \node[rotate=90] at (2.1,0) {$\vdash$};
\end{tikzpicture}
    \end{equation}
    between the categories $\Mndcat{\Set}$ of monads on $\Set$ and $\Opdcat{\W}$ of $\W$-operads.\qed
\end{proposition}

\subsection{Turning Any Monad Into a \texorpdfstring{$\W$}{W}-Operadic Monad}\label{subsec:MAMO}
\begin{definition}[$\W$-operadic refinement $\Rf{\W}{\mnd S}$]\label{def:W-operadic-refinement}
    Let $\W$ be a verbal category and $\mnd S$ be a monad. 
 The ($\W$-operadic) monad $\Mnd{\Opd{\W}{\mnd S}}$ is called the \emph{$\W$-operadic refinement} of $\mnd S$ and is denoted by $\Rf{\W}{\mnd S}$.
\lipicsEnd
\end{definition}

In view of \cref{Mnd-W-Opd-adjunction}, $\Rf{\W}{-}$ is the comonad on $\Mndcat{\Set}$ generated by the adjunction \eqref{eqn:Mnd-W-Opd-adjunction}. 
Hence for each $\mnd S\in \Mndcat{\Set}$, there exists a canonical monad morphism $\varepsilon_{\mnd S}\colon \Rf{\W}{\mnd S}\to \mnd S$ (the counit of \cref{eqn:Mnd-W-Opd-adjunction} at $\mnd S$) with the suitable universal property.
In this sense, $\Rf{\W}{\mnd S}$ is the \emph{universal} $\W$-operadic monad induced by $\mnd S$.

The following is an easy consequence of \cref{ex:operads-induced-by-multiset-monad,ex:Fbij-operadic,ex:Finj-operadic}.
\begin{proposition}\label{Fbij-Finj-operadic-refinement-S-multiset}
    Let $\sr S$ be a semiring. 
    \begin{bracketenumerate}
        \item The $\Fbij$-operadic refinement $\Rf{\Fbij}{\mnd M^{\sr S}}$ of the $\sr S$-multiset monad $\mnd M^{\sr S}$ (\cref{ex:MS}) is the $\sr N[\sr S_\mult]$-multiset monad $\mnd M^{\sr N[\sr S_\mult]}$.
        \item The $\Finj$-operadic refinement $\Rf{\Finj}{\mnd M^{\sr S}}$ of the $\sr S$-multiset monad $\mnd M^{\sr S}$ is the $\sr N_0[\sr S_\multzero]$-multiset monad $\mnd M^{\sr N_0[\sr S_\multzero]}$. 
        In particular, the $\Finj$-operadic refinement $\Rf{\Finj}{\mnd\Pfin}$ of the finite powerset monad $\mnd \Pfin$ is the multiset monad $\mnd M$.\qed
    \end{bracketenumerate}
\end{proposition}

For $\W=\Fsurj$ and $\Fbij$, $\W$-operadic refinements have the following concrete presentations.
\begin{proposition}[{cf.\ \cref{LanJO-simplification}(2) and (3)}]\label{prop:operadicRefPres}
    Let $\mnd S$ be a monad.
\begin{bracketenumerate}
 \item 
 The functor part of its $\Fsurj$-operadic refinement $\Rf{\Fsurj}{\mnd S}$ maps $X$ to $\sum_{A\in\Pfin X} SA$.
 \item 
 The functor part of its $\Fbij$-operadic refinement $\Rf{\Fbij}{\mnd S}$ maps  $X$ to
 \begin{math}
  \bigl\{\,(n,\,[p,\vec x]_{\sim_{n}})\,\big|\,n\in\NN, p\in Sn, \vec x\in X^{n}\,\bigr\}
 \end{math}, where 
 $\sim_{n}$  is given by the natural action of bijections $\alpha\colon n\mathrel{\stackrel{
           \raisebox{.5ex}{$\scriptstyle\cong\,$}}{
           \raisebox{0ex}[0ex][0ex]{$\rightarrow$}}} n$. 
\qed
\end{bracketenumerate}
\end{proposition}

The presentation in \cref{prop:operadicRefPres}(2) (for $\W=\Fbij$) is almost the same as the general presentation 
\begin{equation}\label{eq:refinementGenPres}
 \bigl\{\,(n,p,\vec x)\,\big|\,n\in\NN, p\in Sn, \vec x\in X^{n}\,\bigr\}/{\sim},
\qquad\text{where $\sim$ is induced by arrows in $\W$,}
\end{equation}
of 
\begin{math}
  \Rf{\W}{\mnd S}
=
\Mnd{\Opd{\W}{\mnd S}}
= 
 \int^{n\in \W}Sn\times X^{n}
\end{math} for general $\W$ (cf.\ \cref{subsec:from-operad-to-monad} \& \cref{def:W-operadic-refinement}). The difference is that the equivalence  $\sim_{n}$ in \cref{prop:operadicRefPres}(2) is specific to each $n$, while $\sim$ in \cref{eq:refinementGenPres} can relate $(n,p,\vec{x})\sim(n',p',\vec{x'})$ with $n\neq n'$. An operational intuition is discussed later in \cref{ex:PfinWrite}.

\subsection{The Indexed Valuation Monad Revisited}
\label{subsec:IV-revisited}

As mentioned in \cref{sec:intro}, in order to model the combination of probability and nondeterminism, Varacca and Winskel \cite{Varacca-Winskel-dist-law} consider the existence of a distributive law $V\Pfinne\to \Pfinne V$ of the valuation monad $\mnd V$ (\cref{ex:MS}(3)) over the nonempty finite powerset monad $\mnd \Pfinne$ (\cref{ex:AS}(1)).
Such a distributive law does not exist, but they obtain a distributive law $IV\Pfinne\to \Pfinne IV$ of the \emph{indexed valuation monad} $\mathbb{IV}$ over $\mnd \Pfinne$. 
They construct $\mathbb{IV}$ by removing ``problematic'' equational axioms from an equational presentation of $\mathbb{V}$. 

Since $\mnd\Pfinne$ is $\Finj$-commutative (\cref{ex:aff-mnd}), from our point of view, a natural approach is to replace $\mnd V$ with its $\Finj$-refinement $\Rf{\Finj}{\mnd V}$. 

In this case, these two approaches lead to the same result:

\begin{proposition}\label{IV-as-Finj-operadic-refinement}
    \begin{bracketenumerate}
        \item The $\Finj$-operadic refinement $\Rf{\Finj}{\mnd V}$ of the valuation monad $\mnd V$ is the indexed valuation monad $\mathbb{IV}$.
        \item  The distributive law of $\mathbb{IV}$ over $\mnd \Pfinne$ obtained in \cite[\S~4.3]{Varacca-Winskel-dist-law} coincides with the one obtained from \cref{thm:main-detailed}. \qed
    \end{bracketenumerate}
\end{proposition}

The key to the proof of \cref{IV-as-Finj-operadic-refinement}(1) is to observe that the set
 $IV(X)$ 
of finite indexed valuations \cite[Def.~4.3]{Varacca-Winskel-dist-law} 
can be written as the coend 
\begin{math}
     IV(X)=
    \textstyle\int^{n\in \Finj}M^{\sr R_{\geq 0}}n\times X^n
\end{math}.
Props~\ref{IV-as-Finj-operadic-refinement}(1) and \ref{Fbij-Finj-operadic-refinement-S-multiset}(2) yield the following novel characterization of $\mathbb{IV}$.

\begin{corollary}\label{IV-as-S-multiset-monad}
    The indexed valuation monad $\mathbb{IV}$ is isomorphic to the $\sr N_0[(\sr R_{\geq 0})_{\multzero}]$-multiset monad $\mnd M^{\sr N_0[(\sr R_{\geq 0})_{\multzero}]}$.  \qed
\end{corollary}

\begin{remark}
    In general, simply removing ``problematic'' axioms from an equational presentation of an algebraic theory does not give the same result as the $\W$-operadic refinement. 
    (An easy way to see this is to notice that the $\W$-operadic refinement construction is not idempotent in general.)
    In fact, the former construction is presentation-dependent, whereas our $\W$-operadic refinement construction is a universal construction independent of presentations (see \cref{Mnd-W-Opd-adjunction}). 
    We leave it for future work to see whether the similar modification of the theory done in \cite[\S~5]{Dahlqvist-Parlant-Silva-layer-by-layer} can be captured by a suitable $\W$-operadic refinement.\lipicsEnd
\end{remark}

\section{Application: Distributive Laws via Operadic Refinements}\label{sec:applications}
We have discussed examples of $\W$-operadic monads (\cref{subsec:examples-of-W-operads}), $\W$-commutative monads (\cref{subsec:examples-of-W-commutative-monads}), and canonical distributive laws between them (\cref{subsec:dist-law-ex}). Those distributive laws, in particular, can be seen as application of our theory. In this section, we discuss application of our categorical framework as a whole. 

We follow the storyline sketched in \cref{sec:intro}. Recall that
\begin{itemize}
 \item we are mainly interested in four verbal categories $\W=\Fbij,\Finj,\Fsurj$, and $\F$, with the inclusion order shown in~\cref{eq:fourVerbalCatforIntersection};
 \item and that, to get the canonical distributive law (\cref{thm:main-detailed}) to work, we need a nonempty intersection between 1) the upper set of $\mnd S$'s $\W$-operadicity and 2) the lower set of  $\mnd T$'s $\W$-commutativity.
\end{itemize}
The empty intersection of the two sets is a sign that a distributive law may not exist (although a non-canonical distributive law may exist, see e.g.~\cite[\S~9]{PirogS17}). 

  \begin{example}
 \begin{bracketenumerate}
  \item Let      $\mnd S=\mnd D$ (the probability distribution monad) and $\mnd T=\mnd \Pfin$ (the  finite powerset monad). Then $\mnd S=\mnd D$ is not $\Fbij$-operadic (by \cref{W-operadic-monad-intrinsically}(2)). In contrast, the lower set of $\mnd T$'s commutativity is seen to be $\{ \Fbij\}$ by \cref{W-comm-mnd-characterizations}. Thus this is an ``empty intersection'' situation.

 Indeed, this is the combination for which Plotkin's counterexample shows there is no distributive law~\cite{Varacca-Winskel-dist-law}. 
  \item Similarly, one can show that $(\mnd S,\mnd T)=(\mnd \Pfin,\mnd \Pfin)$ is an ``empty intersection'' situation. 
  \lipicsEnd
 \end{bracketenumerate}
 \end{example}

 In such situations with the empty intersection,
 as discussed in~\cref{sec:intro}, we consider using operadic refinement in~\cref{sec:monad-to-operad} to modify $\mnd S$.

 \begin{example}\label{ex:distr-law-via-refinement-PD}
 For the above pairs $(\mnd S,\mnd T)=(\mnd D,\mnd \Pfin),\,(\mnd \Pfin,\mnd \Pfin)$, we replace $\mnd S$ with the $\Fbij$-refinement $\Rf{\Fbij}{\mnd S}$; the refinements are intensional and indexed versions of $\mnd S$ (cf.\ \cref{prop:operadicRefPres}(2)). We obtain a canonical distributive law $\delta\colon \bigl(\Rf{\Fbij}{\mnd S}\bigr) T\to T \bigl(\Rf{\Fbij}{\mnd S}\bigr)$; \cref{thm:main-detailed} guarantees that  $\delta$ is well-defined and satisfies the required axioms.
\lipicsEnd
 \end{example}

\begin{example}\label{ex:PfinWrite}
 Let $\mnd S=\mnd \Pfin$ and $\mnd T=\mnd W^{\mon M}$ (the writer monad) with commutative $\mon M$. Then $\mnd S$ is not $\Fbij$-operadic while $\mnd T$ is only $\Fbij$-commutative in general (\cref{subsec:examples-of-W-commutative-monads}). Therefore we replace $\mnd S$ with $\Rf{\Fbij}{\mnd S}$, and obtain the following canonical distributive law.
\begin{equation}\label{eq:exPfinWrite}
\begin{tikzpicture}[baseline=-\the\dimexpr\fontdimen22\textfont2\relax ]
      \node(0) at (0,0.3) {$\Rf{\Fbij}{\mnd \Pfin} (W^{\mon M} X)$};
      \node(1) at (5,0.3) {$W^{\mon M} (\Rf{\Fbij}{\mnd\Pfin} X)$};
      \node(2) at (0,-0.3) {$\bigl(n,\,\bigl[p,(m_{i},x_{i})_{i\in n}\bigr]_{\sim_{n}}\bigr)$};
      \node(3) at (5,-0.3) {$\bigl(\,\textstyle\prod_{i\in n}m_{i},\, \bigl(n,\,\bigl[p,(x_{i})_{i\in n}\bigr]_{\sim_{n}}\bigr)\,\bigr)$};
      \draw [->] (0) to node[auto,labelsize] {$\delta_X$} (1);
      \draw [|->] (2) to (3);
\end{tikzpicture}
    \end{equation}
Here we used the presentation of $\Rf{\Fbij}{\mnd \Pfin}$ in \cref{prop:operadicRefPres}(2), with $n\in \NN$, $p\in \Pfin n$, $m_{i}\in M$ and $x_{i}\in X$. 
The operational intuition here is that 1) $n$ processes are spawned, 2) Process $i$ (for each $i\in n$) produces output $x_{i}$, writing $m_{i}$ as side effect, and 3) this $m_{i}$ is used, in the combined effect $\prod_{i\in n}m_{i}$, no matter if Process $i$'s output $x_{i}$ is used in the $\mnd \Pfin$-effect context $p\in \Pfin n$. 
See \cref{subsec:apx-Fbij-intensional} for further discussions.
\lipicsEnd
\end{example}

\section{Conclusions and Future Work}
We developed a general theory of monads and distributive laws that is parametrized by a verbal category $\W$, a notion for substructural contexts by Tronin~\cite{Tronin-operads-varieties-polylinear}. The theory is categorical, unifying previous observations e.g.\ in~\cite{Parlant-thesis,Dahlqvist-Parlant-Silva-layer-by-layer} in a presentation-independent manner. We constructed a canonical distributive law $ST\to TS$ that stands in the balance between $\W$-operadicity of $\mnd S$ and $\W$-commutativity of $\mnd T$. 

Extension to base categories other than $\Set$ is an obvious future direction. Accommodating ``modify $\mnd T$'' approaches~\cite{Jacobs-weakening-contraction,CaretteLZ23,Lindner-affine-part} is another. Application to program logics~\cite{LiellCockS25,MatacheS19} and MDP model checking~\cite{Tsai-et-al-concur} is a practically important direction.



\bibliography{mybib}

\appendix

\section{Details of \texorpdfstring{\cref{subsec:strength-monoidal-str}}{§ 4.1}}\label{apx:details-of-subsec-strength-monoidal-str}
\begin{notation}\label{apx:notation-structure-iso}
    The structure isomorphisms of the symmetric monoidal category $(\Set,1,\times)$ 
    are denoted by 
    $\ell_X\colon \vN 1\times X\to X$, $r_X\colon X\to X\times \vN 1$, $a_{X,Y,Z}\colon (X\times Y)\times Z\to X\times (Y\times Z)$, and $c_{X,Y}\colon X\times Y\to Y\times X$.\lipicsEnd
\end{notation}

Using \cref{apx:notation-structure-iso}, the axioms for left strengths $\str= \bigl(\str_{X,Y}\colon X\times TY\to T(X\times Y)\bigr)_{X,Y\in\Set}$ (\cref{def:strength}) can be written more precisely as follows; we number them for future reference.
    \begin{equation}\label{eqn:right-strength-unit}
        \begin{tikzpicture}[baseline=-\the\dimexpr\fontdimen22\textfont2\relax ]
      \node(01) at (0,0.5) {$\vN 1\times TX$};
      \node(02) at (2,0.5) {$T(\vN 1\times X)$};
      \node(11) at (1,-0.5) {$TX$};
      \draw [->] (01) to node[auto,labelsize] {$\str_{\vN 1,X}$} (02);
      \draw [->] (01) to node[auto,swap,labelsize] {$\ell_{TX}$} (11);
      \draw [->] (02) to node[auto,labelsize] {$T\ell_X$} (11);
\end{tikzpicture}
    \end{equation}\begin{equation}\label{eqn:right-strength-tensor}
        \begin{tikzpicture}[baseline=-\the\dimexpr\fontdimen22\textfont2\relax ]
      \node(01) at (0,0.5) {$(X\times Y)\times TZ$};
      \node(02) at (7,0.5) {$T\bigl((X\times Y)\times Z\bigr)$};
      \node(11) at (0,-0.5) {$X\times (Y\times TZ)$};
      \node(12) at (3.5,-0.5) {$X\times T(Y\times Z)$};
      \node(13) at (7,-0.5) {$T\bigl(X\times (Y\times Z)\bigr)$};
      \draw [->] (01) to node[auto,labelsize] {$\str_{X\times Y,Z}$} (02);
      \draw [->] (01) to node[auto,swap,labelsize] {$a_{X,Y,TZ}$} (11);
      \draw [->] (11) to node[auto,labelsize] {$1\times \str_{Y,Z}$} (12);
      \draw [->] (12) to node[auto,labelsize] {$\str_{X,Y\times Z}$} (13);
      \draw [->] (02) to node[auto,labelsize] {$Ta_{X,Y,Z}$} (13);
\end{tikzpicture}
    \end{equation}

The following proposition contains a direct proof of \cref{unique-strength}.
\begin{proposition}\label{unique-strength-apx}
    Any functor $T\colon\Set\to \Set$ has a unique left (resp.\ right) strength $\str$ (resp.\ $\str'$) and these make the following diagram commute for all sets $X$ and $Y$.
    \begin{equation}\label{eqn:symmetry-of-strength}
        \begin{tikzpicture}[baseline=-\the\dimexpr\fontdimen22\textfont2\relax ]
      \node(01) at (0,0.5) {$X\times TY$};
      \node(02) at (3,0.5) {$T(X\times Y)$};
      \node(11) at (0,-0.5) {$TY\times X$};
      \node(12) at (3,-0.5) {$T(Y\times X)$};
      \draw [->] (01) to node[auto,labelsize] {$\str_{X,Y}$} (02);
      \draw [->] (01) to node[auto,swap,labelsize] {$c_{X,TY}$} (11);
      \draw [->] (02) to node[auto,labelsize] {$Tc_{X,Y}$} (12);
      \draw [->] (11) to node[auto,labelsize] {$\str'_{Y,X}$} (12);
\end{tikzpicture}
    \end{equation}
\end{proposition}
\begin{proof}
\label{proof:unique-strength}
    Let $X$ and $Y$ be sets. For each $x\in X$, the commutativity of 
    \[
        \begin{tikzpicture}[baseline=-\the\dimexpr\fontdimen22\textfont2\relax ]
    \node(t) at (1.5,1) {$TY$};
      \node(01) at (0,0) {$\vN 1\times TY$};
      \node(02) at (3,0) {$T(\vN 1\times Y)$};
      \node(11) at (0,-1) {$X\times TY$};
      \node(12) at (3,-1) {$T(X\times Y)$};
      \draw [->] (t) to node[auto,swap,labelsize] {$\ell_{TY}^{-1}$} (01);
      \draw [->] (t) to node[auto,labelsize] {$T\ell_{Y}^{-1}$} (02);
      \draw [->] (01) to node[auto,labelsize] {$\str_{\vN 1,Y}$} (02);
      \draw [->] (01) to node[auto,swap,labelsize] {$\nameof{x}\times 1$} (11);
      \draw [->] (02) to node[auto,labelsize] {$T(\nameof x\times 1)$} (12);
      \draw [->] (11) to node[auto,labelsize] {$\str_{X,Y}$} (12);
      \node[labelsize] at (1.5,0.5) {\cref{eqn:right-strength-unit}};
      \node[labelsize] at (1.5,-0.5) {(naturality of $\str$)};
\end{tikzpicture}
    \]
    implies that we are forced to define, for each $(x,p)\in X\times TY$, $\str_{X,Y}(x,p)\in T(X\times Y)$ as $ (Ti_x)p$, where $i_x\colon Y\to X\times Y$ is the composite 
    \[
    Y\xrightarrow{\ell_Y^{-1}}\vN 1\times Y\xrightarrow{\nameof x\times 1}X\times Y.
    \]
    The rest of the proof is straightforward.
\end{proof}

\begin{proposition}\label{apx:monad-on-Set-is-strong}
    Any monad $\mnd T=(T,\eta,\mu)$ on $\Set$ is strong:
    the unique left strength $\str$ on the functor $T$ given in \cref{unique-strength}
    makes the following diagrams commute for all sets $X$ and $Y$.
\begin{equation}\label{eqn:right-strength-eta}
        \begin{tikzpicture}[baseline=-\the\dimexpr\fontdimen22\textfont2\relax ]
      \node(01) at (0,-0.5) {$X\times TY$};
      \node(02) at (3,-0.5) {$T(X\times Y)$};
      \node(11) at (1.5,0.5) {$X\times Y$};
      \draw [->] (01) to node[auto,labelsize] {$\str_{X,Y}$} (02);
      \draw [<-] (01) to node[auto,labelsize] {$1\times \eta_Y$} (11);
      \draw [<-] (02) to node[auto,swap,labelsize] {$\eta_{X\times Y}$} (11);
\end{tikzpicture}
    \end{equation} 
    \begin{equation}\label{eqn:right-strength-mu}
        \begin{tikzpicture}[baseline=-\the\dimexpr\fontdimen22\textfont2\relax ]
      \node(01) at (0,-0.5) {$X\times TY$};
      \node(02) at (6,-0.5) {$T(X\times Y)$};
      \node(11) at (0,0.5) {$X\times TTY$};
      \node(12) at (3,0.5) {$T(X\times TY)$};
      \node(13) at (6,0.5) {$TT(X\times Y)$};
      \draw [->] (01) to node[auto,labelsize] {$\str_{X,Y}$} (02);
      \draw [<-] (01) to node[auto,labelsize] {$1\times \mu_Y$} (11);
      \draw [->] (11) to node[auto,labelsize] {$\str_{X,TY}$} (12);
      \draw [->] (12) to node[auto,labelsize] {$T\str_{X,Y}$} (13);
      \draw [<-] (02) to node[auto,swap,labelsize] {$\mu_{X\times Y}$} (13);
\end{tikzpicture}
    \end{equation}
    Similar diagrams involving the unique right strength $\str'$ also commute.\qed
\end{proposition}

\section{Commutative, Affine, Relevant, and Hyperaffine Monads}
\label{apx:comm-aff-rel}
Here we give several equivalent definitions of the classes of commutative, affine, relevant, and hyperaffine monads. 
First we recall the definition of the family of morphisms $\comm$ (\cref{def:comm}) and its dual (cf.\ \cref{rmk:comm-comm'}).

\begin{definition}[lax monoidal structures $\comm$ and $\comm'$ \cite{Kock-monads-on-SMCC}]\label{apx-def:comm}
    Let $\mnd T=(T,\eta,\mu)$ be a monad, and $\str$ and $\str'$ be the unique left and right strengths on $T$ as in \cref{unique-strength}. 
    Define families 
    \begin{align*}
        \comm&=\bigl(\comm_{X,Y}\colon TX\times TY\to T(X\times Y)\bigr)_{X,Y\in\Set}\quad\text{and}\\
        \comm'&=\bigl(\comm'_{X,Y}\colon TX\times TY\to T(X\times Y)\bigr)_{X,Y\in\Set}
    \end{align*}
    of morphisms as follows.
    \[
\comm_{X,Y}\colon TX\times TY\xrightarrow{\str'_{X,TY}}T(X\times TY)\xrightarrow{T\str_{X,Y}}TT(X\times Y)\xrightarrow{\mu_{X\times Y}}T(X\times Y)
\]
\[
\comm'_{X,Y}\colon TX\times TY\xrightarrow{\str_{TX,Y}}T(TX\times Y)\xrightarrow{T\str'_{X,Y}}TT(X\times Y)\xrightarrow{\mu_{X\times Y}}T(X\times Y)\lipicsEnd
\]
\end{definition}

\begin{remark}\label{rmk:two-monoidal-str}
    By \cref{strong-monad-monoidal} and its dual, the functor part $T$ of a monad $\mnd T=(T,\eta,\mu)$ admits two lax monoidal structures $(\eta_1,\comm)$ and $(\eta_1,\comm')$.\lipicsEnd
\end{remark}

\subsection{Affine Monads}
\begin{proposition}\label{affine-monad-charcterization}
For a monad $\mnd T=(T,\eta,\mu)$, the following conditions are equivalent. 
\begin{bracketenumerate}
    \item The following diagram commutes for all sets $X$.
    \begin{equation}\label{eqn:affine-weakening}
\begin{tikzpicture}[baseline=-\the\dimexpr\fontdimen22\textfont2\relax ]
      \node(01) at (0,0.5) {$TX$};
      \node(02) at (2,0.5) {$TX$};
      \node(11) at (0,-0.5) {$\vN 1$};
      \node(12) at (2,-0.5) {$T\vN 1$};
      \draw [->] (01) to node[auto,labelsize] {$1$} (02);
      \draw [->] (01) to node[auto,swap,labelsize] {$!_{TX}$} (11);
      \draw [->] (11) to node[auto,labelsize] {$\eta_{\vN 1}$} (12);
      \draw [->] (02) to node[auto,labelsize] {$T!_X$} (12);
\end{tikzpicture}
    \end{equation}
    \item The following diagram commutes for all sets $X$ and $Y$. 
    \begin{equation}\label{eqn:affine-monad-1}
\begin{tikzpicture}[baseline=-\the\dimexpr\fontdimen22\textfont2\relax ]
      \node(01) at (0,0.5) {$TX\times TY$};
      \node(11) at (0,-0.5) {$T(X\times Y)$};
      \node(12) at (3.5,-0.5) {$TX\times TY$};
      \draw [->] (01) to node[auto,labelsize] {$1$} (12);
      \draw [->] (01) to node[auto,swap,labelsize] {$\comm_{X,Y}$} (11);
      \draw [->] (11) to node[auto,swap,labelsize] {$\langle T\pi_0,T\pi_1\rangle$} (12);
\end{tikzpicture}
    \end{equation}
    \item The functor part $T$ preserves the terminal object $\vN 1$.
\end{bracketenumerate}
\end{proposition}
\begin{proof}
    The equivalence of (2) and (3) is proved in \cite[Thm.~2.1]{Kock-bilinearity}.
    Assuming (3), $T\vN 1$ is terminal and hence \cref{eqn:affine-weakening} trivially commutes. 
    Conversely, if \cref{eqn:affine-weakening} commutes for all sets $X$, then in particular it commutes when $X=\vN 1$. 
    This shows that $\eta_{\vN 1}$ is a split epimorphism, whereas it is also a split monomorphism since its domain is the terminal object. Hence we have $T\vN 1\cong\vN 1$. 
\end{proof}

Note that \cref{eqn:affine-weakening} is the instance of \cref{eqn:Ctxt-dist-monad-Ctxt-morphism-power} or \cref{eqn:Ctxt-dist-monad-Ctxt-morphism-prod} with $\alpha={!_0}\colon 0\to 1$.

\begin{definition}[affine monad \cite{Kock-bilinearity}]\label{def:affine}
    We say that a monad is \emph{affine} if it satisfies the (equivalent) conditions of \cref{affine-monad-charcterization}.\lipicsEnd
\end{definition}

    For any monad $\mnd T=(T,\eta,\mu)$, there exists a \emph{cofree} affine monad $\mnd T^\aff = (T^\aff,\eta',\mu')$ over $\mnd T$, called the \emph{affine part} of $\mnd T$ \cite{Lindner-affine-part}
    (cf.\ \cite[Part~C]{Lawvere-functorial-semantics-TAC}).
    For each set $X$, we define the set $T^\aff X$ by the pullback
\[\begin{tikzpicture}[baseline=-\the\dimexpr\fontdimen22\textfont2\relax ]
      \node(01) at (0,0.5) {$T^\aff X$};
      \node(02) at (2,0.5) {$TX$};
      \node(11) at (0,-0.5) {$\vN 1$};
      \node(12) at (2,-0.5) {$T\vN 1$};
      \draw [->] (01) to node[auto,labelsize] {$\iota_X$} (02);
      \draw [->] (01) to (11);
      \draw [->] (11) to node[auto,labelsize] {$\eta_{\vN 1}$} (12);
      \draw [->] (02) to node[auto,labelsize] {$T!_X$} (12);
\end{tikzpicture}\]
    in $\Set$. 
    The rest of the monad structure of $\mnd T^\aff$ is determined by the requirement that $\iota\colon \mnd T^\aff\to \mnd T$ become a monad morphism; see \cite{Lindner-affine-part} for details.
    Note that each $\iota_X$ is an injection, as it is the pullback of an injection $\eta_{\vN 1}$.
    Thus $\mnd T^\aff$ is a submonad of $\mnd T$.
    Clearly $\iota\colon \mnd T^\aff\to \mnd T$ is an isomorphism iff $\mnd T$ is affine.

    For any semiring $\sr S$, the affine part $(\mnd M^{\sr S})^\aff$ of the $\sr S$-multiset monad $\mnd M^{\sr S}$ (\cref{ex:MS}) 
    is the affine $\sr S$-multiset monad $\mnd A^{\sr S}$ (\cref{ex:AS}), as already mentioned.

    It is easy to see that the affine part $\mnd T^\aff$ of a monad $\mnd T=(T,\eta,\mu)$ is isomorphic to the identity monad on $\Set$ iff the natural transformation $\eta$ is \emph{cartesian}, in the sense that each of its naturality squares is a pullback square in $\Set$. 
    Thus for any set $C$ (resp.\ any monoid $\mon M$), the affine part of the $C$-exception monad $\mnd E^C$ (resp.\ the $\mon M$-writer monad $\mnd W^{\mon M}$) is isomorphic to the identity monad on $\Set$.

\subsection{Commutative Monads}
Recall the two lax monoidal structures $(\eta_{\vN 1},\comm)$ and $(\eta_1,\comm')$ on the functor part $T$ of a monad $\mnd T=(T,\eta,\mu)$ (\cref{rmk:two-monoidal-str}).
\begin{proposition}\label{commutative-monad-characterization}
The following conditions on a monad $\mnd T=(T,\eta,\mu)$ are equivalent.
\begin{bracketenumerate}
    \item The two lax monoidal structures $(\eta_{\vN 1},\comm)$ and $(\eta_1,\comm')$ on $T$ coincide.
    That is, the following diagram commutes for all sets $X$ and $Y$.
    \begin{equation}\label{eqn:commutativity}
\begin{tikzpicture}[baseline=-\the\dimexpr\fontdimen22\textfont2\relax ]
      \node(0) at (0,0) {$TX\times TY$};
      \node(01) at (2,0.75) {$T(X\times TY)$};
      \node(02) at (5,0.75) {$TT(X\times Y)$};
      \node(11) at (2,-0.75) {$T(TX\times Y)$};
      \node(12) at (5,-0.75) {$TT(X\times Y)$};
      \node(1) at (7,0) {$T(X\times Y)$};
      \draw [->] (0) to node[auto,labelsize] {$\tau'_{X,TY}$} (01);
      \draw [->] (01) to node[auto,labelsize] {$T\tau_{X,Y}$} (02);
      \draw [->] (02) to node[auto,labelsize] {$\mu_{X\times Y}$} (1);
      \draw [->] (0) to node[auto,swap,labelsize] {$\tau_{TX,Y}$} (11);
      \draw [->] (11) to node[auto,labelsize] {$T\tau'_{X,Y}$} (12);
      \draw [->] (12) to node[auto,swap,labelsize] {$\mu_{X\times Y}$} (1);
\end{tikzpicture}
    \end{equation}
    \item The lax monoidal functor $(T,\eta_{\vN 1},\comm)$ is symmetric. That is, the following diagram commutes for all sets $X$ and $Y$.
    \begin{equation}\label{eqn:phi-comm}
\begin{tikzpicture}[baseline=-\the\dimexpr\fontdimen22\textfont2\relax]
      \node(01) at (0,0.5) {$TX\times TY$};
      \node(02) at (3,0.5) {$T(X\times Y)$};
      \node(11) at (0,-0.5) {$TY\times TX$};
      \node(12) at (3,-0.5) {$T(Y\times X)$};
      \draw [->] (01) to node[auto,labelsize] {$\comm_{X,Y}$} (02);
      \draw [->] (01) to node[auto,swap,labelsize] {$c_{TX,TY}$} (11);
      \draw [->] (11) to node[auto,labelsize] {$\comm_{Y,X}$} (12);
      \draw [->] (02) to node[auto,labelsize] {$Tc_{X,Y}$} (12);
\end{tikzpicture}
    \end{equation}
    \item The multiplication $\mu$ is a monoidal natural transformation with respect to $(\eta_1,\comm)$. That is, the following diagram commutes for all sets $X$ and $Y$.
    \begin{equation}\label{eqn:mu-monoidal}
        \begin{tikzpicture}[baseline=-\the\dimexpr\fontdimen22\textfont2\relax ]
      \node(01) at (0,0.5) {$TTX\times TTY$};
      \node(02) at (3,0.5) {$T(TX\times TY)$};
      \node(03) at (6,0.5) {$TT(X\times Y)$};
      \node(11) at (0,-0.5) {$TX\times TY$};
      \node(12) at (6,-0.5) {$T(X\times Y)$};
      \draw [->] (01) to node[auto,labelsize] {$\comm_{TX,TY}$} (02);
      \draw [->] (02) to node[auto,labelsize] {$T\comm_{X,Y}$} (03);
      \draw [->] (03) to node[auto,labelsize] {$\mu_{X\times Y}$} (12);
      \draw [->] (01) to node[auto,swap,labelsize] {$\mu_X\times \mu_Y$} (11);
      \draw [->] (11) to node[auto,labelsize] {$\comm_{X,Y}$} (12);
\end{tikzpicture}
    \end{equation}
\end{bracketenumerate}
\end{proposition}
\begin{proof}
    The equivalence of (1) and (2) follows from \cref{eqn:symmetry-of-strength}.
    The equivalence of (1) and (3) is proved in \cite[Prop.~1.5]{Kock-bilinearity}.
\end{proof}

Note that \cref{eqn:phi-comm} is the instance of \cref{eqn:Ctxt-dist-monad-Ctxt-morphism-prod} with $\alpha={\mathrm{swap}}\colon 2\to 2$.

\begin{definition}[commutative monad {\cite{Kock-monads-on-SMCC}}]\label{def:commutative-monad}
    We say that a monad is \emph{commutative} if it satisfies the conditions of \cref{commutative-monad-characterization}. \lipicsEnd
\end{definition}

See \cite[Thm.~4.8]{kock-double-dualization} and \cite[Thm.~9.2]{Kock-comm-mnd-distribution} for further equivalent conditions characterizing the commutativity of a monad.

\begin{example}
    Let $\mon M$ be a monoid. For the $\mon M$-writer monad $\mnd W^{\mon M}$ (see \cref{ex:writer-monad}),
    the top composite in \cref{eqn:commutativity} maps $(m,x,n,y)\in M\times X\times M\times Y$ to $(mn,x,y)\in M\times X\times Y$, whereas the bottom composite in \cref{eqn:commutativity} maps $(m,x,n,y)$ to $(nm,x,y)$. 
    Thus $\mnd W^{\mon M}$ is commutative as a monad iff $\mon M$ is commutative as a monoid.\lipicsEnd
\end{example}

\begin{example}
    Let $\sr S$ be a semiring. For the $\sr S$-multiset monad $\mnd M^{\sr S}$ (see \cref{ex:MS}), the top (resp.\ bottom) composite in \cref{eqn:commutativity} maps $(p,q)\in M^{\sr S}X\times M^{\sr S}Y$ to $\bigl((x,y)\mapsto p(x)q(y)\bigr)\in M^{\sr S}(X\times Y)$ (resp.\ $\bigl((x,y)\mapsto q(y)p(x)\bigr)\in M^{\sr S}(X\times Y)$).
    Thus $\mnd M^{\sr S}$ is commutative as a monad iff $\sr S$ is commutative as a semiring.\lipicsEnd
\end{example}

\begin{example}
    Given a set $C$, the $C$-exception monad $\mnd E^C$ (see \cref{ex:exception-monad}) is commutative iff $|C|\leq 1$.\lipicsEnd
\end{example}

\begin{example}
    Any submonad of a commutative monad on $\Set$ is commutative. In particular, the affine part of a commutative monad on $\Set$ is commutative.
    Thus both the probability distribution monad $\mnd D$ and the nonempty finite powerset monad $\PPfinne$ are affine and commutative.\lipicsEnd
\end{example}

\subsection{Relevant Monads}
\begin{proposition}\label{relevant-monad-characterization}
For a monad $\mnd T=(T,\eta,\mu)$, the following conditions are equivalent. 
\begin{bracketenumerate}
    \item The following diagram commutes for all sets $X$.
        \begin{equation}\label{eqn:relevant-monad-characterization}
\begin{tikzpicture}[baseline=-\the\dimexpr\fontdimen22\textfont2\relax ]
      \node(01) at (0,0.5) {$TX$};
      \node(02) at (3,0.5) {$TX$};
      \node(11) at (0,-0.5) {$TX\times TX$};
      \node(12) at (3,-0.5) {$T(X\times X)$};
      \draw [->] (01) to node[auto,labelsize] {$1$} (02);
      \draw [->] (01) to node[auto,swap,labelsize] {$\Delta_{TX}$} (11);
      \draw [->] (11) to node[auto,labelsize] {$\comm_{X,X}$} (12);
      \draw [->] (02) to node[auto,labelsize] {$T\Delta_X$} (12);
\end{tikzpicture}
    \end{equation}
    \item The following diagram commutes for all sets $X$ and $Y$.
    \begin{equation}\label{eqn:relevant-monad}
\begin{tikzpicture}[baseline=-\the\dimexpr\fontdimen22\textfont2\relax ]
      \node(01) at (0,0.5) {$T(X\times Y)$};
      \node(11) at (0,-0.5) {$TX\times TY$};
      \node(12) at (3,-0.5) {$T(X\times Y)$};
      \draw [->] (01) to node[auto,labelsize] {$1$} (12);
      \draw [->] (01) to node[auto,swap,labelsize] {$\langle T\pi_0,T\pi_1\rangle$} (11);
      \draw [->] (11) to node[auto,swap,labelsize] {$\comm_{X,Y}$} (12);
\end{tikzpicture}
    \end{equation}
    \item (2) with $X=Y$.
\end{bracketenumerate}
\end{proposition}
\begin{proof}
    The equivalence of (1) and (2) is stated and proved in {\cite[Prop.~2.2]{Kock-bilinearity}}; the proof of (2) $\Rightarrow$ (1) there actually uses the special case (3) of (2) only.
\end{proof}

Note that \cref{eqn:relevant-monad-characterization} is the instance of \cref{eqn:Ctxt-dist-monad-Ctxt-morphism-power} or \cref{eqn:Ctxt-dist-monad-Ctxt-morphism-prod} with $\alpha={!_2}\colon 2\to 1$.

\begin{definition}[{relevant monad \cite{Jacobs-weakening-contraction}}]
    We say that a monad is \emph{relevant} if it satisfies the conditions of \cref{relevant-monad-characterization}.\lipicsEnd
\end{definition}

See \cref{ex:rel-mnd} for examples of relevant monads.

\subsection{Hyperaffine Monads}
\begin{proposition}\label{hyperaffine-monad-characterization}
    Let $\mnd T=(T,\eta,\mu)$ be a monad. Then the following conditions are equivalent. 
    \begin{bracketenumerate}
        \item $\mnd T$ is affine and relevant.
        \item $\mnd T$ is commutative, affine, and relevant.
        \item The functor part $T$ preserves finite products.
        \item The functor part $T$ preserves powers by $\vN 0$ and $\vN 2$.
    \end{bracketenumerate}
\end{proposition}
\begin{proof}
    The equivalence of (1--3) follows from \cite[Thm.~2.5]{Kock-bilinearity}.
    Trivially (3) implies (4). To see that (4) implies (1), observe that the preservation of powers by $\vN 0$ is equivalent to the affinity of $\mnd T$ (see \cref{affine-monad-charcterization}(3)), and under the affinity (see \cref{affine-monad-charcterization}(2)), preservation of powers by $\vN 2$ is equivalent to \cref{relevant-monad-characterization}(3), which is the relevance of $\mnd T$.
\end{proof}

\begin{definition}[hyperaffine monad \cite{Johnstone-collapsed-toposes,Garner-cartesian-closed-varieties-1}]
    We say that a monad is \emph{hyperaffine} if it satisfies the conditions of \cref{hyperaffine-monad-characterization}.\lipicsEnd
\end{definition}

\begin{remark}
    Hyperaffine monads are called \emph{cartesian closed} monads in \cite[Def.~2.7]{Kock-bilinearity}, since the Eilenberg--Moore category of a hyperaffine monad is cartesian closed \cite[Thm.~3.2]{Kock-bilinearity}. 
    However, there are non-hyperaffine monads whose Eilenberg--Moore categories are cartesian closed (such as $\mnd W^{\mon M}$ of \cref{ex:writer-monad} with non-trivial $\mon M$).
    See \cite[Thm.~1.2]{Johnstone-collapsed-toposes} and \cite[Thm.~5.5]{Garner-cartesian-closed-varieties-1} for characterizations of monads whose Eilenberg--Moore categories are cartesian closed.\lipicsEnd
\end{remark}

For any set $C$, the functor part $(-)^C\colon \Set\to \Set$ of the $C$-reader monad $\mnd R^C$ (\cref{ex:reader-monad}) preserves finite products (in fact, all limits). Thus $\mnd R^C$ is a hyperaffine monad (\cref{ex:hyperaff-mnd}; see also \cite[Lem.~4.2]{Johnstone-collapsed-toposes}).
When $|C|=2$, the Eilenberg--Moore algebras of $\mnd R^C$ are known as the \emph{rectangular bands} \cite[p.~25]{Clifford-Preston-semigroup-1}. 
See \cite[Prop.~4.2]{Garner-cartesian-closed-varieties-1} for a complete classification of hyperaffine monads.

\section{\texorpdfstring{$\Rf{\Fbij}{\mnd S}$}{Rf(S)} as an Intensional Refinement of \texorpdfstring{$\mnd S$}{S}}
\label{subsec:apx-Fbij-intensional}
Here, we continue the discussion in \cref{ex:PfinWrite}, and argue that 
  $\Rf{\Fbij}{\mnd S}$ is an intensional refinement of $\mnd S$ that is sensitive to the number of spawned processes.

The operational intuition of~\cref{eq:exPfinWrite} is as follows. An element 
\begin{math}
  \bigl(n,\,\bigl[p,(m_{i},x_{i})_{i\in n}\bigr]_{\sim_{n}}\bigr)
\end{math}
represents 1) spawning $n$ processes, 2) letting  Process $i$ (for each $i\in n$) conduct computation with $\mnd W^{\mon M}$-effect, 3) observing Process $i$ to write $m_{i}\in M$ and produce output $x_{i}\in X$, and 4) organizing their output $x_{0},\dotsc,x_{n-1}$ in the $\mnd \Pfin$-context $p=\{i_{0}, \dotsc, i_{k-1}\}\in \Pfin n$ (hence $i_{0}, \dotsc, i_{k-1}\in n$), finally obtaining $\{x_{i_{0}}, \dotsc, x_{i_{k-1}}\}$ as an output. 

Some remarks are in order.
\begin{itemize}
 \item In~\cref{eq:exPfinWrite}, we distinguish different numbers of processes, since $\sim_{n}$ is specific to $n$. This is not the case \emph{without} the $\Fbij$-refinement of $\mnd \Pfin$---in general, $\sim$ 
in~\cref{eq:refinementGenPres}
can relate different numbers of processes. 
 \item The $\mnd W^{\mon M}$-effect $m_{i}\in M$ generated by Process $i$ is never thrown away, even if Process $i$'s output $x_{i}$ is not used in the $\mnd \Pfin$-effect context  $p\in \Pfin n$ (that is, even if $i\not\in p$). This is seen in the occurrence of $\prod_{i\in n}m_{i}$ in~\cref{eq:exPfinWrite}. 
\end{itemize}
In this way, the $\Fbij$-refinement  $\Rf{\Fbij}{\mnd S}$ of $\mnd S$ can be seen as an intensional refinement of the computational effect modeled by $\mnd S$, where the refined effect is conscious about the number of spawned processes.

\section{Proofs}
\label{sec:apx-proof}

\subsection{Proof of \texorpdfstring{\cref{LanJO-simplification}}{Prop. 16}}
    We only show \cref{LanJO-simplification}(3). This is a consequence of the fact that for each $X\in \Set$, we have 
    \begin{equation}\label{eqn:J-left-multiadj}
    \Set(Jn,X)\cong \textstyle\sum_{S\in\Pfin X}\Fsurj(n,|S|)
    \end{equation}
    naturally in $n\in \Fsurj$. (This shows in particular that $J\colon \Fsurj\to \Set$ is a \emph{left multi-adjoint}.)
    To see that \cref{eqn:J-left-multiadj} implies \cref{LanJO-simplification}(3),
    recall that $(\Lan_JO)X$ is the weighted colimit $\Set(J-,X)\star O$ of $O\colon \Fsurj\to \Set$ with respect to the weight $\Set(J-,X)\colon \Fsurj^\op\to \Set$ \cite[(4.18)]{Kelly-BCECT}, and proceed as follows:
    \begin{flalign*}
        &&\Set(J-,X)\star O&\cong\bigl(\textstyle \sum_{S\in\Pfin X}\Fsurj(-,|S|)\bigr) \star O&\text{(by \cref{eqn:J-left-multiadj})}\\
        &&&\cong \textstyle \sum_{S\in\Pfin X}\bigl(\Fsurj(-,|S|) \star O\bigr)&\text{(by \cite[(3.23)]{Kelly-BCECT})}\\
        &&&\cong \textstyle \sum_{S\in\Pfin X}O_{|S|}.&\text{(by \cite[(3.10)]{Kelly-BCECT})}
    \end{flalign*}\qed

\subsection{Proof of \texorpdfstring{\cref{uniqueness-of-operads-inducing-monads}}{Prop. 17}}
\label{proof:uniqueness-of-operads-inducing-monads}
    First suppose that $\Ctxt$ is either $\Fbij$, $\Fsurj$, or $\F$.
    In these cases, 
    the functor $\Lan_J\colon [\Ctxt,\Set]\to[\Set,\Set]$ is faithful and restricts to a fully faithful functor $[\Ctxt,\Set]_{\iso}\to[\Set,\Set]_\iso$ between the core groupoids. (The \emph{core groupoid} $\cat C_\iso$ of a category $\cat C$ is the wide subcategory of $\cat C$ consisting of all isomorphisms in $\cat C$.) 
    These claims follow from \cite[Appendice]{Joyal-analytic-functor} 
    when $\Ctxt=\Fbij$, and from 
    \cite[Thm.~2.2]{Szawiel-Zawadowski-monads-of-regular-theories} when $\Ctxt=\Fsurj$;
    of course, when $\Ctxt=\F$, $J\colon \F\to \Set$ is fully faithful and hence so is $\Lan_J\colon [\F,\Set]\to[\Set,\Set]$.
    The first statement follows from this and the fact that $\Lan_J$, being the left adjoint of the lax monoidal adjunction \cref{eqn:LanJ-JSet-adjunction}, is strong monoidal \cite{Kelly-doctrinal}.
    To see this, suppose that we have $\Ctxt$-operads $\mnd O$ and $\mnd O'$ and an isomorphism $\phi\colon \Mnd{\mnd O}\to \Mnd{\mnd O'}$ of monads.
    Then there exists a unique isomorphism $\chi\colon O\to O'$ inducing $\phi\colon \Lan_JO\to \Lan_JO'$.
    Since $\Lan_J$ is strong monoidal, this $\chi$ is an isomorphism $\chi\colon\mnd O\to \mnd O'$ of $\Ctxt$-operads.
    
    Next suppose that $\Ctxt$ is either $\Fmonoinj$ or $\Finj$. 
    For any monoid $\mon M=(M,1,\cdot)$, we construct a $\Ctxt$-operad $\mnd O^{\mon M}$ as follows. Its underlying functor $O^{\mon M}\colon \Ctxt\to \Set$ is determined by
    \[
    O^{\mon M}_n=\begin{cases}
        \vN 0 &\text{if $n=0$,}\\
        M     & \text{if $n=1$, and}\\
        \vN 1 &\text{otherwise.}
    \end{cases}
    \]
    We set $\ido=1\in O^{\mon M}_1$. The morphisms $\subst_{n,\vec m}\colon O^{\mon M}_n\times \prod_{i\in\vN n} O^{\mon M}_{m_i}\to O^{\mon M}_m$ are almost always trivial; the only non-trivial case is where $n=1$ and $m_0=1$, and in this case we set it to be the multiplication $\cdot\colon M\times M\to M$ of $\mon M$. It is straightforward to see that this defines a $\Ctxt$-operad.
    Now, the functor  $\Lan_JO^{\mon M}\colon \Set\to \Set$ is given by 
    \[
    (\Lan_JO^{\mon M})X=
    \begin{cases}
        \vN 0 &\text{if $X$ is empty, and}\\
        \vN 1 &\text{otherwise;}
    \end{cases}
    \]
    This determines the monad $\Mnd{\mnd O^{\mon M}}$ uniquely.
    Thus $\Mnd{\mnd O^{\mon M}}$ does not depend on $\mon M$.
    This also shows that $\Lan_J\colon [\Ctxt,\Set]\to[\Ctxt,\Set]$ is not faithful when $\Ctxt$ is $\Fmonoinj$ or $\Finj$.\footnote{This shows that the claim in \cite[\S~5]{Curien-operads-clones} that $\Lan_J\colon[\Ctxt,\Set]\to[\Set,\Set]$ is faithful when $\Ctxt$ is any of the six categories in \cref{eqn:Ctxt-logically}, is incorrect.}

    Finally, see \cite{Leinster-are-operads-algebraic-theories} for a pair of non-isomorphic $\Fid$-operads inducing the same monad on $\Set$. \qed

\begin{remark}
    From a logical perspective, one may view the $\Finj$-operad $\mnd O^{\mon M}$ constructed in the above proof as showing the incompleteness (with respect to $\Set$-models) of ``equational logic without contraction'', as we now informally explain; a formal treatment is left for future work.
    
    Elements of $O^{\mon M}_1=M$ can be regarded as terms (in context) $x_0\vdash m\cdot x_0$ with $m\in M$, and if $m$ and $n$ are distinct elements of $M$, then the terms $x_0\vdash m\cdot x_0$ and $x_0\vdash n\cdot x_0$ represent distinct elements of $O^{\mon M}_1=M$.
    That is, we do not have 
    \begin{equation}\label{eqn:OM-unary-terms}
        x_0\vdash m\cdot x_0= n\cdot x_0
    \end{equation}
    unless $m=n$.    
    (The term $x_0\vdash x_0$ is identified with the term $x_0\vdash 1\cdot x_0$, where $1$ is the unit element of $\mon M$.)
    
    On the other hand, elements of $O^{\mon M}_2$ are represented by terms of the form $x_0,x_1\vdash m\cdot x_0$ or $x_0,x_1\vdash m\cdot x_1$ with $m\in M$. Now, since $O^{\mon M}_2$ is a singleton, these terms are all identified up to provable equality. 
    Hence we have e.g.
    \begin{equation}\label{eqn:OM-binary-terms}
        x_0,x_1\vdash m\cdot x_0=n\cdot x_1, \quad
        x_0,x_1\vdash m\cdot x_0=n\cdot x_0, \quad
        \text{and even}
        \quad
        x_0,x_1\vdash x_0= x_1.
    \end{equation}

    However, in $\Set$ (or more generally in any cartesian monoidal category), \eqref{eqn:OM-binary-terms} implies \eqref{eqn:OM-unary-terms}.\lipicsEnd
\end{remark}

\subsection{Proof of \texorpdfstring{\cref{strong-monad-monoidal}}{Prop. 30}}

For later reference, we note that the statement that $\eta$ and $\tau$ are monoidal includes that the following diagrams commute, respectively.
    \begin{equation}\label{eqn:eta-monoidal}
\begin{tikzpicture}[baseline=-\the\dimexpr\fontdimen22\textfont2\relax ]
      \node(01) at (1.5,0.5) {$X\times Y$};
      \node(11) at (0,-0.5) {$TX\times TY$};
      \node(12) at (3,-0.5) {$T(X\times Y)$};
      \draw [->] (01) to node[auto,labelsize] {$\eta_{X\times Y}$} (12);
      \draw [->] (01) to node[auto,swap,labelsize] {$\eta_X\times \eta_Y$} (11);
      \draw [->] (11) to node[auto,labelsize] {$\comm_{X,Y}$} (12);
\end{tikzpicture}
    \end{equation}
        \begin{equation}\label{eqn:str-monoidal}
\adjustbox{scale=0.85}{\begin{tikzpicture}[baseline=-\the\dimexpr\fontdimen22\textfont2\relax ]
      \node(01) at (0,0.5) {$(X\times TY)\times (X'\times TY')$};
      \node(02) at (5.5,0.5) {$(X\times X')\times (TY\times TY')$};
      \node(03) at (11,0.5) {$(X\times X')\times T(Y\times Y')$};
      \node(11) at (0,-0.5) {$T(X\times Y)\times T(X'\times Y')$};
      \node(12) at (5.5,-0.5) {$T\bigl((X\times Y)\times(X'\times Y')\bigr)$};
      \node(13) at (11,-0.5) {$T\bigl((X\times X')\times(Y\times Y')\bigr)$};
      \draw [->] (01) to node[auto,labelsize] {$\cong$} (02);
      \draw [->] (02) to node[auto,labelsize] {$1\times \comm_{Y,Y'}$} (03);
      \draw [->] (03) to node[auto,labelsize] {$\str_{X\times X',Y\times Y'}$} (13);
      \draw [->] (01) to node[auto,swap,labelsize] {$\str_{X,Y}\times \str_{X',Y'}$} (11);
      \draw [->] (11) to node[auto,labelsize] {$\comm_{X\times Y,X'\times Y'}$} (12);
      \draw [->] (12) to node[auto,labelsize] {$\cong$} (13);
\end{tikzpicture}}
    \end{equation}
    Here, the unnamed isomorphisms in \cref{eqn:str-monoidal} are the suitable composites of the structure isomorphisms for the symmetric monoidal category $(\Set,\vN 1,\times)$.
    
    The fact that $(T,\eta_1,\comm)$ is 
    a lax monoidal functor is well-known \cite[Thm.~2.1]{Kock-monads-on-SMCC}. 
    \cref{eqn:eta-monoidal} follows from \cite[proof of Thm.~3.2]{Kock-monads-on-SMCC}, in view of \cite[Rem.~2.4]{Kock-bilinearity}.
    Finally, \cref{eqn:str-monoidal} follows from a routine, though somewhat tedious, calculation. \qed

\subsection{Proof of \texorpdfstring{\cref{W-comm-general-characterization}}{Prop. 35}}
\label{proof:W-comm-general-characterization}
     \cref{W-comm-general-characterization}(2) is a special case of \cref{W-comm-general-characterization}(1), so the latter implies the former.
    To show the reverse implication, assume \cref{W-comm-general-characterization}(2). Take any $n$-tuple $\vec X=(X_j)_{j\in n}$ of sets, and let $X=\sum_{j\in n}X_j$.
    Now consider the following diagram.
    \begin{equation}
    \label{eqn:W-comm-alternative}
    \begin{tikzpicture}[baseline=-\the\dimexpr\fontdimen22\textfont2\relax ]
     \node(01o) at (-3,1.5) {$\prod_{j\in n}TX_j$};
     \node(02o) at (6,1.5) {$T\bigl(\prod_{j\in n}X_j\bigr)$};
     \node(11o) at (-3,-1.5) {$\prod_{i\in m}TX_{\alpha(i)}$};
      \node(12o) at (6,-1.5) {$T\bigl(\prod_{i\in m}X_{\alpha(i)}\bigr)$};
      \node(01) at (0,0.5) {$(TX)^n$};
      \node(02) at (3,0.5) {$T(X^n)$};
      \node(11) at (0,-0.5) {$(TX)^m$};
      \node(12) at (3,-0.5) {$T(X^m)$};
      \draw [->] (01) to node[auto,labelsize] {$\comm_{X}^{(n)}$} (02);
      \draw [->] (01) to node[auto,swap,labelsize] {$(TX)^\alpha$} (11);
      \draw [->] (11) to node[auto,labelsize] {$\comm_{X}^{(m)}$} (12);
      \draw [->] (02) to node[auto,labelsize] {$T(X^\alpha)$} (12);
      \draw [->] (01o) to node[auto,labelsize] {$\comm_{\vec X}$} (02o);
      \draw [->] (11o) to node[auto,labelsize] {$\comm_{\vec X\alpha}$} (12o);
      \draw [->] (01o) to node[auto,swap,labelsize] {$\langle \pi_{\alpha(i)}\rangle_{i\in m}$} (11o);
      \draw [->] (02o) to node[auto,labelsize] {$T\langle \pi_{\alpha(i)}\rangle_{i\in m}$} (12o);
      \draw [<-] (01) to node[auto,swap, near start, labelsize] {$\prod_{j\in n}T\iota_j$} (01o);
      \draw [<-] (02) to node[auto,near start, labelsize] {$T\bigl(\prod_{j\in n}\iota_j\bigr)$} (02o);
      \draw [<-] (11) to node[auto,near start, labelsize] {$\prod_{i\in m}T\iota_{\alpha(i)}$} (11o);
      \draw [<-] (12) to node[auto,swap,near start, labelsize] {$T\bigl(\prod_{i\in m}\iota_{\alpha(i)}\bigr)$} (12o);
\end{tikzpicture}
    \end{equation}
    Here, $\iota_j\colon X_j\to X$ denotes the $j$-th coprojection. 
    In \eqref{eqn:W-comm-alternative}, the five small squares commute. 
    The outer big square in \eqref{eqn:W-comm-alternative} also commutes since the morphism $T\bigl(\prod_{i\in m}\iota_{\alpha(i)}\bigr)$ is a monomorphism (the functor part of a monad (on $\Set$) always preserves monomorphisms; see e.g.\ \cite{Zhen-Lin}).
    Therefore \cref{W-comm-general-characterization}(1) holds. \qed

\subsection{Proof of \texorpdfstring{\cref{W-comm-mnd-characterizations}}{Prop. 36}}
\label{proof:W-comm-mnd-characterizations}

    The ``if'' part of each assertion is trivial. For the ``only if'' parts, see \cite[Lems~4.23, 4.27, and 4.33]{Parlant-thesis} for proofs of \cref{W-comm-mnd-characterizations}(3), (2), and (5), respectively. The remaining cases can be proved similarly.\qed

\subsection{Proof of \texorpdfstring{\cref{thm:main-detailed}}{Thm. 42}}
\label{proof:thm:main-detailed}
The naturality of $\delta$ is clear.
To say that $\delta$ is a distributive law of the monad $\Mnd{\mnd O}$ over the \emph{endofunctor} $T$ is equivalent to saying that the following diagrams commute for all $X\in\Set$:
\begin{equation}\label{eqn:dist-law-eta-LanJO}
\begin{tikzpicture}[baseline=-\the\dimexpr\fontdimen22\textfont2\relax ]
      \node(01) at (1.5,0.5) {$TX$};
      \node(11) at (0,-0.5) {$(\Lan_JO)TX$};
      \node(12) at (3,-0.5) {$T(\Lan_JO)X$};
      \draw [->] (01) to node[auto,swap,labelsize] {$\eta^{\Mnd{\mnd O}}_{TX}$} (11);
      \draw [->] (11) to node[auto,labelsize] {$\delta_X$} (12);
      \draw [->] (01) to node[auto,labelsize] {$T\eta^{\Mnd{\mnd O}}_X$} (12);
\end{tikzpicture}
\end{equation}
\begin{equation}\label{eqn:dist-law-mu-LanJO}
\begin{tikzpicture}[baseline=-\the\dimexpr\fontdimen22\textfont2\relax ]
      \node(01) at (0,0.5) {$(\Lan_JO)(\Lan_JO)TX$};
      \node(02) at (4.5,0.5) {$(\Lan_JO)T(\Lan_JO)X$};
      \node(03) at (9,0.5) {$T(\Lan_JO)(\Lan_JO)X$};
      \node(11) at (0,-0.5) {$(\Lan_JO)TX$};
      \node(12) at (9,-0.5) {$T(\Lan_JO)X$};
      \draw [->] (01) to node[auto,labelsize] {$(\Lan_JO)\delta_X$} (02);
      \draw [->] (02) to node[auto,labelsize] {$\delta_{(\Lan_JO)X}$} (03);
      \draw [->] (01) to node[auto,swap,labelsize] {$\mu^{\Mnd{\mnd O}}_{TX}$} (11);
      \draw [->] (11) to node[auto,labelsize] {$\delta_X$} (12);
      \draw [->] (03) to node[auto,labelsize] {$T\mu^{\Mnd{\mnd O}}_X$} (12);
\end{tikzpicture}
\end{equation}
The commutativity of \cref{eqn:dist-law-eta-LanJO} can be seen from the following.
\begin{equation*}
\begin{tikzpicture}[baseline=-\the\dimexpr\fontdimen22\textfont2\relax ]
      \node(01) at (1.25,1.5) {$TX$};
      \node(11) at (-0.5,0.5) {$\vN 1\times TX$};
      \node(12) at (3,0.5) {$T(\vN 1\times X)$};
      \node(21) at (-0.5,-0.5) {$O_1\times TX$};
      \node(22) at (3,-0.5) {$T(O_1\times X)$};
      \node(31) at (-0.5,-1.5) {$(\Lan_JO)TX$};
      \node(32) at (3,-1.5) {$T(\Lan_JO)X$};
      \draw[->, rounded corners=10pt] (01)--(-2.5,1.5)--(-2.5,-1.5)--(31) ;
      \draw[->, rounded corners=10pt] (01)--(5,1.5)--(5,-1.5)--(32) ;
      \draw [->] (01) to node[auto,swap,labelsize,xshift=6pt,yshift=-1pt] {$\ell^{-1}_{TX}$} (11);
      \draw [->] (11) to node[auto,labelsize] {$\str_{\vN 1,X}$} (12);
      \draw [->] (01) to node[auto,labelsize,xshift=-7pt,yshift=-1.5pt] {$T\ell^{-1}_X$} (12);
      \draw [->] (11) to node[auto,swap,labelsize] {$\nameof{\ido}\times 1$} (21);
      \draw [->] (21) to node[auto,swap,labelsize] {$\kappa_1$} (31);
      \draw [->] (12) to node[auto,labelsize] {$T(\nameof{\ido}\times 1)$} (22);
      \draw [->] (21) to node[auto,labelsize] {$\str_{O_1,X}$} (22);
      \draw [->] (22) to node[auto,labelsize] {$T\kappa_1$} (32);
      \draw [->] (31) to node[auto,labelsize] {$\delta_X$} (32);
      \node[labelsize] at (1.25,1) {\eqref{eqn:right-strength-unit}};
      \node[labelsize] at (1.25,0) {(naturality of $\tau$)};
      \node[labelsize] at (1.25,-1) {(definition of $\delta$)};
      \node[labelsize] at (-1.2,1) {(def.\ of $\eta^{\Mnd{\mnd O}}$)};
      \node[labelsize] at (-3.1,0) {$\eta^{\Mnd{\mnd O}}_{TX}$};
      \node[labelsize] at (3.7,1) {(def.\ of $\eta^{\Mnd{\mnd O}}$)};
      \node[labelsize] at (5.8,0) {$T\eta^{\Mnd{\mnd O}}_{X}$};
\end{tikzpicture}
\end{equation*}
To show the commutativity of  \cref{eqn:dist-law-mu-LanJO},
take any natural number $n$ and $n$-tuple $\vec m=(m_i)_{i\in\vN n}$ of natural numbers. Let $m=\sum_{i\in\vN n}m_i$. Then
\begin{equation*}
\adjustbox{scale=0.69}{
    \begin{tikzpicture}[baseline=-\the\dimexpr\fontdimen22\textfont2\relax ]
      \node(01) at (0,2.5) {$O_n\times \prod \bigl(O_{m_i}\times (TX)^{m_i}\bigr)$};
      \node(02) at (4,3.5) {$O_n\times \prod\bigl(O_{m_i}\times T(X^{m_i})\bigr)$};
      \node(03) at (8,4.5) {$O_n\times \prod T(O_{m_i}\times X^{m_i})$};
      \node(04) at (12,3.5) {$O_n\times T\bigl(\prod(O_{m_i}\times X^{m_i})\bigr)$};
      \node(05) at (16,2.5) {$T\bigl(O_n\times \prod(O_{m_i}\times X^{m_i})\bigr)$};
      \node(c) at (4,2) {$O_n\times (\prod O_{m_i})\times \prod T(X^{m_i})$};
      \node(d) at (12,2) {$O_n\times T\bigl((\prod O_{m_i})\times X^m\bigr)$};
      \node(11) at (0,0) {$O_n\times (\prod O_{m_i})\times (TX)^{m}$};
      \node(12) at (8,0) {$O_n\times (\prod O_{m_i})\times T(X^m)$};
      \node(13) at (16,0) {$T\bigl(O_n\times (\prod O_{m_i}) \times X^m\bigr)$};
      \node(21) at (0,-1.5) {$O_m\times (TX)^{m}$};
      \node(22) at (8,-1.5) {$O_m\times T(X^{m})$};
      \node(23) at (16,-1.5) {$T(O_m\times X^{m})$};
      \draw [->] (01) to node[auto,labelsize,yshift=-3pt,xshift=-3pt] {$1\times \prod(1\times\comm^{(m_i)}_{X})$} (02);
      \draw [->] (02) to node[auto,labelsize,yshift=-3pt] {$1\times \prod \str_{O_{m_i},X^{m_i}}$} (03);
      \draw [->] (03) to node[auto,labelsize,yshift=-3pt] {$1\times \comm_{\vec Y}$} (04);
      \draw [->] (04) to node[auto,labelsize,yshift=-3pt] {$\str_{O_n,\prod (O_{m_i}\times X^{m_i})}$} (05);
      \draw [->] (11) to node[midway,fill=white,labelsize] {$\subst_{n,\vec m}\times 1$} (21);
      \draw [->] (12) to node[midway,fill=white,labelsize] {$\subst_{n,\vec m}\times 1$} (22);
      \draw [->] (13) to node[midway,fill=white,labelsize] {$T(\subst_{n,\vec m}\times 1)$} (23);
      \draw [->] (21) to node[auto,labelsize] {$1\times \comm_{X}^{(m)}$} (22);
      \draw [->] (22) to node[auto,labelsize] {$\str_{O_m,X^m}$} (23);
      \draw [->] (05) to node[auto,labelsize] {$\cong$} (13);
      \draw [->] (11) to node[auto,labelsize] {$1\times 1\times \comm^{(m)}_X$} (12);
      \draw [->] (12) to node[auto,labelsize] {$\str_{O_n\times \bigl(\prod O_{m_i}\bigr),X^m}$} (13);
      \draw [->] (02) to node[auto,swap,labelsize] {$\cong$} (c);
      \draw [->] (04) to node[auto,labelsize] {$\cong$} (d);
      \draw [->] (c) to node[midway,fill=white,labelsize] {$1\times 1\times \comm_{\vec Z}$} (12);
      \draw [->] (11) to node[midway,fill=white,labelsize] {$1\times 1\times \prod\comm_{X}^{(m_i)}$} (c);
      \draw [->] (12) to node[midway,fill=white,labelsize] {$1\times \str_{\prod O_{m_i},X^m}$} (d);
      \draw[->] (d) to node[midway,fill=white,labelsize] {$\str_{O_n,\bigl(\prod O_{m_i}\bigr)\times X^m}$} (13);
      \draw [->] (01) to node[auto,swap,labelsize] {$\cong$} (11);
      \node[labelsize] at (8,2.5) {\cref{eqn:str-monoidal}};
      \node[labelsize] at (12,0.9) {\cref{eqn:right-strength-tensor}};
      \node[labelsize] at (4,0.9) {\cref{eqn:phi-decomposition}};
      \node[labelsize] at (14.7,1.55) {(naturality of $\str$)};
      \node[labelsize] at (1.3,1.55) {(naturality of $\cong$)};
      \node[labelsize] at (4,-0.7) {(bifunctoriality of $\times$)};
      \node[labelsize] at (12,-0.7) {(naturality of $\str$)};
\end{tikzpicture}}
\end{equation*}
completes the proof, where 
$\vec Y=(O_{m_i}\times X^{m_i})_{i\in\vN n}$ and 
$\vec Z=(X^{m_i})_{i\in\vN n}$.
This completes the proof of the first paragraph of \cref{thm:main-detailed}.

Next, to say that $\delta$ is a distributive law of the monad $\Mnd{\mnd O}$ over the \emph{monad} $\mnd T$ is equivalent to saying that, in addition to \eqref{eqn:dist-law-eta-LanJO} and \eqref{eqn:dist-law-mu-LanJO}, the following diagrams commute.

\begin{equation}\label{eqn:dist-law-eta-T}
\begin{tikzpicture}[baseline=-\the\dimexpr\fontdimen22\textfont2\relax ]
      \node(01) at (1.5,0.5) {$(\Lan_JO)X$};
      \node(11) at (0,-0.5) {$(\Lan_JO)TX$};
      \node(12) at (3,-0.5) {$T(\Lan_JO)X$};
      \draw [->] (01) to node[auto,swap,labelsize] {$(\Lan_JO)\eta^{\mnd T}_{X}$} (11);
      \draw [->] (11) to node[auto,labelsize] {$\delta_X$} (12);
      \draw [->] (01) to node[auto,labelsize] {$\eta^{\mnd T}_{(\Lan_JO)X}$} (12);
\end{tikzpicture}
\end{equation}
\begin{equation}\label{eqn:dist-law-mu-T}
\begin{tikzpicture}[baseline=-\the\dimexpr\fontdimen22\textfont2\relax ]
      \node(01) at (0,0.5) {$(\Lan_JO)TTX$};
      \node(02) at (3,0.5) {$T(\Lan_JO)TX$};
      \node(03) at (6,0.5) {$TT(\Lan_JO)X$};
      \node(11) at (0,-0.5) {$(\Lan_JO)TX$};
      \node(12) at (6,-0.5) {$T(\Lan_JO)X$};
      \draw [->] (01) to node[auto,labelsize] {$\delta_{TX}$} (02);
      \draw [->] (02) to node[auto,labelsize] {$T\delta_{X}$} (03);
      \draw [->] (01) to node[auto,swap,labelsize] {$(\Lan_JO)\mu^{\mnd T}_{X}$} (11);
      \draw [->] (11) to node[auto,labelsize] {$\delta_X$} (12);
      \draw [->] (03) to node[auto,labelsize] {$\mu^{\mnd T}_{(\Lan_JO)X}$} (12);
\end{tikzpicture}
\end{equation}

It turns out that the commutativity of \cref{eqn:dist-law-eta-T} can be proved without any additional assumptions. 
To this end, observe that the family $\bigl(\kappa_n\colon O_n\times X^n\to (\Lan_JO)X\bigr)_{n\in\NN}$ is jointly epimorphic.
Hence it suffices to check that for each $n\in\NN$, the morphism $\kappa_n\colon O_n\times X^n\to (\Lan_JO)X$ equalizes the two composites in \cref{eqn:dist-law-eta-T}:
\begin{equation*}
\begin{tikzpicture}[baseline=-\the\dimexpr\fontdimen22\textfont2\relax ]
      \node(01) at (5,1) {$O_n\times X^n$};
      \node(11) at (2,-0.25) {$O_n\times (TX)^n$};
      \node(13) at (5,-0.25) {$O_n\times T(X^n)$};
      \node(12) at (8,-0.25) {$T(O_n\times X^n)$};
      \node(21) at (0,-1.2) {$(\Lan_JO)TX$};
      \node(22) at (10,-1.2) {$T(\Lan_JO) X$};
      \node(31) at (0,1) {$(\Lan_JO)X$};
      \node(32) at (10,1) {$(\Lan_JO)X$};
      \draw [->] (01) to node[auto,swap,labelsize] {$\kappa_n$} (31);
      \draw [->] (01) to node[auto,labelsize] {$\kappa_n$} (32);
      \draw [->] (01) to node[auto,labelsize,xshift=-2pt,yshift=1pt] {$1\times \eta^{\mnd T}_{X^n}$} (13);
      \draw [->] (01) to node[auto,swap,labelsize,xshift=10pt,yshift=-2pt] {$1\times (\eta^{\mnd T}_X)^n$} (11);
      \draw [->] (11) to node[auto,swap,labelsize] {$1\times \comm_{X}^{(n)}$} (13);
      \draw [->] (13) to node[auto,swap,labelsize] {$\str_{O_n,X^n}$} (12);
      \draw [->] (01) to node[auto,labelsize,xshift=-10pt,yshift=-2pt] {$\eta^{\mnd T}_{O_n\times X^n}$} (12);
      \draw [->] (11) to node[auto,labelsize,xshift=-1pt,yshift=4pt] {$\kappa_n$} (21);
      \draw [->] (12) to node[auto,swap,labelsize,xshift=3pt,yshift=5pt] {$T\kappa_n$} (22);
      \draw [->] (21) to node[auto,swap,labelsize] {$\delta_{X}$} (22);
      \draw [->] (32) to node[auto,labelsize] {$\eta^{\mnd T}_{(\Lan_JO)X}$} (22);
      \draw [->] (31) to node[auto,swap,labelsize] {$(\Lan_JO)\eta^{\mnd T}_X$} (21);
      \node[labelsize] at (6.3,0.2) {\eqref{eqn:right-strength-eta}};
      \node[labelsize] at (4,0.2) {\eqref{eqn:eta-monoidal}};
      \node[labelsize] at (5,-0.8) {(definition of $\delta$)};
      \node[labelsize] at (1.5,0.3) {(def.\ of $\Lan_JO$)};
      \node[labelsize] at (8.5,0.4) {(naturality of $\eta^{\mnd T}$)};
\end{tikzpicture}
\end{equation*}

Finally we deal with \cref{eqn:dist-law-mu-T}.
The family $\bigl(\kappa_n\colon O_n\times (TTX)^n\to (\Lan_JO)TTX\bigr)_{n\in\NN}$ is jointly epimorphic. Consider the following.
\begin{equation*}
\adjustbox{scale=0.75}{
    \begin{tikzpicture}[baseline=-\the\dimexpr\fontdimen22\textfont2\relax ]
      \node(01) at (0,1) {$O_n\times (TTX)^n$};
      \node(02) at (4,1) {$O_n\times T((TX)^n)$};
      \node(03) at (8,1) {$T(O_n\times (TX)^n)$};
      \node(04) at (12,1) {$T(O_n\times T(X^n))$};
      \node(05) at (16,1) {$TT(O_n\times X^n)$};
      \node(c) at (8,0) {$O_n\times TT(X^n)$};
      \node(11) at (0,-1) {$O_n\times (TX)^n$};
      \node(12) at (8,-1) {$O_n\times T(X^n)$};
      \node(13) at (16,-1) {$T(O_n\times X^n)$};
      \draw [->] (01) to node[auto,labelsize,yshift=1pt] {$1\times \comm^{(n)}_{TX}$} (02);
      \draw [->] (02) to node[auto,labelsize,yshift=1pt] {$\str_{O_n,(TX)^n}$} (03);
      \draw [->] (03) to node[auto,labelsize,yshift=1pt] {$T(1\times \comm^{(n)}_X)$} (04);
      \draw [->] (04) to node[auto,labelsize,yshift=1pt] {$T\str_{O_n,X^n}$} (05);
      \draw [->] (01) to node[auto,labelsize] {$1\times (\mu^{\mnd T}_X)^n$} (11);
      \draw [->] (05) to node[auto,swap,labelsize] {$\mu^{\mnd T}_{O_n\times X^n}$} (13);
      \draw [->] (11) to node[auto,labelsize] {$1\times \comm^{(n)}_X$} (12);
      \draw [->] (12) to node[auto,labelsize] {$\str_{O_n,X^n}$} (13);
      \draw [->] (02) to node[auto,swap,near start,labelsize,yshift=3pt] {$1\times T\comm^{(n)}_X$} (c);
      \draw [->] (c) to node[auto,swap,labelsize] {$\str_{O_n,T(X^n)}$} (04);
      \draw [->] (c) to node[auto,labelsize] {$1\times \mu^{\mnd T}_{X^n}$} (12);
      \node[labelsize] at (4,-0) {(monad axiom when $n=0$;};
      \node[labelsize] at (4,-0.3) {trivial when $n=1$; \cref{eqn:mu-monoidal} when $n>1$)};
      \node[labelsize] at (12,-0) {\cref{eqn:right-strength-mu}};
      \node[labelsize] at (8,0.6) {(naturality of $\str$)};
\end{tikzpicture}}
\end{equation*}
Note that the monoidality \cref{eqn:mu-monoidal} of $\mu$ is equivalent to the commutativity of $\mnd T$ (see \cref{commutative-monad-characterization}(3)).
Also note that when $O_n=\vN 0$ for all $n>1$, we do not need to invoke \cref{eqn:mu-monoidal}.\qed

\subsection{Proof of \texorpdfstring{\cref{invariance-wrt-cob}}{Prop. 45}}\label{invariance-wrt-cob-proof}
Write the inclusion functors as follows.
\begin{equation*}
        \begin{tikzpicture}[baseline=-\the\dimexpr\fontdimen22\textfont2\relax ]
      \node(01) at (0,0.5) {$\W$};
      \node(02) at (2,0.5) {$\W'$};
      \node(11) at (1,-0.5) {$\Set$};
      \draw [->] (01) to node[auto,labelsize] {$K$} (02);
      \draw [->] (01) to node[auto,swap,labelsize] {$J$} (11);
      \draw [->] (02) to node[auto,labelsize] {$J'$} (11);
\end{tikzpicture}
    \end{equation*}
    We denote the underlying functor $\Lan_KO$ of $\mnd O'$ by $O'\colon\W'\to \Set$. 
    The unit of the left Kan extension $\Lan_KO$ is a natural transformation $\sigma\colon O\to O'K$. For each $n\in\NN$, its $n$-th component is a function $\sigma_n\colon O_n\to O'_n$.
For each set $Y$ and $n\in\NN$, denote by $\kappa_{Y,n}\colon O_n\times Y^n\to (\Lan_JO)Y=SY$ and $\kappa'_{Y,n}\colon O'_n\times Y^n\to (\Lan_{J'}O')Y=SY$ the coprojections associated with the coends.
Observe that the following triangle commutes.
\begin{equation}\label{eqn:sigma-kappa}
        \begin{tikzpicture}[baseline=-\the\dimexpr\fontdimen22\textfont2\relax ]
      \node(01) at (0,0.5) {$O_n\times Y^n$};
      \node(02) at (3,0.5) {$O_n'\times Y^n$};
      \node(11) at (1.5,-0.5) {$SY$};
      \draw [->] (01) to node[auto,labelsize] {$\sigma_n\times 1$} (02);
      \draw [->] (01) to node[auto,swap,labelsize] {$\kappa_{Y,n}$} (11);
      \draw [->] (02) to node[auto,labelsize] {$\kappa_{Y,n}'$} (11);
\end{tikzpicture}
    \end{equation}
    Now consider the following diagram.
    \begin{equation*}
    \begin{tikzpicture}[baseline=-\the\dimexpr\fontdimen22\textfont2\relax ]
        \node(91) at (0,1) {$O_n\times (TX)^n$};
      \node(92) at (3,1) {$O_n\times T(X^n)$};
      \node(93) at (6,1) {$T(O_n\times X^n)$};
      \node(01) at (0,0) {$O'_n\times (TX)^n$};
      \node(02) at (3,0) {$O'_n\times T(X^n)$};
      \node(03) at (6,0) {$T(O'_n\times X^n)$};
      \node(11) at (0,-1) {$STX$};
      \node(12) at (6,-1) {$TSX$};
      \draw [->] (91) to node[auto,labelsize] {$1\times \comm^{(n)}_X$} (92);
      \draw [->] (92) to node[auto,labelsize] {$\str_{O_n,X^n}$} (93);
      \draw [->] (01) to node[auto,swap,labelsize] {$1\times \comm^{(n)}_X$} (02);
      \draw [->] (02) to node[auto,swap,labelsize] {$\str_{O'_n,X^n}$} (03);
      \draw [->] (91) to node[auto,swap,labelsize] {$\sigma_{n}\times 1$} (01);
      \draw [->] (92) to node[auto,swap,labelsize] {$\sigma_{n}\times 1$} (02);
      \draw [->] (93) to node[auto,labelsize] {$T(\sigma_{n}\times 1)$} (03);
      \draw [->] (01) to node[auto,swap,labelsize] {$\kappa_{TX,n}'$} (11);
      \draw [->,transform canvas={yshift=-3pt}] (11) to node[auto,swap,labelsize] {$\delta_X$} (12);
      \draw [->,transform canvas={yshift=3pt}] (11) to node[auto,labelsize] {$\delta'_X$} (12);
      \draw [->] (03) to node[auto,labelsize] {$T\kappa'_{X,n}$} (12);
\end{tikzpicture}
    \end{equation*}
    The morphism $\delta'_X$ is determined by the commutativity of the lower rectangle, whereas $\delta_X$ is determined by the commutativity of the whole rectangle, thanks to \cref{eqn:sigma-kappa}.
    Thus we conclude $\delta=\delta'$.\qed

\subsection{Proof of \texorpdfstring{\cref{Mnd-W-Opd-adjunction}}{Prop. 57}}
This is immediate from \cref{rmk:Lan-restriction-lax-monoidal}, but  
here we sketch a direct proof not relying on the substitution monoidal structure.
Take any $\W$-operad $\mnd O$ and any monad $\mnd S$. 
By the universal property of $\Lan_J$, 
there is a bijective correspondence between natural transformations $\theta\colon O\to SJ$ and $\bar \theta\colon \Lan_JO\to S$.
Thus it suffices to show that $\theta$ is a morphism of $\W$-operads (i.e., it commutes with $\ido$ and $\subst$) iff $\bar \theta$ is a morphism of monads (i.e., it commutes with $\eta$ and $\mu$), which is not hard.\qed
\end{document}